\documentclass[12pt,a4paper,oneside]{article}
\usepackage[utf8]{inputenc}
\usepackage{amssymb,amsmath,color,soul,bm}
\sethlcolor{red}
\usepackage{graphicx}
\usepackage{wrapfig}
\usepackage{amsthm}
\usepackage{multicol}
\usepackage{subcaption}
\usepackage{mathtools}
\usepackage{todonotes}
\usepackage{hyperref}
\usepackage{enumitem}
\usepackage{multirow}
\usepackage{siunitx}
\usepackage{lscape}
\usepackage{authblk}
\usepackage{todonotes}
\usepackage{lineno}
\usepackage{mathrsfs}
\usepackage{xspace}
\usepackage{microtype}

\newtheorem{lemma}{Lemma}
\newtheorem{theorem}{Theorem}

\newtheorem{corollary}{Corollary}
\newtheorem{Remark}{Remark}
\newtheorem{definition}{Definition}

\newtheorem*{Theorem-non}{Theorem}
\newtheorem*{Lemma-non}{Lemma}
\newtheorem*{Conjecture-non}{Conjecture}
\newtheorem{Hypothesis}{Hypothesis}

\DeclareMathOperator{\sech}{sech}

\def\vph{\varphi}
\def\eps{\varepsilon}
\def\rme{\mathrm{e}}
\def\rmi{\mathbf{i}}

\def\rmd{\,\mathrm{d}}
\def\R{\mathbb{R}}

\def\C{\mathbb{C}}

\def\SS{\mathbb{S}}
\def\calL{\mathcal{L}}
\def\calH{\mathcal{H}}

\def\calN{\mathcal{N}}

\def\calO{\mathcal{O}}
\def\sgn{\mathrm{sgn}}
\def\m{\boldsymbol{m}}
\def\n{\boldsymbol{n}}
\def\M{\mathcal{M}}
\def\1{\boldsymbol{e}_1}
\def\3{\boldsymbol{e}_3}

\def\h{\boldsymbol{h}}

\def\cc{c_{\rm cp}}

\def\smu{\sqrt{-\mu}}

\def\b{\boldsymbol{\beta}}
\def\Sess{\Sigma_{\textnormal{ess}}}

\def\Spt{\Sigma_{\textnormal{pt}}}
\def\Sabs{\Sigma_{\textnormal{abs}}}
\def\om{\overline{m}}

\def\tb{\widetilde{\boldsymbol{\beta}}}
\def\e{\boldsymbol{e}}
\def\th{\tilde{h}}
\def\heff{\boldsymbol{h}_\textnormal{eff}}
\def\bG{\boldsymbol{G}_\lambda}
\def\bl{\overline{\boldsymbol{\beta}}\,}

\def\Im{\mathrm{Im}}
\def\Re{\mathrm{Re}}

\newcommand*{\tran}{^{\mkern-1.5mu\mathsf{T}}}
\newcommand{\MATLAB}{\textsc{Matlab}\xspace}
\newcommand{\STABLAB}{\textsc{Stablab}\xspace}

\title{Domain wall motion in axially symmetric spintronic nanowires}
\author[1]{Jens D.\ M.\ Rademacher}
\author[1]{Lars Siemer\footnote{\small{Email address for correspondence: \texttt{lars.siemer@uni-bremen.de}}}}
\affil[1]{\small Department of Mathematics, University of Bremen, 28359 Bremen, Germany}
\date{\today}
\begin{document}
\maketitle
\begin{abstract}
This article is concerned with the dynamics of magnetic domain walls (DWs) in nanowires as solutions to the classical \emph{Landau-Lifschitz-Gilbert} equation augmented by a typically non-variational \emph{Slonczewski} term for spin-torque effects. Taking applied field and spin-polarization as the primary parameters, we study dynamic stability as well as selection mechanisms analytically and numerically in an axially symmetric setting. Concerning the stability of the DWs' asymptotic states, we distinguish the bistable (both stable) and the monostable (one unstable, one stable) parameter regime. In the bistable regime, we extend known stability results of an explicit family of precessing solutions and identify a relation of applied field and spin-polarization for standing DWs. We verify that this family is convectively unstable into the monostable regime, thus forming so-called pushed fronts, before turning absolutely unstable. In the monostable regime, we present explicit formulas for the so-called absolute spectrum of more general matrix operators. This allows us to relate translation and rotation symmetries to the position of the singularities of the pointwise Green's function. Thereby, we determine the linear selection mechanism for the asymptotic velocity and frequency of DWs and corroborate these by long-time numerical simulations. All these results include the axially symmetric \emph{Landau-Lifschitz-Gilbert} equation.
\end{abstract}

\section{Introduction}
We study the dynamics of magnetic \emph{domain walls} (DWs) in nanowires for a constant applied field and current aligned with the material anisotropy. The dynamics of a magnetization field $\m=(m_1,m_2,m_3)\tran\in\SS^2$ in case of an applied spin-polarized current is governed by the \emph{Landau-Lifschitz-Gilbert-Slonczewski} equation
\begin{equation*}\label{eq:LLGS}
\partial_t \m - \alpha\m\times\partial_t\m=-\m\times\heff+\m\times\left(\m\times \boldsymbol{J}\right). \tag{LLGS}
\end{equation*}
Here, $\alpha>0$ is the \emph{Gilbert} damping parameter and $\boldsymbol{J}$ the applied current. The magnetization $\m$ as well as the effective field $\heff$ are normalized by the spontaneous magnetization~\cite{berger1996emission, gilbert2004phenomenological, hubert2008magnetic, landautheory, slonczewski1996current, slonczewski2002currents}.

The effective field is defined as the negative variational derivative of the total magnetic free energy $\boldsymbol{E}(\m)$~\cite{abert2019micromagnetics, hubert2008magnetic}. We consider uniaxial anisotropy, with material parameter $\mu<0$, applied field $h\in\R$, both aligned along the $\3$-axis, and the simplified form
\[\heff\coloneqq -\frac{\delta \boldsymbol{E}}{\delta \m}=\partial_x^2 \m+h\3-\mu m_3\3.\]
In the considered axially symmetric configuration, the applied current is given by
\begin{equation}\label{eq:Current}
\boldsymbol{J}\coloneqq \frac{\beta}{1+\cc m_3}\3,
\end{equation}
where $\beta\ge0$ and $\cc\in(-1,1)$ describe the current density and the ratio of spin-polarization, respectively~\cite{bertotti2008spin, bertotti2005magnetization, mayergoyz2009nonlinear}. We provide more details of the modeling assumptions in \S~\ref{sec:Preliminaries}.

In contrast to the well-known \emph{Landau-Lifschitz-Gilbert}~\eqref{eq:LLG} equation, where $\boldsymbol{J}\equiv 0$ and which can be written as the gradient of the total magnetic free energy, the generalized~\eqref{eq:LLGS} equation is no longer variational. This makes parts of the analysis more challenging, but all results presented in this paper include the~\eqref{eq:LLG} equation by setting $\cc=0$.

The basic states of~\eqref{eq:LLGS} are the constant up- and down-magnetization states $\m=\pm\3$, which form spatially homogeneous equilibria of~\eqref{eq:LLGS}. We are interested in coherent structures connecting these states, i.e., magnetic \emph{domain walls} (DWs), and their dynamics. This fundamentally depends on the stability properties of $\pm\3$, and we refer to the situation in which both asymptotic states $\pm \3$ are dynamically stable as \emph{bistable}, while \emph{monostable} refers to the case in which only of $\pm\3$ is stable.

In the considered axially symmetric setting, coherent structures occur as relative equilibria with respect to translation as well as rotation and are solutions of the form
\begin{equation}\label{eq:CoherentStructureAnsatz}
\m(\xi,t)=\m_\ast(\xi)\rme^{\rmi \vph(\xi,t)},\quad \xi=x-st \quad\textnormal{and}\quad \vph(\xi,t)=\phi(\xi)+\Omega t,
\end{equation}
with \emph{speed} $s\in\R$ and \emph{frequency} $\Omega\in\R$. Here the complex exponential acts on $\m_\ast\in\SS^2$ by rotation around the $\3$-axis. DWs are solutions of the form~\eqref{eq:CoherentStructureAnsatz} with asymptotics $\m_\ast(\xi)\to\pm\3$ for $\xi\to\pm\infty$ or $\xi\to\mp\infty$.

The only known explicit DW solutions, denoted $\m_0$, occur at $\beta\cdot \cc=0$, where \eqref{eq:LLGS} reduces to \eqref{eq:LLG}. This family of solutions was discovered in~\cite{goussev2010domain} for applied field with arbitrary strength, cf.\ Figure~\ref{fig:DWsSphere} (a), and will be discussed in \S~\ref{sec:ExplicitFamily}.

\begin{figure}
	\centering
	\begin{subfigure}[b]{0.3\textwidth}
		\includegraphics[trim=0 0 0 150, clip, width=\textwidth]{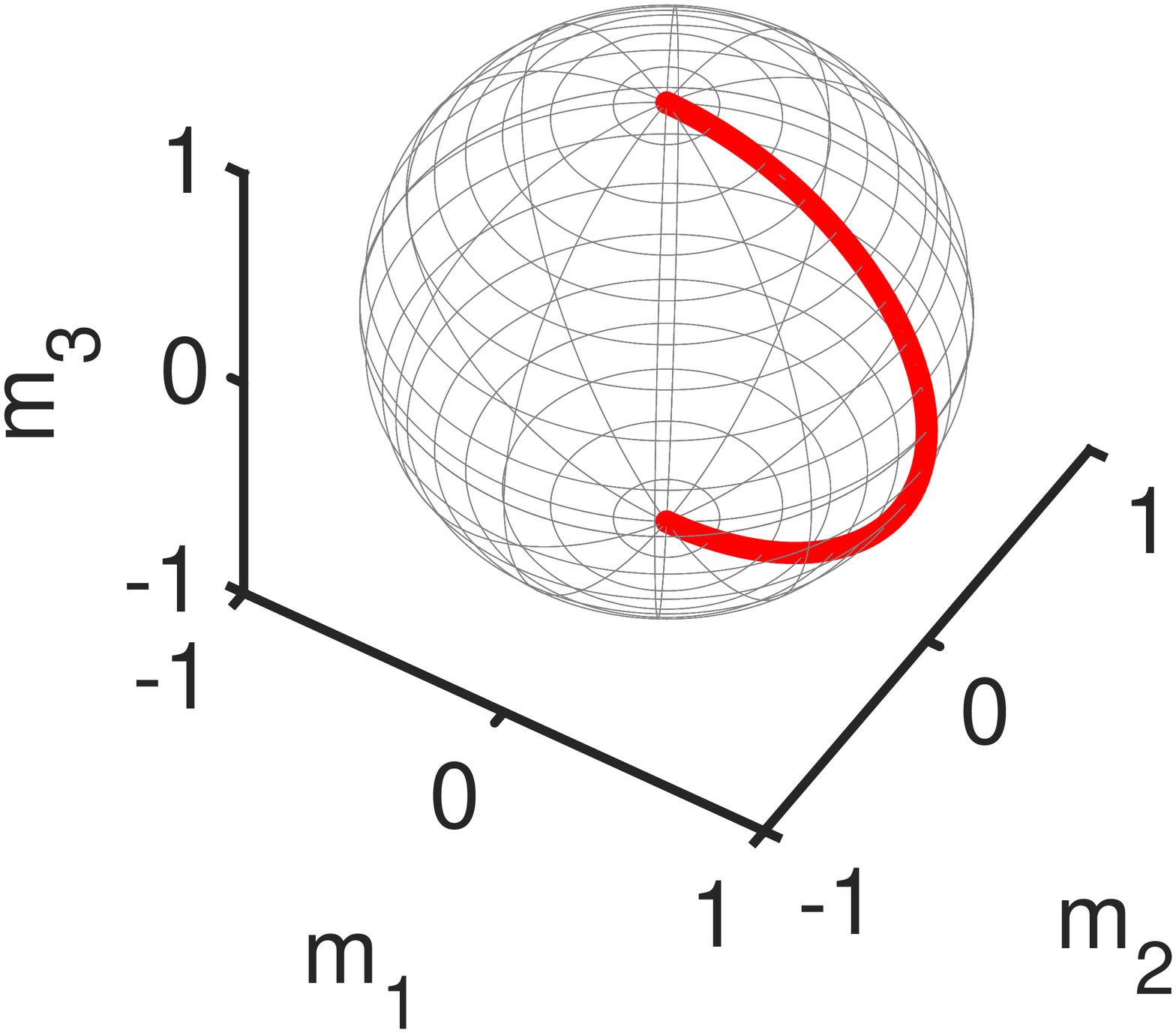}
		\caption*{(a)}
	\end{subfigure}
	\hfill
	\begin{subfigure}[b]{0.3\textwidth}
		\includegraphics[trim=0 0 0 150, clip, width=\textwidth]{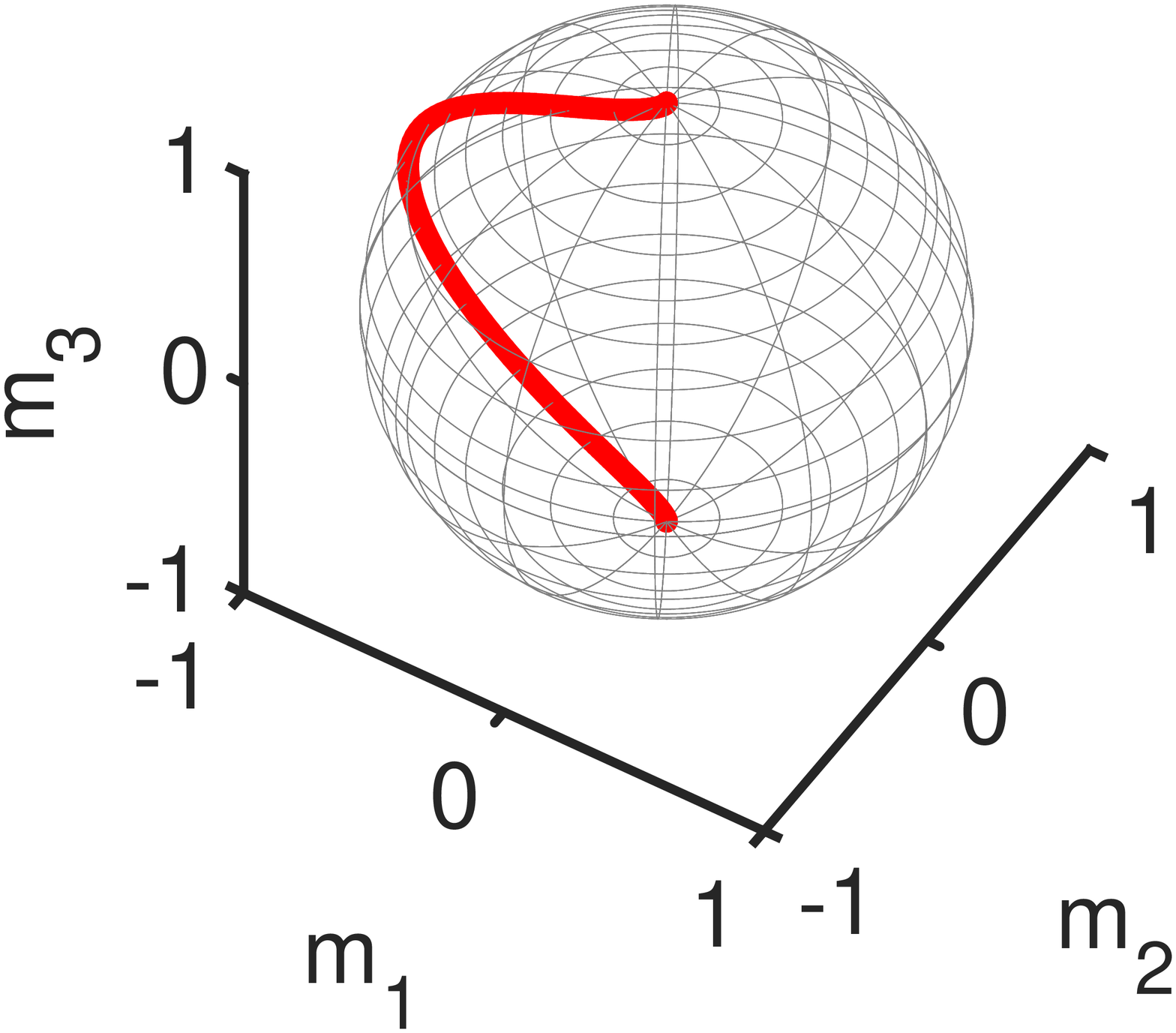}
		\caption*{(b)}
	\end{subfigure}
	\hfill
	\begin{subfigure}[b]{0.3\textwidth}
		\includegraphics[trim=0 0 0 150, clip, width=\textwidth]{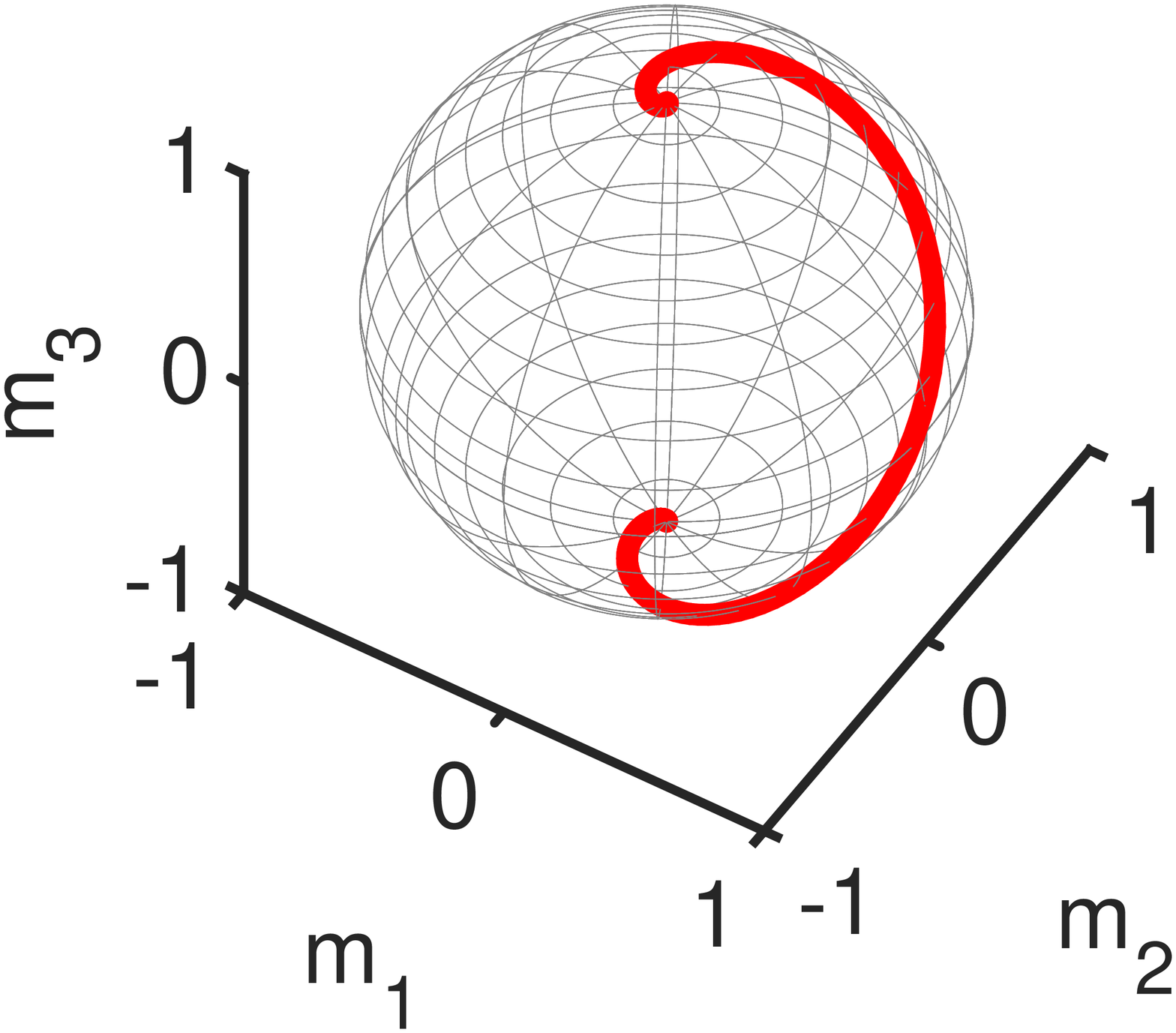}
		\caption*{(c)}
	\end{subfigure}
	\caption{Snapshots of DWs projected onto $\SS^2$. (a) explicit homogeneous DW profile given by~\eqref{eq:AnalyticHomogeneousSolution}, independent of $h$. (b)-(c) profiles of inhomogeneous DWs obtained by direct long-time simulations of the full PDE~\eqref{eq:LLGS} for $\alpha=1,\,\beta=0.75,\, \mu=-1$, and $\cc=0.5$ with $h=1$ (bistable regime) in (b) and $h=10$ (monostable regime) in (c). Numerical asymptotic speeds and frequencies are $s=0.081$ and $\Omega=0.923$ in (b) as well as $s=3.868$ and $\Omega=8.965$ in (c), where in (c) the linear predictions in the monostable regime are $s^\textnormal{lin}=3.873$ and $\Omega^\textnormal{lin}=9$.}
	\label{fig:DWsSphere}
\end{figure}

A first classification of more general DWs is based on the \emph{local wavenumber} $q\coloneqq\phi'(\xi)$, where a DW with $q\equiv 0$ is called \emph{homogeneous} and \emph{inhomogeneous} otherwise~\cite{siemer2020}; in particular $\m_0$ is homogeneous. It is known that DW solutions to~\eqref{eq:LLGS} for $\cc\neq 0$ cannot be homogeneous~\cite{Melcher2017}. An additional classification of DWs is based on the convergence behavior of the local wavenumber $q$ as $\xi\to\pm\infty$: we refer to DWs as \emph{flat}, if $|q(\xi)|$ has a limit in $\overline{\R}$ for $|\xi|\to\infty$, and \emph{non-flat} otherwise; homogeneous DWs are flat by definition.

Utilizing perturbations methods, it has been proven in~\cite{siemer2020} that non-flat DWs also exists in case $\cc=0$, and more importantly that flat as well as non-flat DWs exists for $|\cc|\ll 1$ sufficiently small, where the applied field strength can be arbitrary. Examples of inhomogeneous DWs for different applied fields and $\cc$ away from zero are presented in Figures~\ref{fig:DWsSphere} (b) and (c), which were obtained by direct long-time simulations of~\eqref{eq:LLGS} starting from a Heaviside step function initial condition in (b) and a compact perturbation of the unstable state $-\3$ in (c).

In this paper, we study stability properties and the selection of DWs analytically and numerically. The selection mainly refers to the monostable case for which -- in contrast to the bistable case -- a family of DWs coexists for given parameters. Our main interest lies in understanding selection mechanisms and determining conditions on spreading speeds and frequencies of DWs for an arbitrary strength of the applied field, both in the bistable as well as monostable parameter regime. 

Within the bistable regime, DWs can be orbitally asymptotically stable, i.e., perturbations relax back to a translate of the DW. This occurs for stable point spectrum (up to translation as well as rotation modes), which has been proven in~\cite{goussev2020dynamics} for field strength $h$ below a threshold. Here, we slightly improve this threshold analytically and show that stability is robust for $0<|\cc|\ll 1$. Moreover, we discuss the selection of spreading speed and its sign by heuristically extending the micromagnetic free energy~\eqref{eq:FreeEnergy}. Together with numerical simulations of the full PDE and numerical continuation in the $h$-$\cc$ parameter plane for standing DWs, i.e. $s=0$, we corroborate that: given any $\alpha>0, \beta\ge 0,\mu<0$, and $\cc\in(-1,1)$, the selected DW is approximately, though not exactly, stationary for the applied field 
\[h=\frac{\beta}{2\alpha \cc}\big( \ln(1+\cc)-\ln(1-\cc)\big).\]

In the monostable case, one expects that DWs propagate towards the unstable state~\cite{ducrot2019spreading, holzer2014anomalous, van2003front} and thus invade it, as studied in case $\cc=0$ for $\m_0$ in~\cite{goussev2020dynamics}.
In order to determine the asymptotic speed and frequency of this invasion process, we follow in detail the approach of linear spreading speeds and pointwise growth, which is based on refined linear stability properties of the unstable state~\cite{holzer2014criteria}. Briefly, one considers the stability problem in spaces with weighted norms, which moves the essential spectrum, thus changing stability properties. A barrier for this motion is the so-called \emph{absolute spectrum}~\cite{sandstede2000absolute} which is defined by the spatial roots of the dispersion relation, that relates the exponential rate in spatial direction with that in time. The absolute spectrum contains in particular those spectral values for which (certain) two spatial roots coincide that are at the same time singularities of the so-called \emph{pointwise Green's function}~\cite{beck2017stability, holzer2014criteria, howard1999pointwise, zumbrun1998pointwise}. This function is the resolvent kernel associated to the linearization around the unstable state and pointwise growth occurs if and only if a singularity has positive real part.

We will present an explicit formula for the absolute spectrum of certain (complex) linear matrix operators, which include the operators that arise from linearizing~\eqref{eq:LLGS} around the invaded state. This also includes the operators arising in the well-known complex Ginzburg-Landau equation (CGL) around the zero state, which is not surprising since~\eqref{eq:LLGS} can be written as a generalized (CGL) via stereographic projection. We will further utilize information about the explicit location of the absolute spectrum in order to locate singularities of the pointwise Green's function and to describe the transition from pointwise decay to pointwise growth of perturbations localized around the unstable state. This will enable us to determine the linear spreading speed and the associated linear spreading frequency, which are predictors for the dynamically selected DWs in~\eqref{eq:LLGS}. While the linear spreading speed is determined by the transition from pointwise decay to pointwise growth, the linear spreading frequency is not associated to a selection in the sense of stability. We substantiate our analytic results via direct numerical simulations, which corroborate that the explicit family is selected also inside the monostable region, but only until the linear spreading speed exceeds that of the explicit family, cf.\ Figure~\ref{fig:FreezingZeroCCPLight}. From the perspective of linear versus nonlinear speed selection, these DWs are so-called \emph{pushed} fronts in the monostable regime~\cite{paquette1994structural, stokes1976two, van2003front}, since they propagate faster than the linear prediction. Thus, their selection is necessarily driven by nonlinear effects. When the linear spreading speed exceeds the speed of the explicit family, the selected DWs lie outside the explicit family, propagate according to the linear prediction, and hence are so-called \emph{pulled} fronts~\cite{paquette1994structural, stokes1976two, van2003front}. For large applied fields, the explicit family has speeds that are again lower than that the linear spreading speed. However, the explicit family is not selected in this regime due to unstable point spectrum, which we verify numerically.

Our results match the prediction from~\cite{goussev2020dynamics} when $\cc=0$, i.e., for~\eqref{eq:LLG}, which is based on the so-called saddle-point approximation~\cite{ebert2000front, miller2006applied, van2003front}. Our more general approach is based on weighted $L^2$-spaces, the explicit expression of the absolute spectrum associated to the unstable states, and its relation to the pointwise Green's function. Moreover, utilizing the weighted spectrum paves the way to tackle problems beyond those considered in this paper, see e.g.~\cite{faye2020remnant}.

\begin{figure}
	\centering
	\begin{subfigure}[b]{0.49\textwidth}
		\includegraphics[trim=30 0 170 20, clip, width=\textwidth]{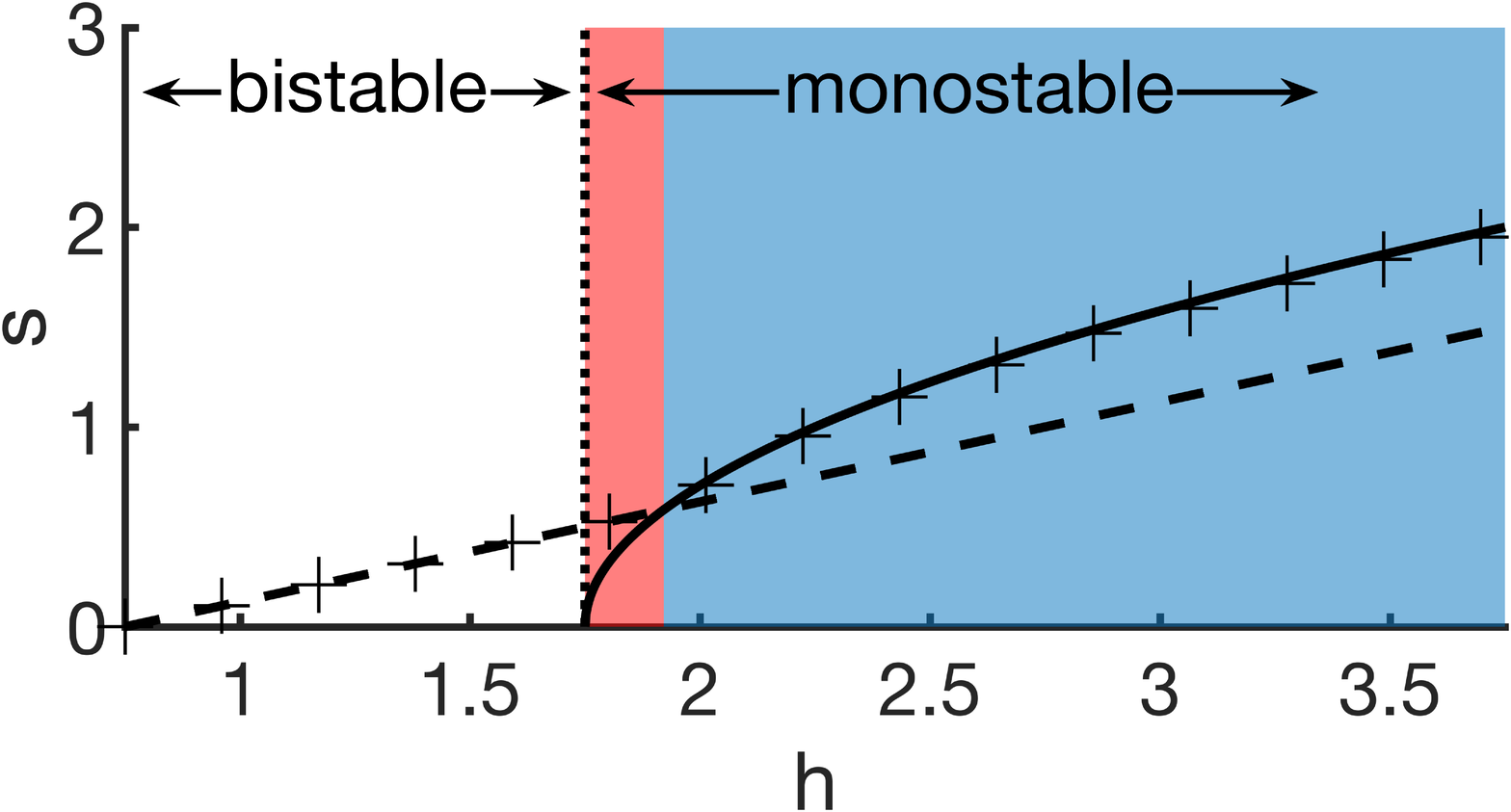}
		\caption*{(a)}
	\end{subfigure}
	\hfill
	\begin{subfigure}[b]{0.49\textwidth}
		\includegraphics[trim=30 0 170 20, clip, width=\textwidth]{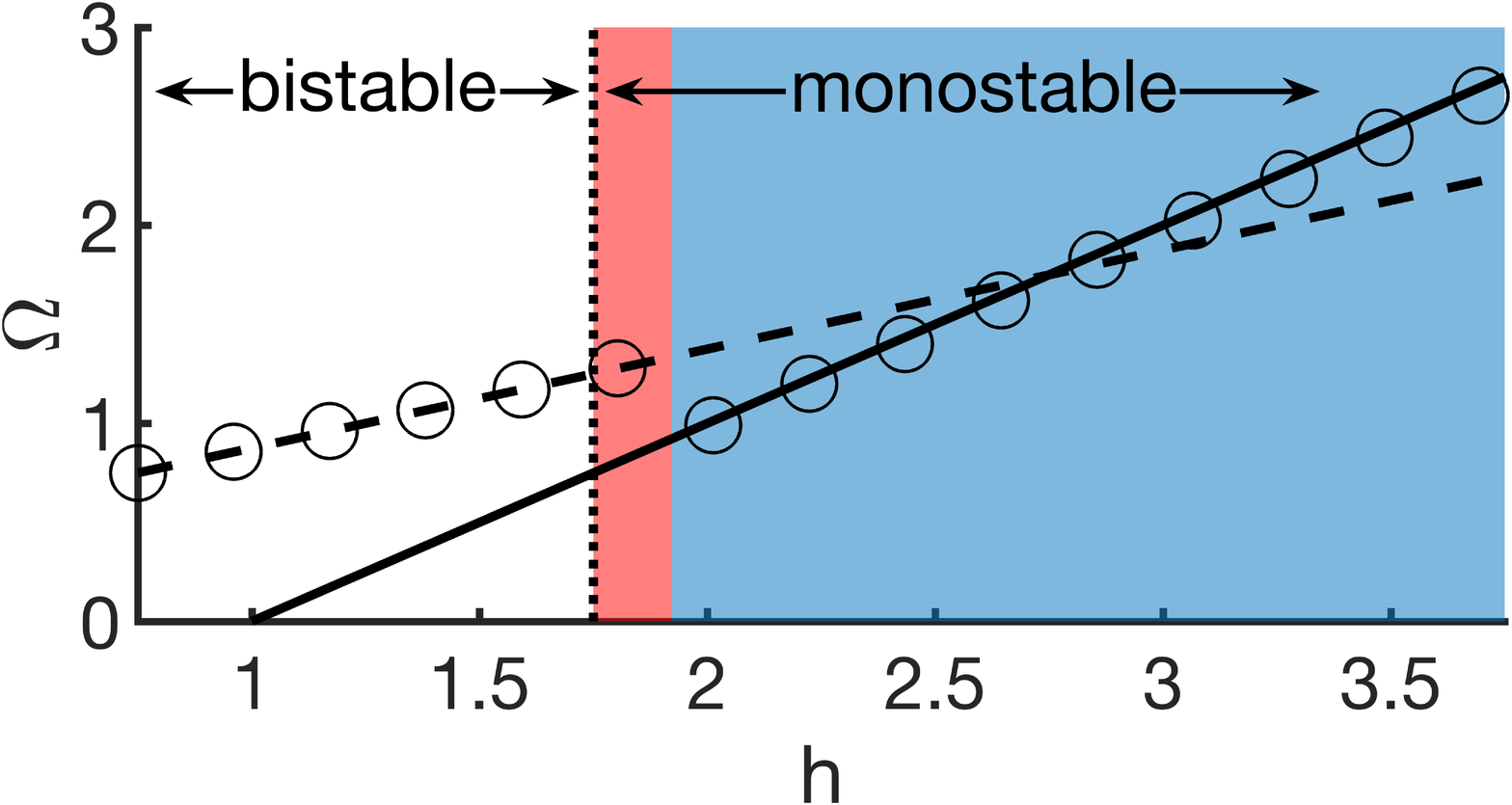}
		\caption*{(b)}
	\end{subfigure}
	\caption{Numerically observed spreading speed (crosses) in (a) and spreading frequency (circles) in (b) in the absence of spin torque compared with the explicit speed and frequency of $\m_0$ (dashed lines) and the linear predictions (solid lines). Parameters are $\alpha=1, \beta=0.75, \mu=-1, \cc=0$ and an applied field $h$ between $\beta/\alpha$ and $\beta/\alpha-3\mu$. Vertical dotted line at $h=1.75$ separates bistable and monostable regions. Pushed fronts are numerically observed for $1.75<h<1.92$ (red shaded region), whereas pulled fronts are observed for $1.92<h$ (blue shaded region).}
	\label{fig:FreezingZeroCCPLight}
\end{figure}

The remainder of the paper is organized as follows. We motivate the problem, describe first properties, introduce the linearized eigenvalue problem, and discuss the stability of the asymptotic states $\pm\3$ in \S~\ref{sec:Preliminaries}. Next, propagation dynamics of DWs in the \emph{bistable} regime are discussed in \S~\ref{sec:Bistable}, where we also extend the stability threshold for the explicit family analytically and present the heuristic relation between $h$ and $\cc$ for standing DWs. The \emph{monostable} regime is discussed in \S~\ref{sec:Monostable}, where we present the explicit formula of the absolute spectrum for the more general class of linear operators and relate this to the singularities of the pointwise Green's function. From this, we compute the \emph{linear spreading speed} and \emph{linear spreading frequency} and show that non-flat DWs are favored for large applied fields. In \S~\ref{sec:Numerics}, we numerically verify stable point spectrum of $\m_0$ into the monostable region. Additionally, we provide some details behind the numerical continuation and simulation results which are presented throughout the paper. Finally, we conclude with a discussion and outlook in \S~\ref{sec:DiscussionAndOutlook}.

\section{Preliminaries on homogeneous domain walls and asymptotic states}\label{sec:Preliminaries}

We start with a brief outline of the model background. The classical equation which describes the magnetization dynamics for a unit vector field $\m\in\SS^2$ was proposed by L.\ Landau and E.\ Lifschitz~\cite{landautheory} and later extended by T.\ Gilbert~\cite{gilbert2004phenomenological}. For normalized time, it reads
\begin{equation*}\label{eq:LLG}
\partial_t \m - \alpha \m \times \partial_t \m=-\m\times \heff , \tag{LLG}
\end{equation*}
where $\alpha>0$ is the Gilbert damping parameter. It describes a damped precession of the magnetization $\m=(m_1,m_2,m_3)$ around the effective field $\heff$, both normalized by the spontaneous magnetization $M_S$~\cite{hubert2008magnetic, lakshmanan2011fascinating}. Based on normalization, time is measured in units of $(\gamma M_S)^{-1}$ ($\gamma$ is the absolute value of the gyromagnetic ration) and space is measured in units of $2A/\mu_0 M_S^2$ ($A$ is the exchange stiffness and $\mu_0$ the permeability)~\cite{hubert2008magnetic}. The effective field itself is defined as the negative variational derivative of the micromagnetic free energy $\boldsymbol{E}(\m)$ and thus collects the exchange, the anisotropy, the Zeeman or external field, the demagnetization or stray field, and also the magnetoelastic energy~\cite{abert2019micromagnetics}. The leading order contributions of uniaxial anisotropy with parameter $\mu$ and easy axis $\3$, as well as local approximation of the demagnetization field, without effects caused by mechanical stress or deformation of the ferromagnetic object, lead to free energies of the form
\begin{equation}\label{eq:FreeEnergy}
\boldsymbol{E}(\m)=\frac{1}{2}\int (\nabla \m)^2 + \mu m_3^2\, \textnormal{d}x - \int \h \cdot \m \, \textnormal{d}x,
\end{equation}
where $\h\in\R^3$ is the applied field. We choose $\h$ to be aligned along the $\3$ axis, i.e., $\h=h\3$, with a time-independent and uniform field parameter $h\in\R$. Because the anisotropy energy is rotationally symmetric around $\3$ in our case, \emph{easy-axis} and \emph{easy-plane} anisotropy correspond to $\mu<0$ and $\mu>0$, respectively, according to the energetically preferred directions. Since we are interested in DWs in a nanowire geometry, we make the standing assumption $\mu<0$ throughout this paper.

In order to take effects caused by spin polarized currents into account, J.\ Slonczewski derived a model in~\cite{slonczewski1996current} in which the term $-\m\times \boldsymbol{J}$ was added to the effective field $\heff$. We consider the axially symmetric case, i.e., rotation symmetric around the $\3$-axis, which yields a current of the form~\eqref{eq:Current} and where the parameters $\beta\ge 0$ and $\cc\in(-1,1)$ describe the strength of spin-transfer and ratio of polarization in the considered axially symmetric situation~\cite{berger1996emission, bertotti2008spin, mayergoyz2009nonlinear, slonczewski1996current, slonczewski2002currents}.

In one space dimension, where $\m=\m(x,t), x\in\R$, we consider the~\eqref{eq:LLGS} equation, which in detail reads
\[\partial_t \m -\alpha \m\times \partial_t \m= -\m\times \left( \partial_x^2\m +(h-\mu m_3)\3-\frac{\beta}{1+\cc m_3}\m \times \3\right).\]

The~\eqref{eq:LLG} equation, with its Lyapunov (variational) structure, arises as a special case of~\eqref{eq:LLGS} if $\beta=0$ or $\cc=0$, where in the latter case this follows from shifting the applied field $h\to h-\beta/\alpha$~\cite{Melcher2017}. Therefore, while our focus is on the non-variational spintronic case $\beta\cdot \cc\neq 0$, all results presented in this paper that contain the case $\beta\cdot \cc=0$ carry over to~\eqref{eq:LLG}.

Additionally, \eqref{eq:LLGS} possesses a skew variational structure (see for example~\cite{Melcher2017}) as it can be written in the form
\[\partial_t\m-\alpha\m\times \partial_t\m=\m\times \frac{\delta \boldsymbol{E}}{\delta \m} + \m \times \left(\m \times \frac{\delta \boldsymbol{\Psi}}{\delta \m} \right),\]
with potential for $\boldsymbol{J}$ given by
\begin{equation}\label{eq:Potential}
\boldsymbol{\Psi}(\m)\coloneqq \frac{\beta}{\cc}\ln\left(1+\cc (\m\cdot\3)\right).
\end{equation}
It is also well-known that~\eqref{eq:LLGS} admits an equivalent form as an explicit evolution equation of quasilinear parabolic type,  which reads
\begin{equation}\label{eq:QuasilinearLLGS}
\partial_t \m=D(\m)\partial_x^2\m+ \frac{\alpha}{1+\alpha^2} | \partial_x \m|^2 \m+f(\m)
\end{equation}
with $| \partial_x \m|^2=(\partial_x m_1)^2 + (\partial_x m_2)^2 + (\partial_x m_3)^2$,
\begin{align*}
D(\m)&\coloneqq\frac{1}{1+\alpha^2}\begin{pmatrix} \alpha & m_3 & -m_2 \\ -m_3 & \alpha & m_1 \\ m_2 & -m_1 & \alpha \end{pmatrix},\\
f(\m)&\coloneqq \frac{\b(m_3)}{1+\alpha^2}\begin{pmatrix} m_1m_3-\alpha m_2 \\ m_2m_3+\alpha m_1 \\ m_3^2-1 \end{pmatrix} -\frac{h-\mu m_3}{1+\alpha^2} \begin{pmatrix}\ \alpha m_1m_3+m_2 \\ \alpha m_2m_3-m_1 \\\alpha \left( m_3^2-1 \right) \end{pmatrix},
\end{align*}
and where we introduced
\begin{align}\label{e:defbeta}
\b(m_3)=\b\coloneqq\frac{\beta}{1+\cc m_3} \quad \textnormal{as well as} \quad \b^\pm\coloneqq \b(\pm 1)=\frac{\beta}{1\pm\cc}.
\end{align}
For solutions of the form~\eqref{eq:CoherentStructureAnsatz}, which satisfy $\m(\pm\infty) =\pm \3$, or vice versa, we now consider the third component only. Since rotation acts on the first and second component only, we may write $m_3(x,t)=m_3(\xi)$ in a co-moving frame with $\xi=x-st$.

Following~\cite{komineas2019traveling}, for such solutions consider the integral
\[M\coloneqq \int_{-\infty}^\infty m_3(x-st) \rmd x \]
(in the Cauchy principle sense), which relates to the speed $s$ via
\[\frac{d}{dt} M = -s\int_{-\infty}^\infty m_3'\rmd x. \]
Since for a DW $\partial_t m_3=-s m'_3$ with $'=\textnormal{d}/\textnormal{d}\xi$, such a solution to~\eqref{eq:QuasilinearLLGS} fulfills
\[(1+\alpha^2)\partial_t m_3=\alpha m''_3 + m''_1 m_2 - m_1 m''_2 + \alpha|\m'|^2 m_3 - (\alpha(h-\mu m_3)-\b(m_3))(m_3^2-1),\]
and we obtain, via integration by parts, 
\begin{align*}
\Delta_3\cdot  s &= -\int_{-\infty}^\infty \partial_t m_3 \rmd x\\
&= \frac{\alpha}{1+\alpha^2}\int_{-\infty}^\infty |\m'|^2 m_3 \rmd x + \frac{1}{1+\alpha^2}\int_{-\infty}^\infty \b(m_3)(m_3^2-1) \rmd x \\
&\quad -\frac{\alpha h}{1+\alpha^2}\int_{-\infty}^\infty m_3^2-1 \rmd x- \frac{\alpha \mu}{1+\alpha^2}\int_{-\infty}^\infty m_3(m_3^2-1) \rmd x.
\end{align*}
This simplifies further for symmetric DWs. In particular for even $m_1, m_2$ and odd $m_3$ the first and last integral vanish. Moreover, for odd $m_3$ we obtain to leading order
\begin{equation}\label{eq:PropagationSignIntegral}
\Delta_3 \cdot s = \frac{\beta - \alpha h}{1+\alpha^2}\int_{-\infty}^\infty m_3^2-1 \rmd x + \calO(\cc^2).
\end{equation}
In particular, this defines a curve in the $h$-$\cc$-plane with quadratic tangency at $\cc=0$, cf.\ \S~\ref{sec:StandingDomainWalls}. Moreover, for the sign of propagation we obtain the leading order relation
\begin{equation}\label{eq:velsym}
\mathrm{sgn}(s) = \mathrm{sgn}\big(\Delta_3\big)\cdot  \mathrm{sgn}(\beta - \alpha h).
\end{equation}

\subsection{Coherent structure form and explicit DW}\label{sec:ExplicitFamily}
As we are interested in the dynamics of and near DWs, it is natural to apply the coherent structure ansatz~\eqref{eq:CoherentStructureAnsatz}. In a co-moving frame $\xi=x-st$ with wavespeed $s\in\R$ as well as co-rotating frame with (superimposed) frequency $\Omega\in\R$ around $\3$ system \eqref{eq:QuasilinearLLGS} turns into
\begin{equation}\label{eq:QuasilinearSpeedFrequencyLLGS}
\partial_t \m=D(\m)\m''+(1+\alpha^2)s\m'+ R\m+\frac{\alpha}{1+\alpha^2} | \m'|^2 \m+f(\m),
\end{equation}
where
\[R\coloneqq (1+\alpha^2)\Omega\begin{pmatrix}
0&1&0\\
-1&0&0\\
0&0&0
\end{pmatrix}.\]
Since $\m(\xi,t)$ lies on the unit sphere, it is convenient to rewrite~\eqref{eq:LLGS} in spherical coordinates via
\begin{equation}\label{eq:SphericalCoordinates}
\m=\sin(\theta)\left(\cos(\vph)\e_1 + \sin(\vph)\e_2 \right) + \cos(\theta) \e_3.
\end{equation}
Applying the coherent structure ansatz~\eqref{eq:CoherentStructureAnsatz} within these coordinates, \eqref{eq:LLGS} transforms into \begin{equation}\label{eq:SphericalCoherentStructurePDE}
\begin{split}
\begin{pmatrix}
\alpha & -1\\
1 & \alpha
\end{pmatrix}
\begin{pmatrix}\dot{\phi}\sin(\theta)\\-\dot{\theta}
\end{pmatrix}&=
\begin{pmatrix}
2\phi ' \theta '\cos(\theta) + s\theta'\\
-\theta '' -\alpha s \theta'
\end{pmatrix}\\
&\quad +\sin(\theta)\begin{pmatrix}
\phi '' +\alpha s \phi'-\alpha\Omega+ \beta/(1+\cc\cos(\theta))\\
(\phi ')^2 \cos(\theta) + s\phi'- \Omega + h-\mu\cos(\theta)
\end{pmatrix}
\end{split}
\end{equation}
with derivatives $\dot{\,}=\rmd/\!\rmd t$ and again $'=\textnormal{d}/\textnormal{d}\xi$. The existence of domain walls asymptotically connecting $\pm\3$ in~\eqref{eq:LLGS} now translates into the existence of heteroclinic orbits between $\theta=0$ and $\theta=\pi$ in the steady state equation to~\eqref{eq:SphericalCoherentStructurePDE}, which is given by the system of ODEs
\begin{equation}\label{eq:SphericalCoherentStructureODE}
\begin{split}
\theta''&=-\alpha s \theta' + \sin(\theta)\left(s\phi'+(\phi')^2\cos(\theta)+\alpha(h-\mu\cos(\theta))-\Omega\right)\\
\phi''\sin(\theta)&=-s\theta'-2\phi'\theta'\cos(\theta)-\sin(\theta)\left(\alpha s \phi'+\beta/(1+\cc \cos(\theta))-\alpha \Omega\right)
\end{split}
\end{equation}
In case $\cc=0$, the aforementioned explicit family of solutions $\m_0$ from~\cite{goussev2010domain, Melcher2017} in the original coordinates reads (recall $\mu<0$)
\begin{equation}\label{eq:AnalyticHomogeneousSolution}
    \m_0=\begin{pmatrix}
    \sech\left(\sigma\smu \xi\right)\\0\\ -\tanh\left(\sigma\smu \xi\right)
    \end{pmatrix}.
\end{equation}
In spherical coordinates, we denote this solution by $(\theta_0, \phi_0)$, where $\phi_0=\textnormal{constant}$ and $\theta_0=2\arctan(\exp(\sigma \smu \xi))$. The sign of propagation is, consistent with~\eqref{eq:velsym}, given by $\sigma=\sgn(s\cdot(\alpha h-\beta))$ in terms of speed
\begin{equation}\label{eq:HomogeneousSpeed}
s^2=-\frac{(\alpha h-\beta)^2}{\mu(1+\alpha^2)^2},
\end{equation}
and the frequency is
\begin{equation}\label{eq:HomogeneousFrequency}
\Omega=\frac{h+\alpha\beta}{1+\alpha^2}.
\end{equation}
Note that the shape as well as width of $\m_0$ depend on the anisotropy only and that these DWs are homogeneous since $\phi_0$ is constant. We further note that two explicit solutions of the form~\eqref{eq:AnalyticHomogeneousSolution} exist for any value of the applied field: for $s\neq0$ one left- and one right-moving, and for $s=0$ two stationary ones. 

\subsection{Eigenvalue Problem} In order to study spectral stability properties of solutions to~\eqref{eq:LLGS}, we now introduce the associated eigenvalue problem. The linearization around a (steady) solution $\om=(\om_1(\xi), \om_2(\xi), \om_3(\xi))\tran$ of~\eqref{eq:LLGS} can be obtained by substituting $\m(\xi,t)=\om(\xi)+\varepsilon \n(\xi)\rme^{\lambda t}+\calO(\varepsilon^2)$ for $|\varepsilon| \ll 1$ into~\eqref{eq:QuasilinearSpeedFrequencyLLGS}, which leads (at order $\varepsilon$) to the matrix eigenvalue problem
\begin{equation}\label{eq:LinearEigenvalueProblem}
    \calL(\om)\n=\lambda \n,
\end{equation}
where $\n$ is pointwise tangent to the sphere at $\om$ and
\[(1+\alpha^2)\calL(\om)=\]
\begin{align*}
    &\quad\begin{pmatrix}
    \alpha\partial_\xi^2 & \om_3\partial_\xi^2-\om_3'' & -\om_2\partial_\xi^2+\om_2''\\ -\om_3\partial_\xi^2+\om_3'' & \alpha\partial_\xi^2 & \om_1\partial_\xi^2-\om_1'' \\ \om_2\partial_\xi^2-\om_2'' & -\om_1\partial_\xi^2+\om_1'' & \alpha\partial_\xi^2
    \end{pmatrix}+ (1+\alpha^2)\begin{pmatrix}
    s\partial_\xi &\Omega&0\\
    -\Omega&s\partial_\xi&0\\
    0&0&s\partial_\xi
    \end{pmatrix}\\
    &+\left(\alpha|\overline{\m}'|^2+2\alpha\overline{\m}\left( \overline{\m}'\right)\tran\partial_\xi\right)\textnormal{Id}+\begin{pmatrix}
    \bl\om_3 & -\alpha\bl & \bl\om_1-\tb(\om_1\om_3-\alpha\om_2)
    \\
    \alpha\bl& \bl\om_3 & \bl\om_2-\tb(\om_2\om_3+\alpha\om_1)\\
    0&0 & 2\bl\om_3 + \tb(1-\om_3^2)
    \end{pmatrix}\\
    &+(h-\mu\om_3)\begin{pmatrix}
    -\alpha\om_3 & -1 & -\alpha\om_1\\
    1 & -\alpha\om_3 & -\alpha\om_2\\
    0 & 0 & -2\alpha\om_3
    \end{pmatrix}
    +\begin{pmatrix}
    0 & 0 & \mu(\alpha\om_1\om_3+\om_2)\\
    0 & 0 & \mu(\alpha\om_2\om_3-\om_1)\\
    0 & 0 & \alpha\mu(\om_3^2-1)
    \end{pmatrix}
\end{align*}
with $\om'=\partial_\xi \om$ and where we set $\bl\coloneqq \b(\om_3)$ as well as $\tb\coloneqq\beta\cc/(1+\cc \om_3)^2$.

The spectrum of $\calL(\om)$ for DWs, e.g.\ in $L^2$ (tangent to the sphere at $\om$), decomposes in terms of Fredholm properties into \emph{essential} and \emph{point} spectrum, denoted by $\Sess(\calL)$ and $\Spt(\calL)$, respectively. Moreover, the boundary of $\Sess(\calL)$ is contained in the (purely essential) spectra of $\pm\3$, cf.\ e.g.~\cite{Melcher2017}.

\subsection{Stability of Steady States}\label{sec:StabilitySteadyStates}
To start our analysis, we first review the stability of the steady states $\pm\3$ as homogeneous solutions to~\eqref{eq:LLGS} and its dependence on our main parameters $h$ and $\cc$. It is well-known that the constant up-magnetization state $+\3$ loses stability for an applied field below a threshold, whereas the constant down-magnetization state $-\3$ loses it above a certain threshold, both via a Hopf instability of the (purely) essential spectrum with frequency $\b^\pm/\alpha$, cf.\ e.g.~\cite{Melcher2017}. Moreover, due to the presence of a spin torque term, there exists an additional parameter regime, mainly characterized by the anisotropy, in which both states are (linearly) unstable (cf.\ Figure~\ref{fig:StabilityRegionsUp-AndDown-Magnetization}), cf.\ e.g.~\cite{Melcher2017, siemer2020}. More precisely, allowing only for perturbations tangential to the sphere at the poles, i.e., $\n=(n_1,n_2,0)\in \textnormal{T}_{\m} \SS^2$, we obtain from~\eqref{eq:LinearEigenvalueProblem} with $s=\Omega=0$ the eigenvalue problems $\widetilde{\calL}^+ \n = \lambda \n$ associated to $+\3$ as well as $\widetilde{\calL}^- \n = \lambda \n$ associated to $-\3$, where
\[\widetilde{\calL}^+\coloneqq\begin{pmatrix}
\frac{\alpha}{1+\alpha^2}\partial_\xi^2+\frac{\b^+-\alpha(h-\mu)}{1+\alpha^2} & \frac{1}{1+\alpha^2}\partial_\xi^2 - \frac{\alpha\b^++h-\mu}{1+\alpha^2} \\ -\frac{1}{1+\alpha^2}\partial_\xi^2+ \frac{\alpha\b^++h-\mu}{1+\alpha^2} & \frac{\alpha}{1+\alpha^2}\partial_\xi^2 +  \frac{\b^+-\alpha(h-\mu)}{1+\alpha^2}
\end{pmatrix}\]
and
\[\widetilde{\calL}^-\coloneqq\begin{pmatrix}
\frac{\alpha}{1+\alpha^2}\partial_\xi^2+\frac{\alpha(h+\mu)-\b^-}{1+\alpha^2} & -\frac{1}{1+\alpha^2}\partial_\xi^2 - \frac{\alpha\b^-+h+\mu}{1+\alpha^2} \\ \frac{1}{1+\alpha^2}\partial_\xi^2+ \frac{\alpha\b^-+h+\mu}{1+\alpha^2} & \frac{\alpha}{1+\alpha^2}\partial_\xi^2 +  \frac{\alpha(h+\mu)-\b^-}{1+\alpha^2}
\end{pmatrix}.\]
The spectra of the (closed) linear matrix operators $\widetilde{\calL}^\pm : H^2(\R)\times H^2(\R)\subset L^2(\R)\times L^2(\R) \rightarrow L^2(\R)\times L^2(\R)$ determine the stability of the states $\pm\3$ for these function spaces. The associated dispersion relations read
\begin{equation}\label{eq:DispersionRelationHomogeneousStates}
\begin{split}
\widetilde{d}^\pm(\lambda, \rmi k)&=\left(-\frac{\alpha}{1+\alpha^2}k^2+\frac{\b^\pm\tau-\alpha(h-\mu\tau)\tau}{1+\alpha^2}-\lambda\right)^2\\
&\quad+\left(\frac{\tau}{1+\alpha^2}k^2-\frac{\alpha\b^\pm +h-\mu\tau}{1+\alpha^2} \right)^2,
\end{split}
\end{equation}
where $\tau=1$ for $+\3$ and $\tau=-1$ for $-\3$. Since each of $\widetilde{\calL}^\pm$ has constant coefficients, $\Spt(\widetilde{\calL}^\pm)=\emptyset$, and the essential spectra parameterized by $k$ read
\begin{equation}\label{eq:EssentialSpectrumUp}
\Sess(\widetilde{\calL}^+)=\left\lbrace \frac{\b^+ -\alpha(k^2+ h-\mu)}{1+\alpha^2}\pm\rmi\frac{\alpha\b^+ +k^2+h-\mu}{1+\alpha^2}, k\in\R \right\rbrace
\end{equation}
as well as
\begin{equation}\label{eq:EssentialSpectrumDown}
\Sess(\widetilde{\calL}^-)=\left\lbrace \frac{\alpha( h+\mu-k^2)-\b^-}{1+\alpha^2}\pm\rmi\frac{\alpha\b^- +h+\mu-k^2}{1+\alpha^2}, k\in\R \right\rbrace.
\end{equation}
In particular, $k=0$ gives the most unstable points of $\Sess(\widetilde{\calL}^\pm)$, so that the homogeneous state $\m=+\3$ is (strictly) stable in case $h>\b^+/\alpha+\mu$ and $\m=-\3$ in case $h<\b^-/\alpha-\mu$. Note that the stability regions in the $h$-$\cc$-plane crucially depend on the sign of the anisotropy parameter. Although we are only interested in $\mu<0$, for clarity we include $\mu>0$, and hence define the boundaries $\Gamma^+\coloneqq(\b^+/\alpha)/(h-\mu)-1$ and $\Gamma^-\coloneqq 1-(\b^-/\alpha)/(h+\mu)$. For each fixed $\mu$, this results in the stability regions in the $h$-$\cc$-plane labelled in terms of the stability of $\pm\3$ as the aforementioned \emph{bistable} and \emph{monostable}, as well as \emph{biunstable}, where both are unstable. The monostable region may be further divided into two regions: $+\3$ stable and $-\3$ unstable, or vice versa. All regions are separated by the curves $\Gamma^\pm$ (cf.~Figure~\ref{fig:StabilityRegionsUp-AndDown-Magnetization} and see~\cite{siemer2020} for more details). 

\begin{figure}[h]
    \centering
    \begin{subfigure}[b]{0.4\textwidth}
    \includegraphics[trim= 80 50 50 80, clip, width=\textwidth]{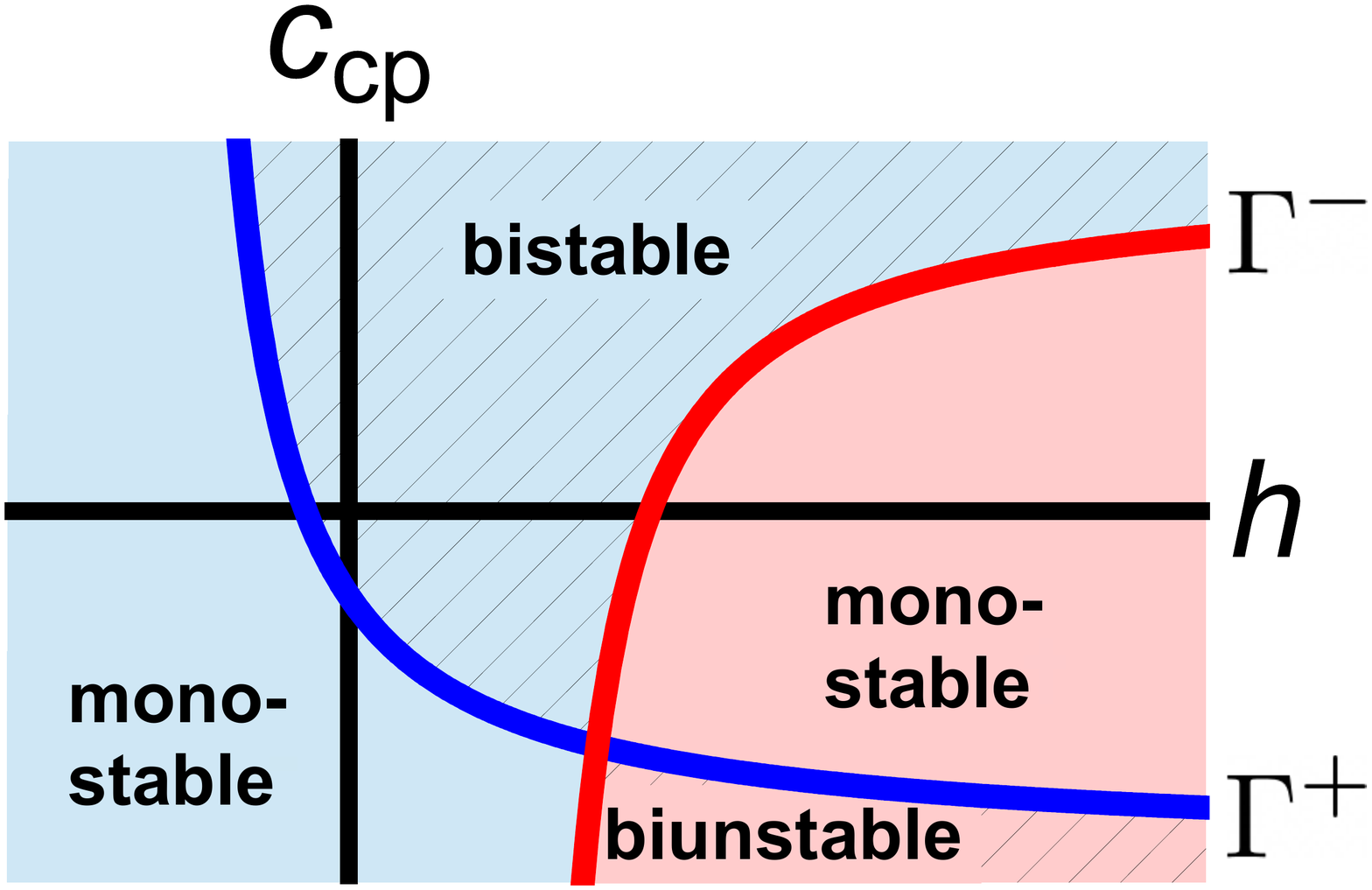}
    \caption*{\small $\mu<0$}
    \end{subfigure}
    \hspace*{15mm}
    \begin{subfigure}[b]{0.4\textwidth}
    \includegraphics[trim= 80 50 50 80, clip, width=\textwidth]{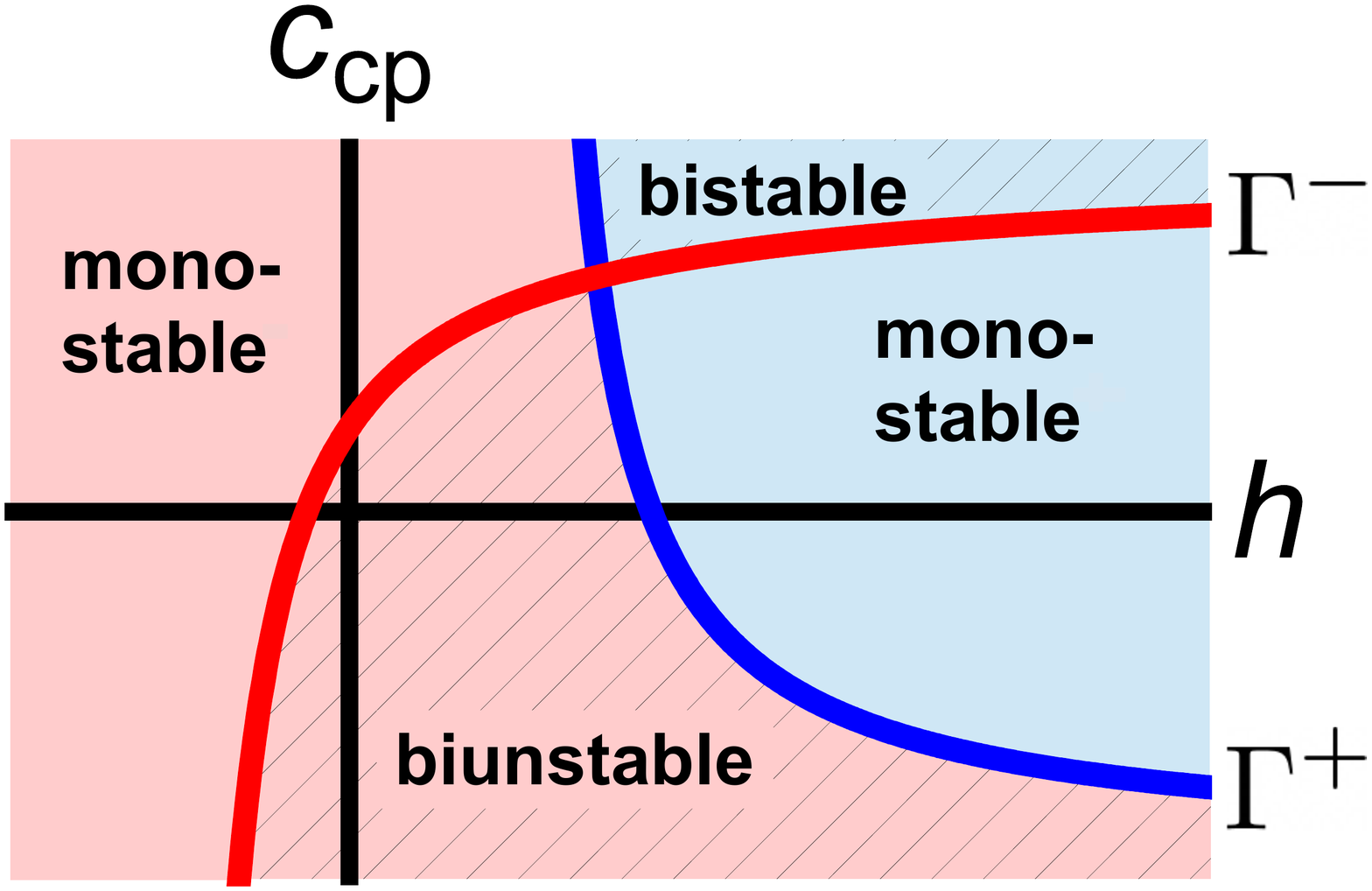}
    \caption*{\small $\mu>0$}
    \end{subfigure}
    \caption{Stability diagram of homogeneous states $\pm\3$ in $h$ and $\cc$ for $\alpha=1, \beta=3/4$, and $\mu=-1$ (left) as well as $\mu=1$ (right). State $+\3$ unstable to the left and stable to the right of $\Gamma^+$, $-\3$ stable to the left and unstable to the right of $\Gamma^-$.}
    \label{fig:StabilityRegionsUp-AndDown-Magnetization}
\end{figure}

\begin{Remark}
The terms \emph{bistable} and \emph{monostable} have been coined in the study of scalar parabolic PDEs, for which the theory of fronts, their stability, invasion and spreading has been originally developed. In the bistable case two homogeneous states are separated by an unstable state in the corresponding ODE for spatially constant states, and prototypical examples are the Nagumo equation in neuroscience~\cite{nagumo1962active} and the Allen-Cahn equation in materials science~\cite{allen1979microscopic}. The generic and extensively studied monostable paradigm for scalar parabolic PDEs is the Fisher-Kolmogorov-Petrovskii-Piscounov (FKPP) equation~\cite{aronson1975nonlinear, faye2019asymptotic, fisher1937wave, kolmogorov1937bulletin, tikhomirov1991study}.
\end{Remark}

In order to capture the monostable parameter regime in case $\mu<0$, we define:
\begin{equation}\label{eq:MonostableSetPlus}
\M^+\coloneqq\lbrace h\in\R,\, \cc\in(-1,1): \b^-/\alpha-\mu<h \textnormal{ and } \b^+/\alpha+\mu<h \rbrace
\end{equation}
\begin{equation}\label{eq:MonostableSetMinus}
\M^-\coloneqq\lbrace h\in\R,\, \cc\in(-1,1): h<\b^-/\alpha-\mu \textnormal{ and } h<\b^+/\alpha+\mu \rbrace
\end{equation}
Parameters from $\M^+$ lie to the `right' of $\Gamma^\pm$ in the $(h,\,\cc)$-plane, whereas parameters which lie to the `left' of $\Gamma^\pm$ correspond to $\M^-$.
\begin{Remark}\label{rem:AlphaDependence}
The bistable parameter region increases for $\alpha \to 0$, whereas the monostable region with stable $\3$ moves to the left and that with $-\3$ stable to the right, respectively. On the other hand, the bistable region shrinks for $\alpha \to +\infty$, whereas the monostable regions as well as the biunstable region increase. This is in accordance with the underlying model and the description of damped precession around the applied field aligned with the $\3$-axis.
\end{Remark}

\section{Domain wall motion: bistable case}\label{sec:Bistable}
In case both asymptotic states are stable, the existence of a DW \emph{and} the selection of the speed $s$ and frequency $\Omega$ is equivalent to the construction of a heteroclinic orbit in the underlying coherent structure ODE~\eqref{eq:SphericalCoherentStructureODE}. It was shown in~\cite{siemer2020} that such a heteroclinic connection is generically `codimension-2' in the sense that both $s$ and $\Omega$ need to be selected in order to close the connection. It was also proven there that these DWs can be continued to inhomogeneous DWs in the parameter $\cc$ and it was further shown numerically that this holds for the entire interval $\cc\in(-1,1)$, see also Figure~\ref{fig:DWsSphere} (b)-(c). As mentioned above, this means that -- in contrast to the monostable case -- in general the selection of the speed and frequency are determined by the full nonlinear problem for the `interior' profile of a DW and not by the asymptotic states alone.

To start, we recall that~\eqref{eq:LLGS} reduces to~\eqref{eq:LLG} in case $\beta=0$ which is equivalent to $\cc=0$ with the additional shift in the constant applied field $h\to h-\beta/\alpha$~\cite{Melcher2017, siemer2020}. Let us now denote the shifted applied field by $\th$ and set $\cc=0$. Concerning the explicit family of homogeneous DWs~\eqref{eq:AnalyticHomogeneousSolution}, denoted by $\m_0$, an upper bound on the applied field for linear stability was obtained in~\cite{gou2011stability}, where it was already mentioned that this is not optimal. To simplify notations, we rewrite the free energy~\eqref{eq:FreeEnergy} in the spherical coordinates~\eqref{eq:SphericalCoordinates} and without applied field as
\begin{equation}\label{eq:FreeEnergySpherical}
\boldsymbol{E}(\theta,\phi)=\frac{1}{2}\int \left(\left(\partial_x\theta\right)^2+\left(\partial_x\phi\right)^2\sin^2(\theta)+\mu\cos^2(\theta)\right)\, \textnormal{d}x.
\end{equation}

The idea of the proof presented in~\cite{gou2011stability} is not study the $\calO(\eps)$ correction, i.e., the spectrum of the linear operator arising from linearizing around $\m_0$, but rather show that to leading order the difference in the energies of the perturbed profile $\m_\eps=\m_0 + \eps \m_1 + \calO(\eps^2)$ and the explicit DW $\m_0$, denoted by $\Delta \boldsymbol{E}_\eps$, decays to zero asymptotically in time.

One difference in the setups is that~\cite{gou2011stability} considers a different direction of the easy-axis anisotropy as well as applied field, which (for a nanowire geometry) requires an anisotropy parameter larger than zero that is normalized to one. However, our choice of applied field and anisotropy axis, and the requirement $\mu<0$ for a nanowire geometry, does not affect the computations compared with~\cite{gou2011stability}, since the temporal evolution of the energy $\!\rmd/\!\rmd t \, \boldsymbol{E}(\m)$ and thus also $\!\rmd/\!\rmd t \, \Delta \boldsymbol{E}_\eps$ does not change. We obtain the following result.

\begin{lemma}\label{lem:LinearStabilityArctan}
For any $\alpha>0,\,\beta\ge 0$, and $\mu<0$, the homogeneous DW given by~\eqref{eq:AnalyticHomogeneousSolution} for $\cc=0$ is linearly stable for $|h|<\beta/\alpha-9\mu/(9+2\sqrt{3})$. Moreover, this also holds for its unique continuation along $\cc$ for any sufficiently small $|\cc|$.
\end{lemma}

\begin{proof}We linearize around the explicit DW given by~\eqref{eq:AnalyticHomogeneousSolution} for $\sigma=-1$ only, since the result also holds for $\sigma=+1$ due to symmetry. We will obtain the bound on the applied field by improving the estimate in~\cite[Equation (38)]{gou2011stability}, therefore we use the notation as in~\cite{gou2011stability}, and thus set $\mu=-1$ for the moment.
Again, we denote by $\theta_0$ the function of the polar angle of $\m_0$ in spherical coordinates, see~\eqref{eq:AnalyticHomogeneousSolution}. Our simple observation is that $|\sin^2(\theta)\cos(\theta)|=|\cos(\theta)-\cos^3(\theta)|\le 2/(3\sqrt{3})$
for any $\theta$. This means the estimates in~\cite[Equation (38)]{gou2011stability} can be improved to
\[|F|\le \frac{2}{\eps^2}\left(1+\frac{2}{3\sqrt{3}}\right)\Delta \boldsymbol{E}_\eps,\]
where
\[F\coloneqq\int \cos(\theta_0) f_0 + \cos(\theta_0)\sin^2(\theta_0)\theta_1^2 \, \textnormal{d}x,\]
$f_0\coloneqq (\theta_1')^2 + \cos(2\theta_0)\theta_1^2 + \sin^2(\theta_0)(\phi_1')^2$, and where $\phi_1$ as well as $\theta_1$ are the angles corresponding to $\m_1$. Therefore, we obtain analogous to~\cite[Equation (39)]{gou2011stability} the slightly improved estimate
\begin{equation}\label{eq:EnergyEstimate}
\frac{\rmd}{\rmd t}\Delta \boldsymbol{E}_\eps \le -2\alpha\left(1-\left(1+\frac{2}{3\sqrt{3}}\right)|\th| \right)\Delta \boldsymbol{E}_\eps .
\end{equation}
Although the estimate on $\sin^2(\theta)\cos(\theta)$ is indeed sharp, the one in~\eqref{eq:EnergyEstimate} is certainly not, as this includes several (non-sharp) estimates (see~\cite{gou2011stability} for detailed computations).

For general $\mu<0$, we obtain from~\eqref{eq:EnergyEstimate} that linear stability holds now for $|\th|< -9\mu /(9+2\sqrt{3})$ and by applying the backward shift in the applied field, it holds for $|h|<\beta/\alpha-9\mu /(9+2\sqrt{3})$.

Therefore, in this regime, the linearization in $\m_0$ admits a spectral gap to the simple zero eigenvalues from translation as well as rotation. It directly follows from compact perturbation theory, see e.g.~\cite{kapitula2013spectral}, that the continuation of $\m_0$ to DWs for $0<|\cc|\ll 1$ established in~\cite{siemer2020} is also spectrally stable in this sense, as long as $|\cc|$ is sufficiently small.
\end{proof}
In direct comparison, the stability result presented in~\cite{gou2011stability} holds for an applied field less than $0.5$ with $\mu=-1$, while the improved bound is $9/(9+2\sqrt{3})\approx 0.722$. We emphasize that~\cite{gou2011stability} have studied nonlinear stability of $\m_0$ numerically via time-evolution of $\Delta \boldsymbol{E}_\eps$ and have obtained that $\m_0$ is stable for an applied field $|h|<\beta/\alpha-\mu$ in our notation. This threshold is exactly the transition point to the monostable parameter regime $\M^\pm$ discussed in \S~\ref{sec:StabilitySteadyStates}, where one of the two stable asymptotic states $\pm\3$ becomes unstable.

\begin{Remark}\label{rem:HomogeneousDW}
In~\cite{gou2011stability}, also stability results for the case of a small transverse applied field and also for a small hard-axis anisotropy have been obtained, based on their stability bound. Therefore, one could potentially improve these results utilizing Lemma~\ref{lem:LinearStabilityArctan}.
\end{Remark}

\subsection{Standing Domain Walls}\label{sec:StandingDomainWalls}
We turn to the asymptotic propagation behavior in the bistable parameter regime. In particular, we introduce a heuristic approach for standing DWs to describe the higher order terms in~\eqref{eq:PropagationSignIntegral}. This approach utilizes~\eqref{eq:Potential} and although the result is not exact, it provides a surprisingly good approximation even for large $|\cc|<1$.

Bistability for a free energy corresponds to (at least) two local minima and one expects front propagation towards the state with lower energy. In analogy, we consider the free energy~\eqref{eq:FreeEnergy} in case $\cc=0$ and apply again the shift $h\mapsto h-\beta/\alpha$. Upon neglecting the Dirichlet energy term, the shifted free energy depends on $m_3$ only and possesses the potential
\[V(m_3)\coloneqq-\left(h-\beta/\alpha\right)m_3+\frac{1}{2}\mu m_3^2.\]
From this, we compute that $V(+1)>V(-1)$ in case $h<\beta/\alpha$ and $V(+1)<V(-1)$ if $\beta/\alpha<h$. Therefore, any DW connecting $+\3$ $[-\3]$ on the left with $-\3$ $[+\3]$ on the right will invade the state $-\3$ $[+\3]$ for $\beta/\alpha<h<\beta/\alpha-\mu$ $[\textnormal{for }\beta/\alpha+\mu<h<\beta/\alpha]$ and will invade the state $+\3$ $[-\3]$ for $\beta/\alpha+\mu<h<\beta/\alpha$ $[\textnormal{for }\beta/\alpha<h<\beta/\alpha-\mu]$. Note that these directions of propagation are in complete agreement with~\eqref{eq:velsym} and the definition of $\sigma$ as well as the sign of $s$ in~\eqref{eq:HomogeneousSpeed} of the (explicit) homogeneous DW family~\eqref{eq:AnalyticHomogeneousSolution}.

In order to extend the expression of $V$ to cover also the case $\cc\neq0$ and thus approximate higher order terms in~\eqref{eq:PropagationSignIntegral}, we include~\eqref{eq:Potential}, set $\beta=\b$ from \eqref{e:defbeta}, and obtain
\[\widetilde{V}(m_3)\coloneqq -hm_3 + \frac{1}{2}\mu m_3^2 + \frac{\beta}{\alpha \cc}\ln\left(1+\cc m_3 \right),\]
which coincides with $V$ in the limit $\cc\to0$. Concerning the energy difference and direction of motion, we define the quantity
\[\mathcal{H}\coloneqq\frac{\beta}{2\alpha\cc}\left(\ln(1+\cc)-\ln(1-\cc)\right)\] 
and observe that
$\widetilde{V}(+1)>\widetilde{V}(-1)$ holds for $h<\mathcal{H}$ and $\widetilde{V}(+1)<\widetilde{V}(-1)$ if $h>\mathcal{H}$. From this, we presume that a DW connecting $+\3$ on the left and $-\3$ on the right will invade $+\3$ if $\b^+/\alpha+\mu<h<\mathcal{H}$ and $-\3$ if $\mathcal{H}<h<\b^-/\alpha-\mu$. On the other hand, DWs connecting $-\3$ on the left and $+\3$ on the right should propagate towards $+\3$ for $\mathcal{H}<h<\b^-/\alpha-\mu$ and towards $-\3$ if $\b^+/\alpha+\mu<h<\mathcal{H}$. This is indeed in good agreement with numerical simulations for $h$ and $\cc$ above $\Gamma^-$ as well as above $\Gamma^+$, as illustrated in Figure~\ref{fig:SpeedSignBistableCase} (a). Here we have simulated the PDE dynamics with \textsc{Pde2Path}\xspace from step function-like initial conditions connecting the state $+\3$ on the left and $-\3$ on the right, and have tracked the asymptotic propagation speed over time by the freezing method, which is described in more detail in \S~\ref{sec:NumericalSimulationMethod}. We corroborate these predictions for $s=0$ by numerical continuation in \textsc{AUTO-07P} with $s=0$ fixed in order to obtain the loci of (numerically) exact stationarity. We illustrate the results in Figure~\ref{fig:SpeedSignBistableCase} (b) for the range $-0.598<\cc\le 0.9$, where the lower bound represents the point of transition into the $\M^-$ regime. We further remark that $\mathcal{H}$ does not give precise predictions: it is symmetric in the parameter $\cc$, whereas system~\eqref{eq:SphericalCoherentStructureODE} is not, and indeed the numerical results in Figure~\ref{fig:SpeedSignBistableCase} (b) obtained from continuation (red dashed) are slightly asymmetric. Nevertheless, we conclude that $\mathcal{H}-h$ approximates the sign of the asymptotic spreading speed surprisingly well for a broad range of $\cc\neq 0$.

\begin{figure}
    \centering
    \begin{subfigure}[b]{0.56\textwidth}
    	\includegraphics[trim=50 120 70 130,clip,width=\textwidth]{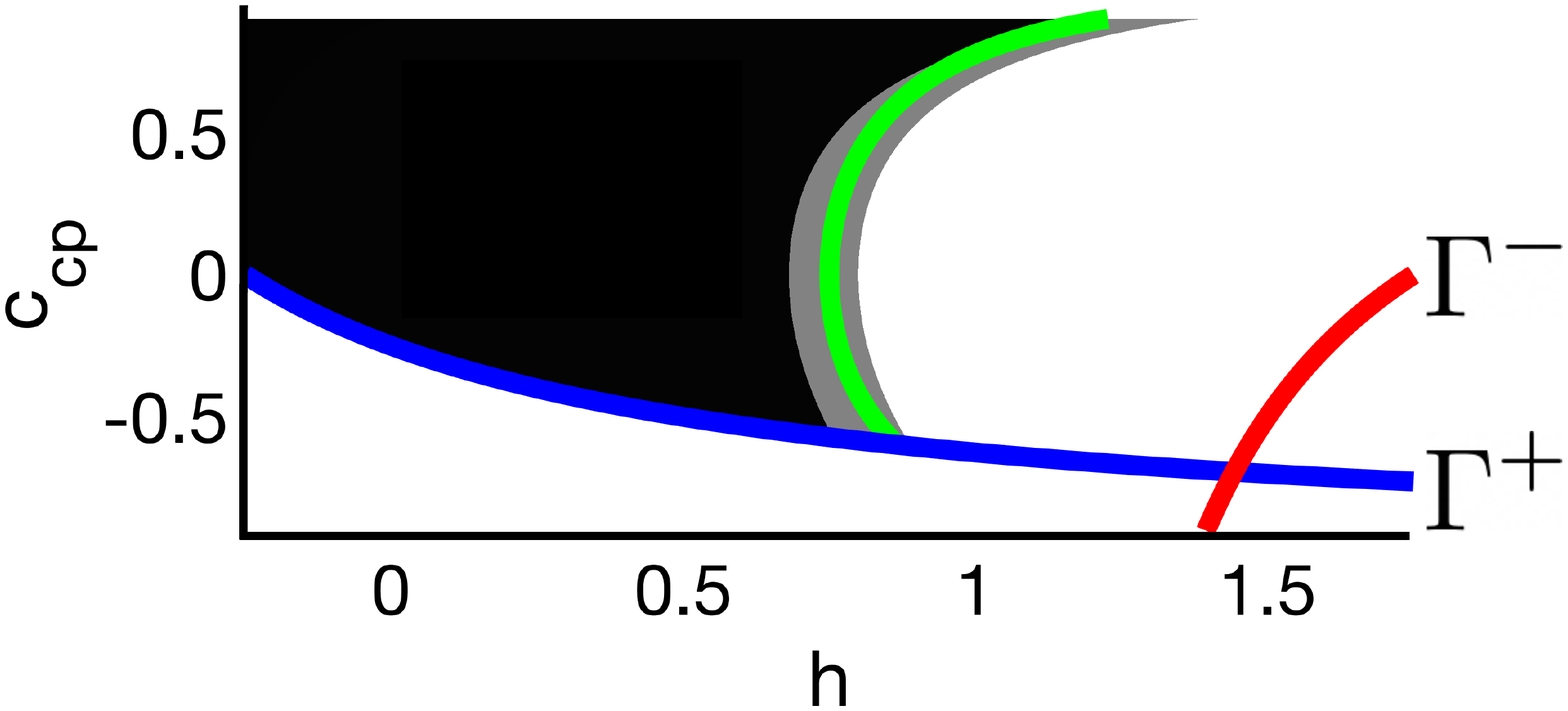}
    	\caption*{(a)}
    \end{subfigure}
    \hspace*{10mm}
    \begin{subfigure}[b]{0.3\textwidth}
    	\includegraphics[trim=10 0 20 35,clip,width=\textwidth]{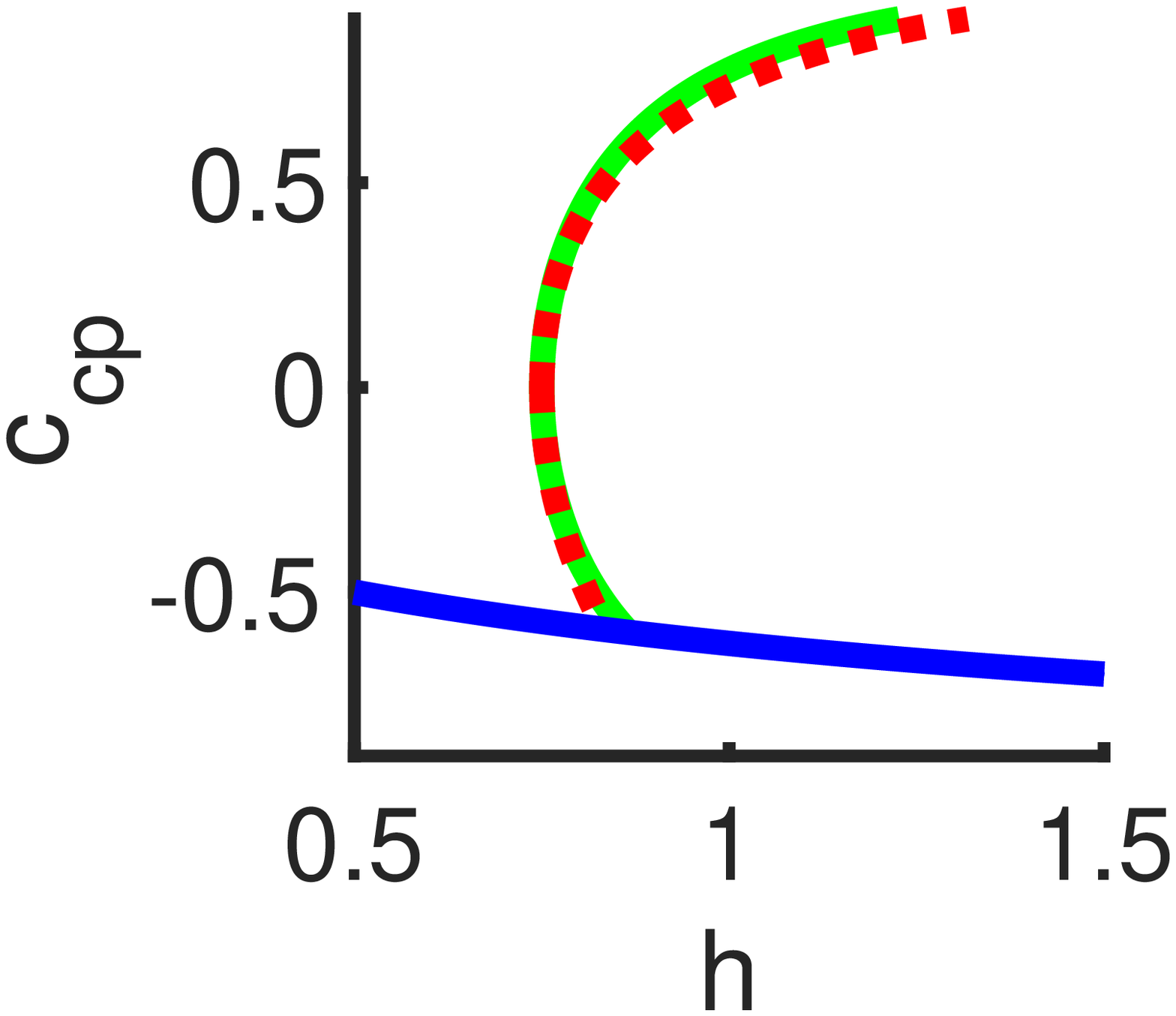}
    	\caption*{(b)}
    \end{subfigure}
    \caption{Comparison of $h=\calH$ 
     with numerical computations in part of the bistable region (above $\Gamma^+$ and $\Gamma^-$) in left panel of Figure~\ref{fig:StabilityRegionsUp-AndDown-Magnetization}; the other parameters are taken as there. (a) $h\in[\beta/\alpha+\mu, \beta/\alpha-\mu]$ and $\cc\in[-0.9,0.9]$. Green curve marks $h=\calH$; left of it $h>\calH$, 
    right of it $h<\calH$. 
    Black shaded region marks numerically observed approach of DWs with negative speed. In gray region the speed is within $\pm$\num{e-6}. In white bistable region DWs attain positive speed. Initial data is step function-type connecting $+\3$ on the left and $-\3$ on the right. (b) Magnification around green curve in (a) and red dashed line results from numerical continuation for $-0.598<\cc\le 0.9$ with $s=0$ fixed.}
    \label{fig:SpeedSignBistableCase}
\end{figure}

\section{Domain wall motion: monostable case}\label{sec:Monostable}
The study of interfaces invading an unstable state, also in the context of ferromagnetic dynamics, has a long history, see e.g.~\cite{aronson1975nonlinear, ebert2000front, elmer1993front, elmer1994dual, fisher1937wave, garnier2012inside, goussev2020dynamics, holzer2014anomalous, kolmogorov1937bulletin, van2003front}. The most accurate method for predicting the asymptotic behavior in systems of parabolic equations is based on the so-called pointwise Green's function associated with the unstable state~\cite{beck2017stability, holzer2014anomalous, holzer2014criteria}, rather than the saddle-point analysis that is sufficient for the classical scalar case~\cite{ebert2000front, van2003front}. It turns out that in the present case this can be studied via the so-called absolute spectrum as defined below.

We will discuss a more general problem and present an explicit expression for the absolute spectrum, denoted by $\Sabs(\mathscr{L})$, of certain complex linear $2\times 2$ matrix operators $\mathscr{L}$. This class also includes the operator arising from linearizing the complex Ginzburg-Landau equation around the origin, as discussed below. Based on the explicit formula for $\Sabs(\mathscr{L})$, we will show that its most unstable elements are \emph{simple pinched double roots} and correspond to singularities of the pointwise Green's function. From this, we will infer the so-called \emph{linear spreading speed} and also define the associated \emph{linear spreading frequency}, where both provide predictions for the asymptotic behavior of DWs in the nonlinear system~\eqref{eq:LLGS}.

Although the linear spreading speed can be obtained from values of $s$ for which the critical pinched double root is placed on the imaginary axis, the spreading frequency is not selected in the sense of (linear) stability properties. It is selected \emph{passively} from the oscillatory dynamics, since -- due to rotation symmetry -- changes in frequency modify the imaginary part only.

To simplify notations, we will focus on the parameter regime $\M^+$ introduced in~\eqref{eq:MonostableSetPlus}; all results presented in this section also transfer to the regime $\M^-$.

\subsection{Pointwise growth, absolute spectra, and double roots}\label{sec:Pointwise growth, absolute spectra, and double roots}

The pointwise Green's function in our situation can be defined as follows. Again, we allow for perturbations tangential to the sphere at the pole $-\3$ only and obtain from linearization of~\eqref{eq:QuasilinearSpeedFrequencyLLGS}
\[\partial_t \n=\calL^- \n + \calN(\n),\]
where the linear matrix operator $\calL^-$ is given via~\eqref{eq:LinearEigenvalueProblem} for $\overline{m}=-\3$ and $\calN(\n)$ describes the nonlinear terms. Considering the linear part only and after a Laplace transform in time, the equation reduces to a system of second order ODEs which depends on the spectral parameter $\lambda$. The associated pointwise (matrix) Green's function $\bG(\xi-y)$ is defined as the solution to
\[(\calL^--\lambda)\bG(\xi-y)=\delta(y)\]
in the distributional sense. Applying the inverse Laplace transform with a contour $\mathcal{C}$ to the right of the essential spectrum of $\calL^-$, which we determine later, the solution to the initial problem can now be written as
\begin{equation}\label{eq:LinearSolution}
\n=-\frac{1}{2\pi\rmi}\int_\mathcal{C}\rme^{\lambda t}\int_\R \bG(\xi-y)\n(y,0)\,\textnormal{d}y\, \textnormal{d}\lambda,
\end{equation}
where we supposed that $\n(y,0)$ has compact support. We notice that the contour $\mathcal{C}$ can be deformed beyond the essential spectrum up to singularities of $\bG$~\cite{gardner1998gap, kapitula2013spectral}. In this way, if $\bG$ can be analytically continued into the left half-plane, decay bounds on the solution can be inferred. 

Analytical properties of $\bG$ depend on the frame of reference and thus on $s$ and $\Omega$ in our situation. It turns out that for $s=0$, $\bG$ possesses singularities at the complex conjugated pair of points of the essential spectrum with maximal real part. For more general operators it is well known that upon changing $s$, the singularities can be moved into the stable half-plane, which describes precisely the transition from pointwise growth to pointwise decay. As mentioned before, the frequency $\Omega$ has no influence on the real part of the singularities. Viewing $s$ now as a parameter, the transition from pointwise growth to pointwise decay determines the change from so-called \emph{convective} to \emph{absolute} instability of $-\3$. Here, we call an instability convective, if perturbations grow in norm, but decay pointwise, while for an absolute instability a generic perturbation grows exponentially at each point where it was applied~\cite{kapitula2013spectral, sandstede2002stability, sandstede2000absolute}.

As mentioned, relevant singularities of $\bG$ in our situation can be identified via the absolute spectrum associated to $\calL^-$. Here, we consider the more general (complex) linear operators 
\begin{equation}\label{eq:GeneralMatrixOperator}
\mathscr{L}\coloneqq \begin{pmatrix} a_2\partial_\xi^2+a_1\partial_\xi + a_0 & -b_2\partial_\xi^2 -b_1\partial_\xi- b_0\\ b_2\partial_\xi^2 +b_1\partial_\xi  +b_0 & a_2\partial_\xi^2+a_1\partial_\xi + a_0\end{pmatrix},
\end{equation}
with $a_j, b_j\in \C, j=0,1,2$, $a_2, b_2\neq 0$, and whose dispersion relation reads
\begin{equation}\label{eq:GeneralDispersionRelation}
\mathscr{D}(\lambda,\nu)\coloneqq \left(a_2\nu^2+a_1\nu+a_0-\lambda\right)^2+\left( b_2\nu^2+b_1\nu+b_0\right)^2,
\end{equation}
For what follows it is convenient to denote
\[a_2^r\coloneqq\Re(a_2),\, b_2^r\coloneqq\Re(b_2),\, a_2^i\coloneqq\Im(a_2),\, b_2^i\coloneqq\Im(b_2).\]
We shall assume the following, which relates to uniform ellipticity of $\mathscr{L}$.
\begin{Hypothesis}\label{hyp:Coefficients} The coefficients $a_2, b_2\in\C$ satisfy $a^r_2>|b_2^i|^2$.
\end{Hypothesis}

In order to define the absolute spectrum, cf.~\cite{sandstede2000absolute}, we solve $\mathscr{D}(\lambda,\nu)=0$ for given $\lambda$, which leads to four so-called spatial eigenvalues $\nu_j(\lambda),\, j=1,\dots,4$, including multiplicity, that we order by descending real parts as
\begin{equation}\label{eq:Ordering}
\Re(\nu_1(\lambda))\ge \Re(\nu_2(\lambda))\ge\Re(\nu_3(\lambda))\ge \Re(\nu_4(\lambda)).
\end{equation}
Next, we determine the Morse index~\cite{sandstede2000absolute} associated to the spatial eigenvalues, i.e., the number of unstable ones for sufficiently large $\Re(\lambda)$.
\begin{lemma}\label{lem:MorseIndexGeneral}
For $a_2, b_2$ satisfying Hypothesis~\ref{hyp:Coefficients}, the Morse index associated to $\mathscr{L}$ is $i_\infty=2$.
\end{lemma}
\begin{proof}
We consider $\Re(\lambda)\gg 1$ by taking $0<\varepsilon\ll1$ in the parabolic scaling $\lambda=\tilde{\lambda}/\varepsilon^2$ as well as $\nu=\tilde{\nu}/\varepsilon$ for arbitrary fixed $\tilde{\lambda}\in \C$ with $\Re(\tilde{\lambda})>0$. Multiplication of $\mathscr{D}(\tilde{\lambda},\tilde{\nu})=0$ by $\varepsilon^4$ leads to
\[(a^2_2+b^2_2)\tilde{\nu}^4-2a_2\tilde{\lambda}\tilde{\nu}^2+\tilde{\lambda}^2+\calO(\varepsilon)=0,\]
whose solutions for $\varepsilon=0$ with four choices of signs are
\[\tilde{\nu}_\pm^\pm(\tilde{\lambda})=\pm\sqrt{\frac{\tilde{\lambda}}{a_2\pm\rmi b_2}}.\]
Under the assumptions on $a_2$ and $b_2$, we show that $\tilde{\nu}_\pm^\pm(\tilde{\lambda})\notin\rmi\R$, which implies the claim, i.e., $\Re(\nu_{1})=\Re(\nu_{2})>0> \Re(\nu_{3})=\Re(\nu_{4})$ in terms of the notation \eqref{eq:Ordering}.
Since the sign of $b_2$ will be irrelevant, it suffices to consider the case $+ \rmi b_2$ and assume that $\tilde{\lambda}/(a_2+ \rmi b_2)\in \R^-$ under Hypothesis~\ref{hyp:Coefficients}. From this, $\Im(\tilde{\lambda}/(a_2+ \rmi b_2))=0$ is equivalent to $(a_2^r-b_2^i)\tilde{\lambda}^i=(a_2^i+b_2^r)\tilde{\lambda}^r$, where $\tilde{\lambda}^r\coloneqq \Re(\tilde{\lambda})$ and $\tilde{\lambda}^i\coloneqq \Im(\tilde{\lambda})$. In case $a_2^r=b_2^i$, it follows that $a_2^i=-b_2^r$, which together leads to $a_2+\rmi b_2=0$, contradicting Hypothesis~\ref{hyp:Coefficients}. Thus $a_2^r\neq b_2^i$ and we obtain
\[\tilde{\lambda}^i=\tilde{\lambda}^r \frac{a_2^i+b_2^r}{a_2^r-b_2^i}.\]
Now $\Re(\tilde{\lambda}/(a_2\pm \rmi b_2))<0$ with $\Re(\tilde{\lambda})>0$ implies
\[a_2^r-b_2^i + \frac{(a_2^i+b_2^r)^2}{a_2^r-b_2^i}<0.\]
This inequality only holds for $a_2^r<b_2^i$ and hence contradicts Hypothesis~\ref{hyp:Coefficients}. 
\end{proof}

It follows that under Hypothesis \ref{hyp:Coefficients} the absolute spectrum of $\mathscr{L}$ with~\eqref{eq:Ordering} is given by the set
\[\Sabs(\mathscr{L})=\lbrace \lambda \in \C : \mathscr{D}(\lambda,\nu)=0 \textnormal{ and } \Re(\nu_2(\lambda))=\Re(\nu_3(\lambda))\rbrace.\]

We now turn to the aforementioned points in the absolute spectrum that can give rise to singularities of the pointwise Green's function~\cite{holzer2014criteria, sandstede2000absolute}.
\begin{definition}\label{def:DoubleRoots}
A pair $(\lambda^\textnormal{dr},\nu^\textnormal{dr})\in\C^2$ is called \emph{double root} if
\[\mathscr{D}(\lambda^\textnormal{dr},\nu^\textnormal{dr})=0\quad \textnormal{and} \quad \partial_\nu \mathscr{D}(\lambda^\textnormal{dr},\nu^\textnormal{dr})=0.\]
A double root is called \emph{simple} if $\partial_\lambda \mathscr{D}(\lambda^\textnormal{dr},\nu^\textnormal{dr})\neq 0$ and $\partial_\nu^2 \mathscr{D}(\lambda^\textnormal{dr},\nu^\textnormal{dr})\neq 0$.
\end{definition}

To obtain a direct relation between double roots and singularities of $\bG$, one has to further require a \emph{pinching condition}, where non-pinched double roots correspond to removable singularities and an analytic extension of $G_\lambda$ beyond these roots is available~\cite{holzer2014anomalous, holzer2014criteria, sandstede2000absolute}.
\begin{definition}\label{def:PinchingCondition}
A double root $(\lambda^\textnormal{dr},\nu^\textnormal{dr})$ is called \emph{pinched} if there exists a continuous curve $\Lambda(\tau)\in\C, \tau\ge 0$, with $\Lambda(0)=\lambda^\textnormal{dr}$, $\Re(\Lambda(\tau))\to +\infty$ for $\tau\to +\infty$, and continuous curves of roots $\nu_\pm(\Lambda(\tau))$ to $\mathscr{D}(\Lambda(\tau),\nu)=0$ with $\nu_\pm(\lambda^\textnormal{dr})=\nu^\textnormal{dr}$ and $\Re(\nu_\pm(\Lambda(\tau)))\to\pm\infty$ for $\tau\to +\infty$.
\end{definition}

Decisive for the identification of pinched double roots in our setting is the following observation. In preparation, denote by $\mathcal{G}_\mathrm{abs}\subset \C$ the connected component of $\C \setminus\Sabs$ that contains an unbounded interval of $\R^+$. Since the Morse index $i_\infty$ is well-defined, there exists a $\rho\in \R$ such that $\Sabs\cap \lbrace \lambda\in\C : \Re(\lambda)\ge \rho \rbrace=\emptyset$, so that $\mathcal{G}_\mathrm{abs}$ is well-defined.

\begin{lemma}\label{lem:PinchingAndMorseIndex}
A double root $(\lambda^\textnormal{dr},\nu^\textnormal{dr})$ for $\mathscr{D}$ with $\lambda^\textnormal{dr}\in \partial \mathcal{G}_\mathrm{abs}$ is pinched. In particular, this is the case if $\Re(\lambda^\textnormal{dr}) = \max\Re (\Sabs)$.
\end{lemma}
\begin{proof}
Since $\lambda^\mathrm{dr}\in\partial \mathcal{G}_\mathrm{abs}$, one finds a continuous curve $\Lambda(\tau)\in\C, \tau\ge 0$, with $\Lambda(0)=\lambda^\mathrm{dr}$, into the region $\lbrace \lambda\in\C:\Re(\lambda)>\max \Re(\Sess)\rbrace\subset \C$, such that $\Sabs\cap\Lambda=\lbrace \lambda^\mathrm{dr}\rbrace$. We may assume there is no crossing with $\Sabs$ along $\Lambda$ for $\tau>0$ and thus the roots $\nu(\Lambda(\tau))$ of $\mathcal{D}$ are ordered as in~\eqref{eq:Ordering}. The fact that the real parts grow unboundedly with opposite signs follows from the scaling result in Lemma~\ref{lem:MorseIndexGeneral} (and similarly for general parabolic well-posed problems).
\end{proof}
\begin{Remark}
If all curves that connect $\lambda^\mathrm{dr}$ with a point in $\lbrace \lambda\in\C:\Re(\lambda)>\max \Re(\Sess)\rbrace\subset \C$ intersect $\Sabs$, then Lemma~\ref{lem:PinchingAndMorseIndex} is not applicable and it is not known under what (additional) conditions the double root would be pinched.
\end{Remark}

Due to the index conditions, explicit representations of $\Sabs$ for higher order or non-scalar cases, e.g.\ the Swift-Hohenberg equation, are rarely possible. An exception is the complex Ginzburg-Landau (CGL) equation and due to the similar structure in the dispersion relation, it is possible to obtain explicit expressions for the class of operators $\mathscr{L}$. Our approach is based on the fact that $\Sabs$ is the union of curves that must lie fully in the region bounded by the essential spectrum that includes an unbounded part of $\R^-$~\cite{kapitula2013spectral, rademacher2006geometric}. Moreover, each point of the absolute spectrum arises as an intersection point of (at least) two curves of the essential spectrum in exponentially weighted $L^2$-spaces~\cite{fiedler2003spatio, rademacher2006geometric}. Here the essential spectrum of $\mathscr{L}$ in such a space with exponential weight $\eta\in\R$ consists of solutions to the dispersion relation~\eqref{eq:GeneralDispersionRelation} with $\nu=\rmi k+\eta$, i.e., the set
\[\Sess^\eta(\mathscr{L})=\lbrace \lambda\in\C:\mathscr{D}(\lambda,\rmi k+\eta)=0, k\in\R \rbrace;\]
note $\Sess(\mathscr{L})=\Sess^0(\mathscr{L})$. The key to locate pinched double roots in our case is the following.
\begin{theorem}\label{theo:AbsolutSpectrumGeneralLinearMatrixOperator}
For $a_2, b_2$ satisfying Hypothesis~\ref{hyp:Coefficients}, the absolute spectrum of $\mathscr{L}$ is given by the following two (straight) half-lines parameterized by $r\in[0,\infty)$,
\[-(a_2+\rmi b_2) r-\frac{(a_1+\rmi b_1)^2}{4(a_2+\rmi b_2)}+a_0+\rmi b_0\quad \textnormal{and}\quad -(a_2-\rmi b_2) r-\frac{(a_1-\rmi b_1)^2}{4(a_2-\rmi b_2)}+a_0-\rmi b_0.\]
The rightmost points lie at $r=0$ and are pinched double roots of~\eqref{eq:GeneralDispersionRelation}.
\end{theorem}
\begin{proof} To shorten notations, we set $c_j=a_j+\rmi b_j, \tilde{c}_j=a_j-\rmi b_j$ and
\[c_j^r=\Re(c_j), \,c_j^i=\Im(c_j), \,\tilde{c}_j^r=\Re(\tilde{c}_j), \,\tilde{c}_j^i=\Im(\tilde{c}_j)\]
for $j=0,1,2$. Starting from $\mathscr{D}(\lambda,\rmi k +\eta)=0$ with wavenumber $k\in\R$ and weight $\eta\in\R$, the weighted essential spectrum of $\mathscr{L}$ is given by the two (parabolic) curves
\[\lambda_1(k;\eta)=-c_2k^2 + \rmi(c_1+2c_2\eta)k+c_2\eta^2+c_1\eta+c_0\]
and
\[\lambda_2(k;\eta)=-\tilde{c}_2k^2 + \rmi(\tilde{c}_1+2\tilde{c}_2\eta)k+\tilde{c}_2\eta^2+\tilde{c}_1\eta+\tilde{c}_0\]
for $k\in\R$. The maximal real part of $\lambda_1$ satisfies $\Re\left(\partial_k \lambda_1(k;\eta)\right)=0$ which is solved by $k^\ast\coloneqq-\frac{c_1^i+2c_2^i\eta}{2c_2^r}$.
Additionally,
\[\Re\left(\partial_\eta \lambda_1(k^\ast; \eta)\right)=2\frac{(c_2^r)^2 + (c_2^i)^2}{c_2^r}\eta+c_1^r+\frac{c_2^i c_1^i}{c_2^r},\]
which vanishes at
\[\eta^\ast=-\frac{c_2^r c_1^r + c_2^i c_1^i}{2\left((c_2^r)^2+(c_2^i)^2\right)}.\]
For this value of $\eta$, we obtain
\begin{align*}\lambda_1(k;\eta^\ast)&=-c_2k^2+\rmi\frac{c_2^i c_1^r - c_2^rc_1^i}{c_2^i+\rmi c_2^r}k\\ &\quad + \frac{\left(c_2^rc_1^2 +  c_2^i c_1^i - 2\rmi c_2^ic_1^r + 2\rmi c_2^rc_1^i \right) \left( c_2^r c_1^r+ c_2^i c_1^i \right)}{4c_2 \left(c_2^i + \rmi c_2^r \right)^2} + c_0\\
&=-c_2\overline{k}^2 -\frac{c_1^2}{4c_2} + c_0,
\end{align*}
where we introduced the shift
\[\overline{k}\coloneqq k+\frac{c_2^r c_1^i-c_2^i c_1^r}{2(c_2^r)^2+2(c_2^i)^2}\in\R.\]
This leads to a straight half-line with slope $-c_2$ and maximal real part at
\[\frac{c_2^r(c_1^i)^2-c_2^r(c_1^r)^2-2c_2^ic_1^rc_1^i}{4((c_2^r)^2+(c_2^i)^2)}+c_0^r \quad \textnormal{for} \quad k=\frac{c_2^ic_1^r-c_2^rc_1^i}{2(c_2^r)^2 + 2(c_2^i)^2}.\]
Analogously, for $\lambda_2$ we obtain
\[\lambda_2(k;\eta^\ast)=-\tilde{c}_2 \check{k}^2-\frac{\tilde{c}_1^2}{4\tilde{c}_2}+\tilde{c}_0, \quad \check{k}\coloneqq k+\frac{\tilde{c}_2^r \tilde{c}_1^i-\tilde{c}_2^i \tilde{c}_1^r}{2(\tilde{c}_2^r)^2+2(\tilde{c}_2^i)^2}\in\R.\]
To locate double roots, we determine the roots of the resultant function of $\mathscr{D}(\lambda,\nu)$ and $\partial_\nu \mathscr{D}(\lambda,\nu)$ in $\lambda$, which indeed results in the two endpoints. Since $\Sabs(\mathscr{L})$ has to lie in the region bounded by $\Sess^\eta(\mathscr{L})$ (see~\cite{rademacher2006geometric} for full statement) which forms (doubly covered) straight lines precisely at $\eta=\eta^*$, it follows that $\Sabs(\mathscr{L})$ is given by $\lambda_{1,2}(k;\eta^\ast)$. One readily checks that the double roots $\lambda^\pm_\textnormal{dr}$ are the rightmost points of $\lambda_{1,2}(k;\eta^\ast)$ and thus also of $\Sabs(\mathscr{L})$. Lemma~\ref{lem:PinchingAndMorseIndex} implies that these are pinched.
\end{proof}

Theorem~\ref{theo:AbsolutSpectrumGeneralLinearMatrixOperator} allows to explicitly represent the absolute spectrum of $\calL^-$ and its most unstable modes. See Figure~\ref{fig:SpecsOfDownMagnetization} for illustrations. The dispersion relation of $\calL^-$ reads
\begin{equation}\label{eq:DispersionRelationDownState}
\begin{split}
d^-(\lambda, \nu)&\coloneqq\left(\frac{\alpha}{1+\alpha^2}\nu^2+s\nu+\frac{\alpha(h+\mu)-\b^-}{1+\alpha^2}-\lambda\right)^2\\
&\quad+\left(\frac{1}{1+\alpha^2}\nu^2-\Omega+\frac{h+\mu+\alpha\b^-}{1+\alpha^2} \right)^2
\end{split}
\end{equation}
which gives the form~\eqref{eq:GeneralDispersionRelation} with
\begin{equation}\label{eq:GeneralCoefficientsA}
a_2=\frac{\alpha}{1+\alpha^2},\quad a_1=s,\quad a_0=\frac{\alpha(h+\mu)-\b^-}{1+\alpha^2},
\end{equation}
as well as
\begin{equation}\label{eq:GeneralCoefficientsB}
b_2=\frac{1}{1+\alpha^2},\quad b_1=0, \quad b_0=\frac{h+\mu+\alpha\b^-}{1+\alpha^2}-\Omega.
\end{equation}
In particular, Hypothesis~\ref{hyp:Coefficients} is always satisfied for~\eqref{eq:GeneralCoefficientsA} and~\eqref{eq:GeneralCoefficientsB}, and thus we obtain the following statement.

\begin{corollary}\label{cor:AbsoluteSpectrum}
The absolute spectrum $\Sabs(\calL^-)$ of the linearization around the unstable state $-\3$ in the $\M^+$ regime is given by
\[\Sabs(\calL^-)=\Sabs^+(\calL^-) \cup \Sabs^-(\calL^-),\]
where $\Sabs^\pm(\calL^-)$ are the straight half-lines
\[\Big\lbrace -\frac{\alpha\pm\rmi}{1+\alpha^2}k-\frac{\alpha\mp\rmi}{4}s^2+\frac{\alpha(h+\mu)-\b^-}{1+\alpha^2}\pm\rmi\left( \frac{h+\mu+\alpha\b^-}{1+\alpha^2}-\Omega \right), k\in[0,\infty) \Big\rbrace\]
in the complex plane. Moreover, the most unstable points are pinched double roots that are simple if $s\neq 0$, and are given by
\begin{equation}\label{eq:DoubleRootLambda}
\lambda^\textnormal{dr}_\pm=-\alpha s^2/4\pm \rmi s^2/4+\frac{\alpha(h+\mu)-\b^-}{1+\alpha^2}\pm\rmi\left(\frac{h+\mu+\alpha\b^-}{1+\alpha^2}-\Omega\right),
\end{equation}
with $\nu^\textnormal{dr}_+=\nu(\lambda_\textnormal{dr}^+)=-\alpha s/2+\rmi s/2$ and $\nu^\textnormal{dr}_-=\nu(\lambda_\textnormal{dr}^-)=-\alpha s/2-\rmi s/2$.
\end{corollary}
\begin{proof}
The first part is a direct consequence of Theorem~\ref{theo:AbsolutSpectrumGeneralLinearMatrixOperator} via~\eqref{eq:GeneralCoefficientsA} and~\eqref{eq:GeneralCoefficientsB}, from which it also follows that the rightmost points are pinched double roots given by $(\lambda^\textnormal{dr}_\pm,\nu^\textnormal{dr}_\pm)$. Solving for $d_\lambda(\lambda^\textnormal{dr}_\pm,\nu^\textnormal{dr}_\pm)=0$ and $d_{\nu\nu}(\lambda^\textnormal{dr}_\pm,\nu^\textnormal{dr}_\pm)=0$ requires $s=0$, which completes the proof.
\end{proof}
\begin{Remark}\label{rem:OptimalWeight}
The absolute spectrum $\Sabs(\calL^-)$ coincides with the weighted essential spectrum $\Sess^{\eta^*}(\calL^-)$ for $\eta^*\coloneqq-(\alpha s)/2$.
\end{Remark}
\begin{Remark}\label{rem:ComplexGinzburgLandau}
We highlight the structural relation to the extensively studied complex Ginzburg-Landau equation (CGL), whose spectra are well-known, see e.g.~\cite{aranson2002world} for details. The CGL in a co-moving frame as well as co-rotating frame with speed $s$ and frequency $\Omega$ reads
\begin{equation*}\label{eq:ComplexGinzburgLandauEquation}
\partial_t A=(1+\rmi\gamma)\partial_\xi^2 A +s\partial_\xi A + A - (1+\rmi\kappa)|A|^2 A-\rmi \Omega A, \tag{CGL}
\end{equation*}
where $\xi=x-st$, $A=A(\xi,t)\in\C$ and the real parameters $\gamma, \kappa$ characterize the linear and nonlinear dispersion. In terms of real and imaginary parts, $A=u+\rmi v$, the linearization around the zero state reads
\begin{align*}
\partial_t \begin{pmatrix}
u\\ v
\end{pmatrix}=
\begin{pmatrix}
\partial_\xi^2 + s\partial_\xi + 1 & - (\gamma\partial_\xi^2 -\Omega)\\
\gamma \partial_\xi^2 -\Omega & \partial_\xi^2 +s\partial_\xi + 1 
\end{pmatrix}\begin{pmatrix}
u\\ v
\end{pmatrix},
\end{align*}
which is of the form~\eqref{eq:GeneralMatrixOperator} with $a_2=a_0=1$, $a_1=s$, $b_2=\rmi\gamma $, $b_1=0$, and $b_0=-\rmi\Omega$. Since $a_2$ and $b_2$ satisfy the conditions of Hypothesis~\ref{hyp:Coefficients}, we directly obtain from Theorem~\ref{theo:AbsolutSpectrumGeneralLinearMatrixOperator} that the absolute spectrum of (CGL) linearized around the zero state $A=0$ is given by
\begin{align*}
\Sabs&=\bigg\lbrace -(1+\rmi \gamma)k - \frac{s^2}{4(1+\rmi \gamma)}+1-\rmi \Omega, \, k\in[0,\infty) \bigg\rbrace
\\ &\,\cup \bigg\lbrace -(1-\rmi \gamma)k - \frac{s^2}{4(1-\rmi \gamma)}+1+\rmi \Omega, \, k\in[0,\infty) \bigg\rbrace
\end{align*}
with maximal real part located at $-s^2/(4(1+\gamma^2))+1$. It is well-known (see e.g.~\cite{goh2020spectral, van2003front}), that $\Sabs$ is marginally stable for $|s|=2\sqrt{1+\gamma^2}$.
\end{Remark}

\subsection{Linear spreading speed and frequency}\label{sec:LinearSpreadingSpeed}

It has been shown in~\cite{holzer2014criteria} that for any simple pinched double root $(\lambda^\textnormal{dr}, \nu^\textnormal{dr})$ it holds that $\lambda^\textnormal{dr}$ is a singularity of $\bG$. Indeed there are examples that simplicity is in general necessary~\cite{holzer2014criteria}. Thus, the most definite indicator for the transition from convective to absolute instability is the location of the most unstable simple pinched double root~\cite{holzer2014criteria}.
\begin{definition}\label{def:LinearSpreadingSpeed}
The (positive) \emph{linear spreading speed} is
\[s^\textnormal{lin}\coloneqq\sup \,\lbrace s\,:\, \mathscr{D} \textnormal{ possesses a simple pinched double root with } \Re(\lambda^\textnormal{dr})>0 \rbrace.\]
\end{definition}
The set in this definition is non-empty in our case of unstable essential spectrum and the spreading speed is bounded, cf.~\cite{holzer2014criteria}. One analogously defines a negative spreading speed via the infimum and for reflection symmetric problems, as ours, these have the same absolute value. Based on Corollary~\ref{cor:AbsoluteSpectrum}, we directly obtain the following result.
\begin{lemma}\label{lem:LinearSpreadingSpeed}
For any $\alpha>0, \beta>0, \mu<0, \cc\in(-1,1)$, the positive linear spreading speed associated to $\calL^-$ is
\begin{equation}\label{eq:SelectedSpeedMonostable}
    s^{\textnormal{lin}}=2\sqrt{\frac{h+\mu-\b^-/\alpha}{1+\alpha^2}}
\end{equation}
in the $\M^+$ regime, where in particular $\b^-/\alpha-\mu<h$ holds.
\end{lemma}
\begin{proof}
Due to Corollary~\ref{cor:AbsoluteSpectrum}, the two points in $\Sabs(\calL^-)$ with maximal real part are simple pinched double roots for $s\neq 0$; $\lambda^\textnormal{dr}_\pm$ with $\lambda^\textnormal{dr}_+=\overline{\lambda^\textnormal{dr}_-}$. Moreover, we readily see that $\Re(\lambda^\pm_\textnormal{dr})=0$ has positive unique solution $s=s^{\textnormal{lin}}$. Based on Definition~\ref{def:LinearSpreadingSpeed}, we obtain the statement.
\end{proof}

\begin{Remark}\label{rem:SymmetryLinearSpreadingSpeed}
Recall that fronts, in particular DWs, that travel with the linear spreading speed are called pulled. By reflection symmetry, the negative spreading speed is $-s^\textnormal{lin}$. A pulled DW with stable state $+\3$ to the left and unstable state $-\3$ to the right will propagate with $s^\textnormal{lin}$, whereas it will propagate with $-s^\textnormal{lin}$ if the stable state is to the right and the unstable state to the left, both in the $\mathcal{M}^+$ parameter regime (see also \S~\ref{sec:StandingDomainWalls} for discussion on the spreading direction). For parameters in $\M^-$, one has to interchange the directions of propagation.
\end{Remark}

Setting $s=s^\textnormal{lin}$, all singularities of $\bG$ lie in the closed left half-plane and one can compute~\eqref{eq:LinearSolution}. Thus, we directly obtain the following definition for the solution to be also stationary regarding rotation.
\begin{definition}\label{def:LinearSpreadingFrequency}
The \emph{linear spreading frequency} for $\mathcal{L}$ is $\Omega^\textnormal{lin}\coloneqq \Im\left(\lambda^\textnormal{dr}\big\vert_{\Omega=0}\right)$.
\end{definition}

This reflects the fact that in a frame moving with speed $s^\textnormal{lin}$, critical perturbations in the weighted space with weight $\eta^*$ as specified in Theorem~\ref{theo:AbsolutSpectrumGeneralLinearMatrixOperator} oscillate for $\Omega=0$ with frequency $\Im(\lambda^\textnormal{dr})$. Hence, in a frame moving and rotating with the linear spreading speed as well as linear spreading frequency, critical modes of~\eqref{eq:LinearSolution} are stationary with respect to the two symmetries; translation and rotation, see Figure~\ref{fig:SpecsOfDownMagnetization}.

Again based on Corollary~\ref{cor:AbsoluteSpectrum}, we obtain the  following statement.
\begin{lemma}\label{lem:LinearSpreadingFrequency}
For any $\alpha>0, \beta>0, \mu<0, \cc\in(-1,1)$, the linear spreading frequency associated to $\calL^-$ in the $\M^+$ regime is
\begin{equation}\label{eq:SelectedFrequencyMonostable}
\Omega^\textnormal{lin}=\frac{2h+2\mu+\alpha \b^- - \b^-/\alpha}{1+\alpha^2}
\end{equation}
\end{lemma}
\begin{proof}
By Lemma~\ref{lem:LinearSpreadingSpeed}, the linear spreading speed is given by~\eqref{eq:SelectedSpeedMonostable} and from Corollary~\ref{cor:AbsoluteSpectrum}, we compute that $\Im(\lambda_\pm^\textnormal{dr}\big\vert_{\Omega=0})=0$ for $s=s^\textnormal{lin}$ has unique solution $\Omega^\textnormal{lin}$. The statement follows now from Definition~\ref{def:LinearSpreadingFrequency}.
\end{proof}

\begin{Remark}
Although the sign of the linear spreading speed depends on whether the invaded state is on the left or on the right (see Remark~\ref{rem:SymmetryLinearSpreadingSpeed}), the linear spreading frequency is positive within the parameter space to $\mathcal{M}^+$, where in particular $\b^-/\alpha-\mu<h$ holds. Additionally, for $\alpha=1$, $\Omega^\textnormal{lin}$ is independent of the parameter $\cc$ (recall $\beta>0$ in case of spin-torque), whereas $s^\textnormal{lin}_\pm$ depends on $\cc$. Moreover, we remark that in case $s=s^\textnormal{lin}_\pm$ and $\Omega=\Omega^\textnormal{lin}$ one obtains from~\eqref{eq:DoubleRootLambda} that $\lambda^\textnormal{dr}_\pm=0$ and that solving $d^-(0,\nu)=0$ leads to the four spatial solutions $\nu_j=-\alpha s^\textnormal{lin}_\pm/2 \pm \rmi s^\textnormal{lin}_\pm/2$. This is a `doubled' double root in the sense of the two symmetries: translation and rotation.
\end{Remark}

\begin{figure}
    \centering
    \begin{subfigure}[b]{0.24\textwidth}
    \includegraphics[width=\textwidth]{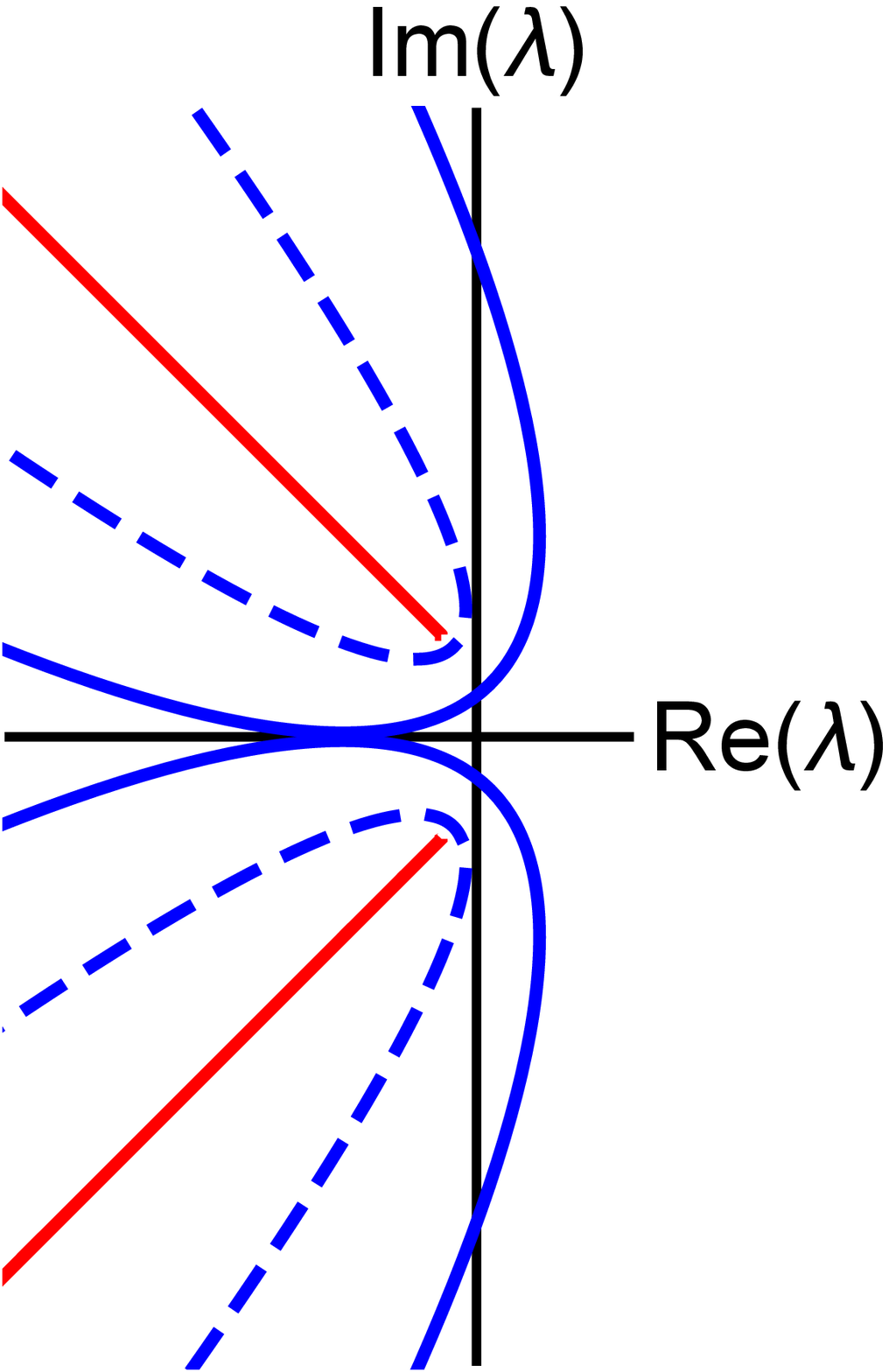}
    \caption*{(a)}
    \end{subfigure}
    \hspace*{20mm}
    \begin{subfigure}[b]{0.32\textwidth}
    \includegraphics[trim=0 70 0 70, clip, width=\textwidth]{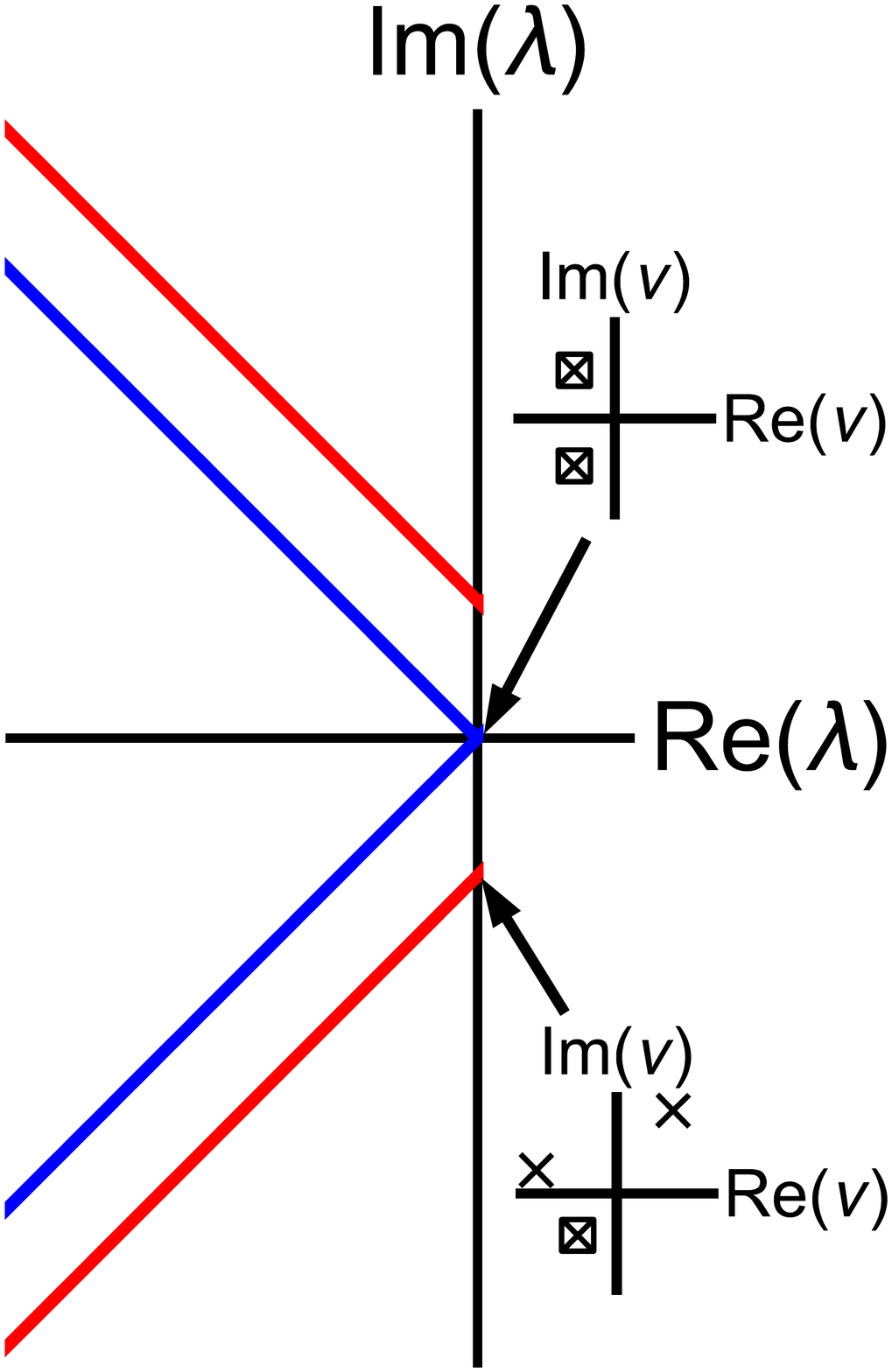}
    \caption*{(b)}
    \end{subfigure}
    \caption{Spectra of $\calL^-$ in the $\M^+$ parameter regime with $s>s^\textnormal{lin}_+$ and $\Omega>\Omega^\textnormal{lin}$ in (a), where blue solid line is essential and blue dashed the weighted essential spectrum with weight $\eta^*<\eta<0$. Red lines illustrate the absolute spectrum. (b) absolute spectrum for $s=s^\textnormal{lin}_+$ as defined in~\eqref{eq:SelectedSpeedMonostable} and different frequencies: red for $\Omega>\Omega^\textnormal{lin}$ and blue for $\Omega=\Omega^\textnormal{lin}$ as defined in~\eqref{eq:SelectedFrequencyMonostable}.}
    \label{fig:SpecsOfDownMagnetization}
\end{figure}

As usual in spatial dynamics, multiple spatial eigenvalues at the origin indicate the bifurcation of coherent structures. Indeed, one readily verifies that $\Omega^\textnormal{lin}=(s^\textnormal{lin})^2/2+\b^-/\alpha$ is the degenerate \emph{center case} in \cite{siemer2020} that implies special properties for the `spatial' coherent structure ODE~\eqref{eq:SphericalCoherentStructureODE}. In case $\Omega=\Omega^\textnormal{lin}$, it was proven (see~\cite[Theorem 2]{siemer2020} for full statement) that if $|\cc|\ll1$ then a family of inhomogeneous DWs (flat as well as non-flat ones) co-exist.

We again note that the linear spreading speed and linear spreading frequency of DWs in the monostable parameter regime were determined in~\cite{goussev2020dynamics} for the case $\cc=0$ utilizing a saddle-point approximation. While this is generally valid for scalar reaction-diffusion equations and the~\eqref{eq:ComplexGinzburgLandauEquation} equation, the saddle-point approach may fail for systems of PDEs, as already conjectured in~\cite{ebert2000front}. Roughly speaking, each component leads to a candidate for the spreading speed and possibly additional ones in case of (linear) coupling, and one has to pick the right one~\cite{ebert2000front, van2003front}. However, the situation of~\eqref{eq:LLGS}, i.e., the computation of the linear spreading speed and frequency associated to $\calL^-$, is analogous to the~\eqref{eq:ComplexGinzburgLandauEquation} case, and both approached yields the same prediction. Heuristically, this is not surprising since~\eqref{eq:LLGS} can be written as a generalized~\eqref{eq:ComplexGinzburgLandauEquation} by pointwise stereographic projection, cf. e.g.~\cite{Melcher2017}. Here we have generalized this to operators $\calL$ that correspond to a general scalar complex equation, which means the dispersion relation~\eqref{eq:GeneralDispersionRelation} has the form $d(\lambda,\nu)=d_1^2(\lambda,\nu) + d_2^2(\nu)$ with polynomials $d_1$ and $d_2$. As shown above, it follows that there is only one candidate for the linear spreading speed (in absolute value) as well as frequency. Beside these observations, we again remark that the study of the pointwise Green's function, and in particular the locations of its singularities, can be generalized from scalar equations to systems of PDEs in a straightforward way. Moreover, the description of the absolute spectrum via weighted spectral curves can possibly be generalized to even more complicated equations, see e.g.~\cite{faye2020remnant}.

\subsection{Relation to homogeneous domain walls}\label{sec:RelationHomogeneousDW}
We return to the full nonlinear~\eqref{eq:LLGS} and the question whether an initial interface will propagate asymptotically with $s^\textnormal{lin}$ and $\Omega^\textnormal{lin}$ or not. For this, we first focus on the situation of zero spin polarization, $\cc=0$, and thus assume $\beta/\alpha<h$. We further select the positive speed in~\eqref{eq:HomogeneousSpeed}, denoted by $s_+^\textnormal{hom}$, by taking the unstable state $-\3$ on the right side. In this parameter regime, $s_+^\textnormal{hom}$ equals $s^\textnormal{lin}$ for
\begin{equation}\label{eq:AbsoluteSpectrumAppliedField}
h=h_\pm^s\coloneqq\frac{\beta}{\alpha}-2\mu-\frac{2\mu}{\alpha^2}\pm\frac{2\mu}{\alpha^2}\sqrt{1+\alpha^2},
\end{equation}
where $h^s_+<h^s_-$, since we consider the easy-axis case $\mu<0$. Concerning the boundary between the bistable and monostable regimes given by $\beta/\alpha-\mu$, i.e., the left boundary of $\M^+$ in case $\cc=0$, one readily verifies that $\beta/\alpha-\mu<h^s_+$. This means that the linear spreading speed is smaller than the speed of $\m_0$ for $h\in(\beta/\alpha-\mu,h^s_+)$, so that $\m_0$ may be selected rather than a DW with linear spreading speed~\eqref{eq:SelectedSpeedMonostable}. Indeed, we verify numerically in \S~\ref{sec:PointSpectrumNumerics} that $\m_0$ has stable point spectrum up to $h^s_+$, which means that in this regime $\m_0$ is convectively unstable only, since the weighted essential spectrum is stable by the above analysis. Moreover, also in direct numerical simulations, we observe the selection of $\m_0$ from a localized perturbation around $-\3$, so that $\m_0$ would be a pushed front in the applied field regime $h\in(\beta/\alpha-\mu,h^s_+)$, see \S~\ref{sec:NumericalSimulationMethod} for details.

Analogously, we compute that~\eqref{eq:HomogeneousFrequency} is equal to $\Omega^\textnormal{lin}$ in case
\[h=h^\Omega\coloneqq \frac{\beta}{\alpha}-2\mu.\]
Since $\alpha>0$, it follows that $h_\pm^s\neq h^\Omega$ for all $\beta\ge 0$ and $\mu<0$, and thus $h^s_+<h^\Omega<h^s_-$. We summarize these findings as follows.
\begin{corollary}\label{cor:DifferenceToHomogeneousDW}
In the $\M^+$ regime in case $\cc=0$, a DW that propagates with the linear spreading speed and frequency given by~\eqref{eq:SelectedSpeedMonostable} and~\eqref{eq:SelectedFrequencyMonostable}, respectively, differs from $\m_0$ given by~\eqref{eq:AnalyticHomogeneousSolution}. In particular, such a DW is non-flat, due to~\cite[Theorem 2]{siemer2020}. By symmetry, this also holds in the $\M^-$ regime given by~\eqref{eq:MonostableSetMinus}.
\end{corollary}

\section{Numerical Studies}\label{sec:Numerics}
In this section, we study on the one hand the location of the point spectrum of the linear operator arising from linearization around the explicitly known family of homogeneous DWs~\eqref{eq:AnalyticHomogeneousSolution} and its dependence on the applied field. On the other hand, we describe the method of numerical time integration of~\eqref{eq:LLGS} used in this paper, and also the method of freezing from which the long-time behavior within the bistable as well as monostable parameter regime were obtained.

\subsection{Point Spectrum}\label{sec:PointSpectrumNumerics}

We are interested in the location of the point spectrum of $\m_0$ and thus set $\cc=0$. Due to symmetry, we only consider the case $\beta/\alpha<h$ and hence $\sigma=1$ in~\eqref{eq:AnalyticHomogeneousSolution}. The associated speed and frequency is given by~\eqref{eq:HomogeneousSpeed} as well as~\eqref{eq:HomogeneousFrequency}.

Stability of $\m_0$ was numerically verified in~\cite{gou2011stability} for $h<\beta/\alpha-\mu$, where the aforementioned energy difference $\Delta E_\eps$ was used. For applied fields beyond $\beta/\alpha-\mu$, however, $\m_0$ is unstable in $L^2$ -- also in case of stable point spectrum -- due to the presence of unstable essential spectrum. As discussed in Sections~\ref{sec:Pointwise growth, absolute spectra, and double roots} and~\ref{sec:RelationHomogeneousDW}, in the convectively unstable regime $|h|\notin(h^s_+,h_-^s)$, the essential spectrum is stable in the exponentially weighted spaces $L^2_\eta$ for suitable $\eta<0$. In particular, the absolute spectrum $\Sabs$ is stable, and in the following we distinguish the parameter regimes $\beta/\alpha-\mu<h<h^s_+$ and $h^s_-<h$.

Numerical computations of the point spectrum were performed using the \MATLAB toolbox \STABLAB . This toolbox uses the so-called \emph{Evans function} to numerically locate eigenvalues, and we refer to~\cite{barker2018evans, barker2009stablab} for further details. For implementation in \STABLAB , we rewrite~\eqref{eq:LinearEigenvalueProblem}, including exponential weights, as a first order system of the form
\begin{equation}\label{eq:WeightedFirstOrderSystem}
W'=A_\eta(\xi;\lambda)W,
\end{equation}
where $A_\eta(\xi;\lambda)$ is defined in the Appendix. The fixed parameters are $\alpha=1$, $\beta=3/4$, and $\mu=-1$, which yields $\beta/\alpha-\mu=1.75$ and $h^s_+= 1.92$. The domain is set to $\xi \in [-100,100]$ and since we know the locations of $\Sess$, $\Sess^\eta$, and $\Sabs$ explicitly (see \S~\ref{sec:StabilitySteadyStates} and~\ref{sec:Pointwise growth, absolute spectra, and double roots}), we can choose an integration contour that excludes zero eigenvalues which arise from the translation and rotation invariances.

Within the first parameter region, we take the samples $h\in\lbrace 1.75, 1.8, 1.85, 1.9\rbrace$ with $\eta\in\lbrace 0, -0.1, -0.15, -0.25 \rbrace$. Due to the relatively small range, $1.75\le h\le 1.9$, the individual Evans function images of the chosen contour do not vary significantly. Thus, we only illustrate one exemplary Evans function output for $h=1.9$ with $\eta=-0.25$ in Figure~\ref{fig:EvansFunctionOutput} (b). Here we utilize the normalization option within \STABLAB for illustration purposes. In all these cases, the winding number is zero, demonstrating stability of the point spectrum of $\m_0$ for $h<h^s_+$ in $L^2_\eta$.

For the second region, we choose $h\in\lbrace 8, 15, 20\rbrace$ and $\eta\in\lbrace -1.5, -1.2, -1.1 \rbrace$. The images of the same contour as before (cf.\ Figure~\ref{fig:EvansFunctionOutput} (a)) under the resulting Evans functions are illustrated for $h=8$ as well as $h=20$ in Figure~\ref{fig:EvansFunctionOutputLargeH}. In all cases, the winding number is non-zero, indicating unstable point spectrum. For an applied field in this region, i.e., beyond $h^s_-$, we also observe numerically in \S~\ref{sec:NumericalSimulationMethod} that not $\m_0$, but a critical pulled front is dynamically selected (cf.\ Figure~\ref{fig:SpeedAndFrequencyComparison}). This selection behavior persists for $\cc\neq 0$ (cf.\ Figure~\ref{fig:SelectedSpeedAndFrequency}), where pushed DWs are asymptotically selected also for small applied fields, and where pulled DWs are selected for $h$ beyond a certain threshold. It is clear that this threshold lies above $\b^-/\alpha-\mu$, but more specific analytic results do not seem to be known.

\begin{figure}
    \centering
    \begin{subfigure}[b]{0.37\textwidth}
    \includegraphics[trim=50 0 90 30, clip, width=\textwidth]{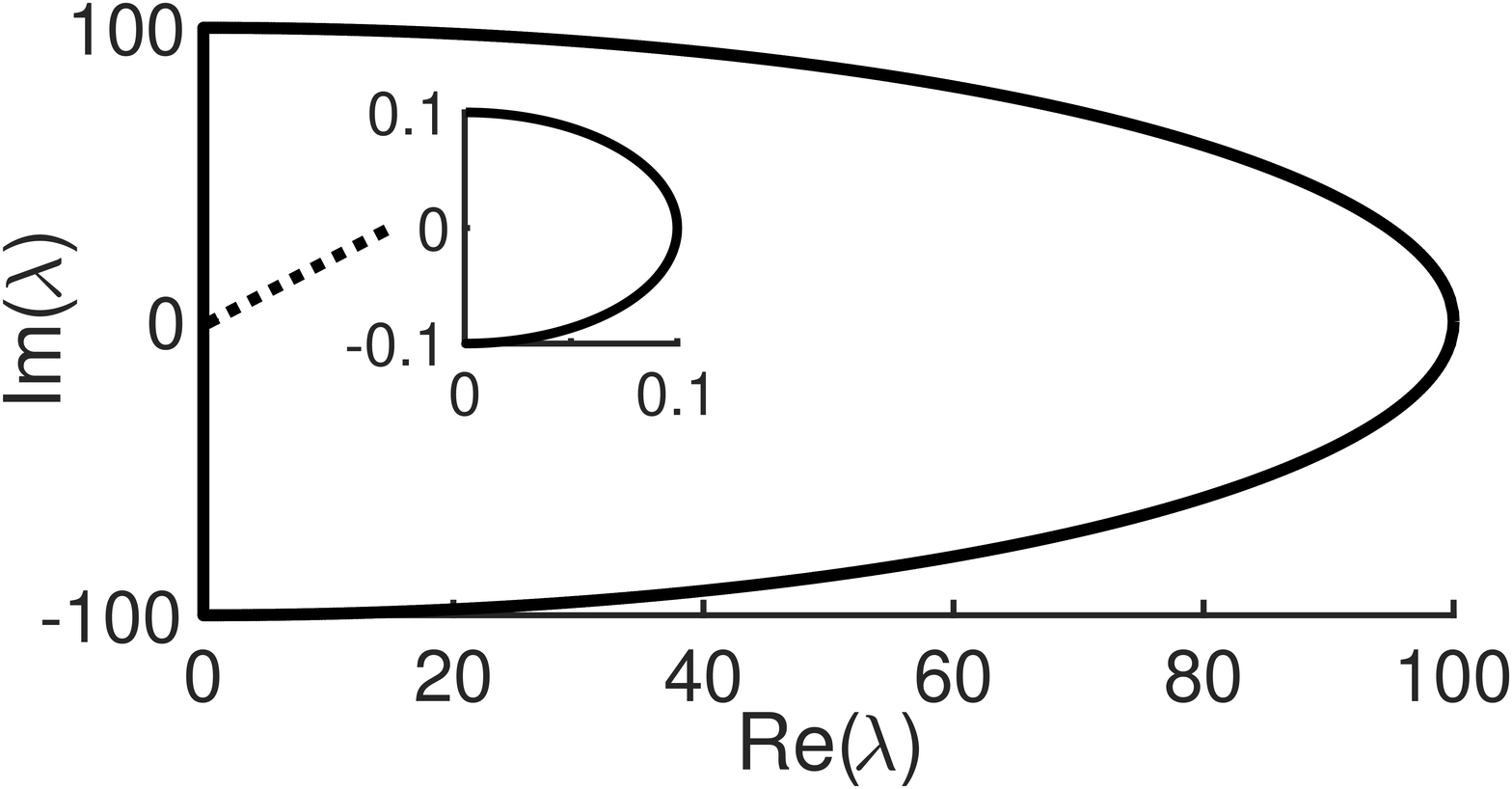}
    \caption*{(a)}
    \end{subfigure}
    \hspace*{10mm}
    \begin{subfigure}[b]{0.37\textwidth}
    \includegraphics[trim=50 30 120 30, clip, width=\textwidth]{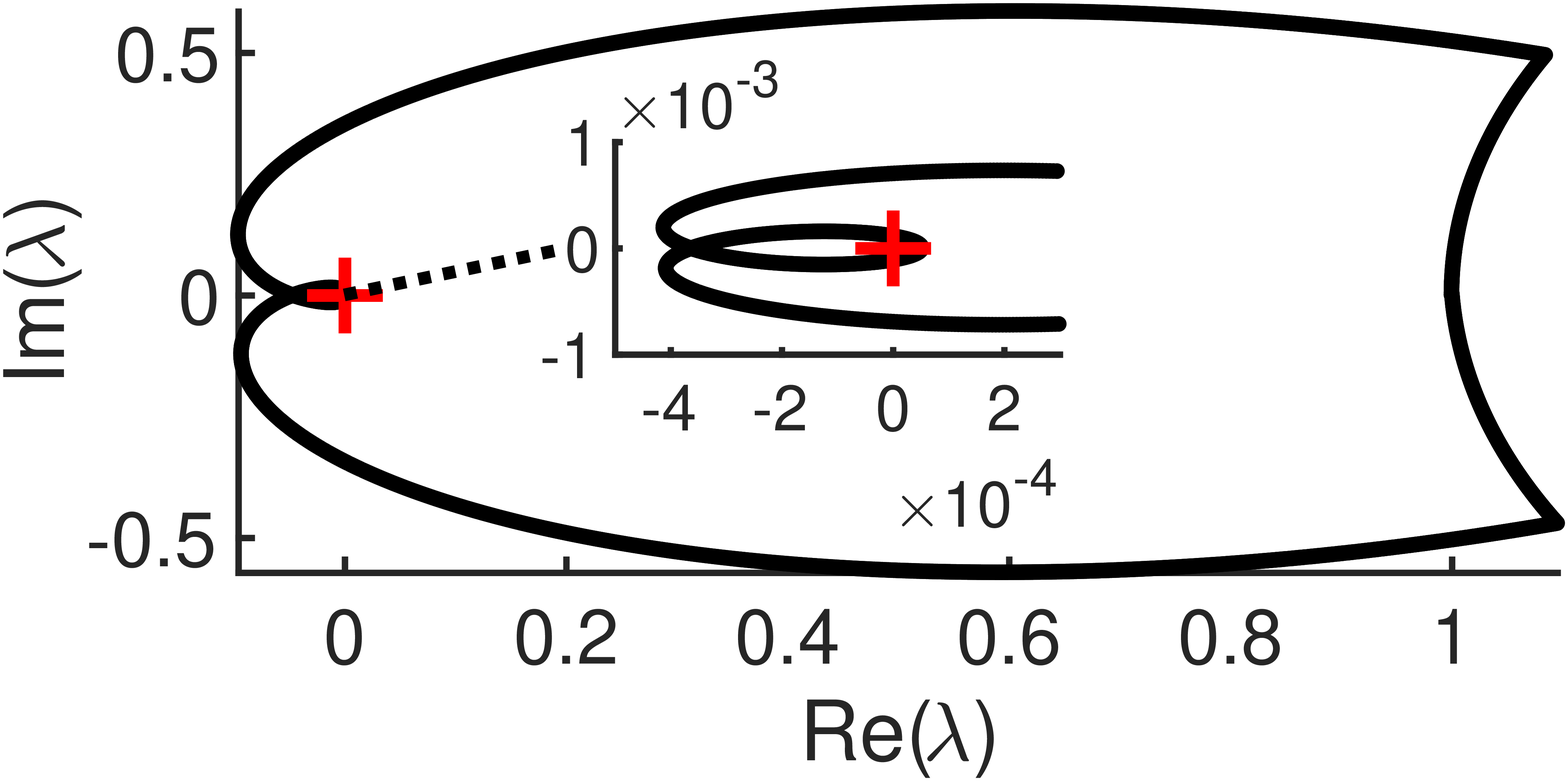}
    \caption*{(b)}
    \end{subfigure}
    \caption{Shown are the Evans-function computations for~\eqref{eq:WeightedFirstOrderSystem} with $\alpha=1, \beta=0.75, \mu=-1, h=1.9, \eta=-0.29, s=0.575,$ and $\Omega=1.325$. Semi-circular contour of radius $100$ in (a) using $1500$ mesh points, excluding the half-circle of radius $0.1$ around the origin shown in the inset. The normalized image of the contour under the Evans function is shown in (b), where the red cross indicates the origin and the inset zooms in around the origin. The winding number is zero.}
    \label{fig:EvansFunctionOutput}
\end{figure}

\begin{figure}
    \centering
    \begin{subfigure}[b]{\textwidth}
    \centering
    \includegraphics[trim=20 0 100 30, clip, width=0.3\textwidth]{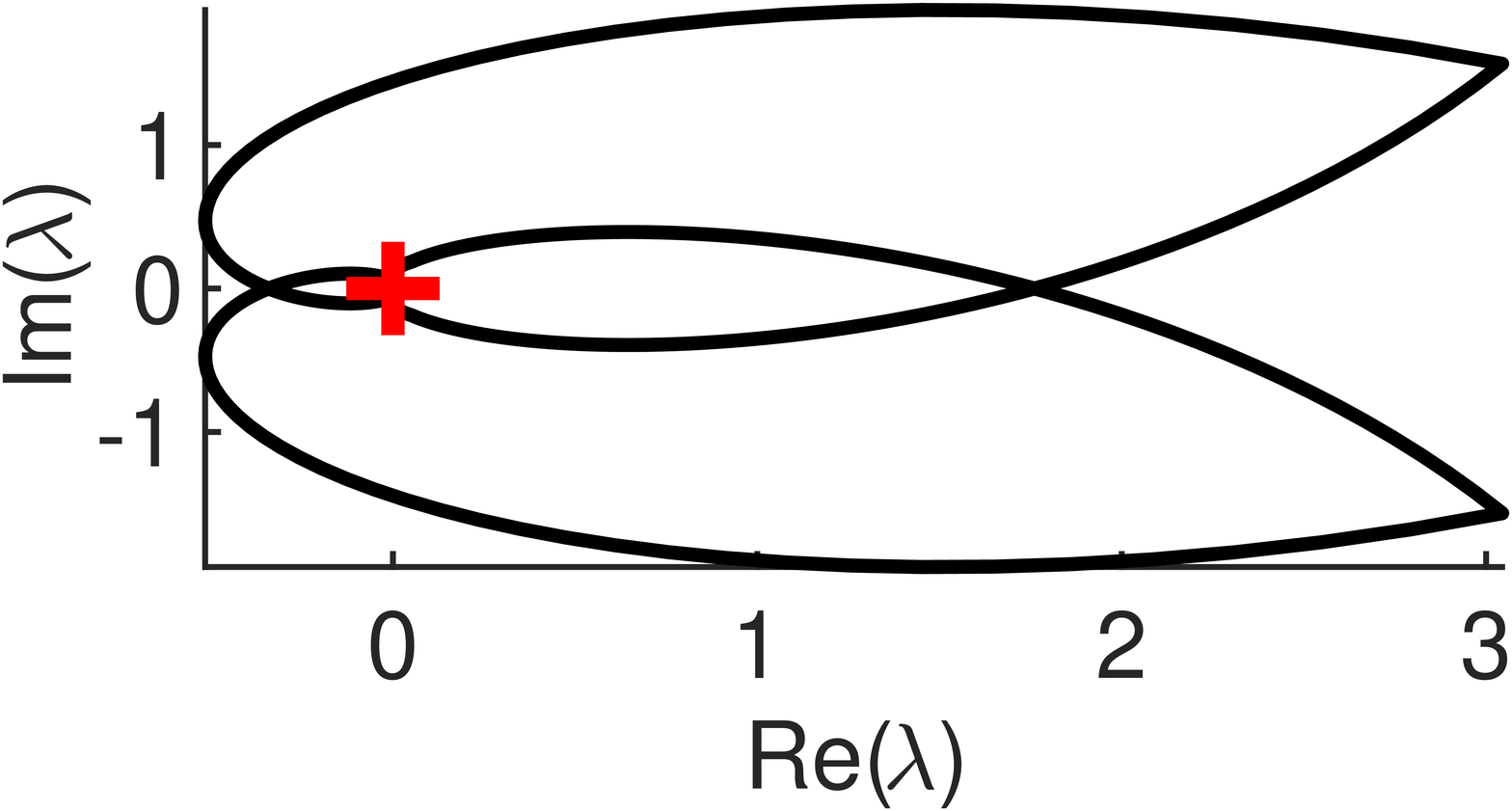}
    \hfill
    \includegraphics[trim=20 0 100 30, clip, width=0.3\textwidth]{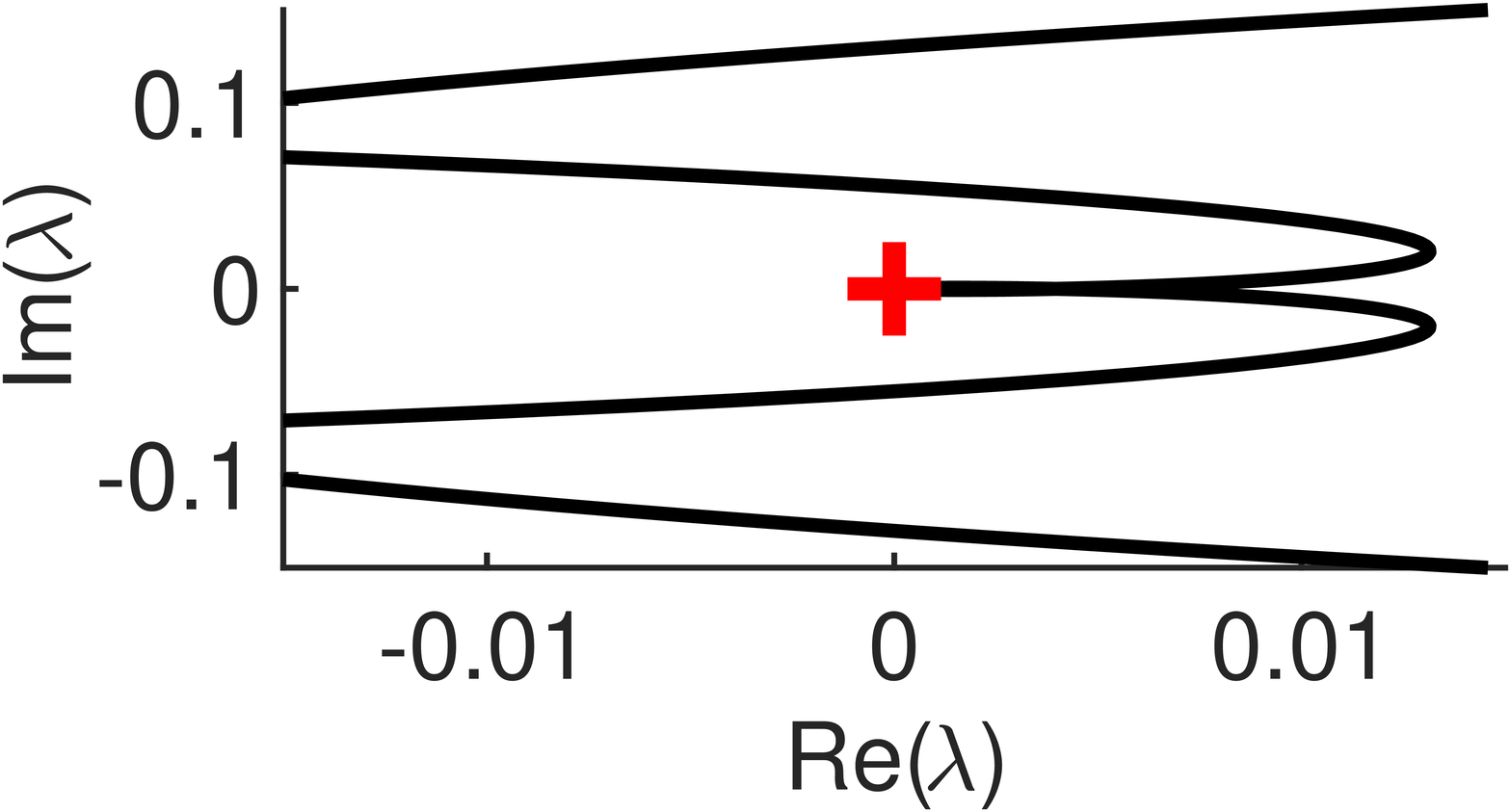}
	\hfill
    \includegraphics[trim=20 0 100 30, clip, width=0.3\textwidth]{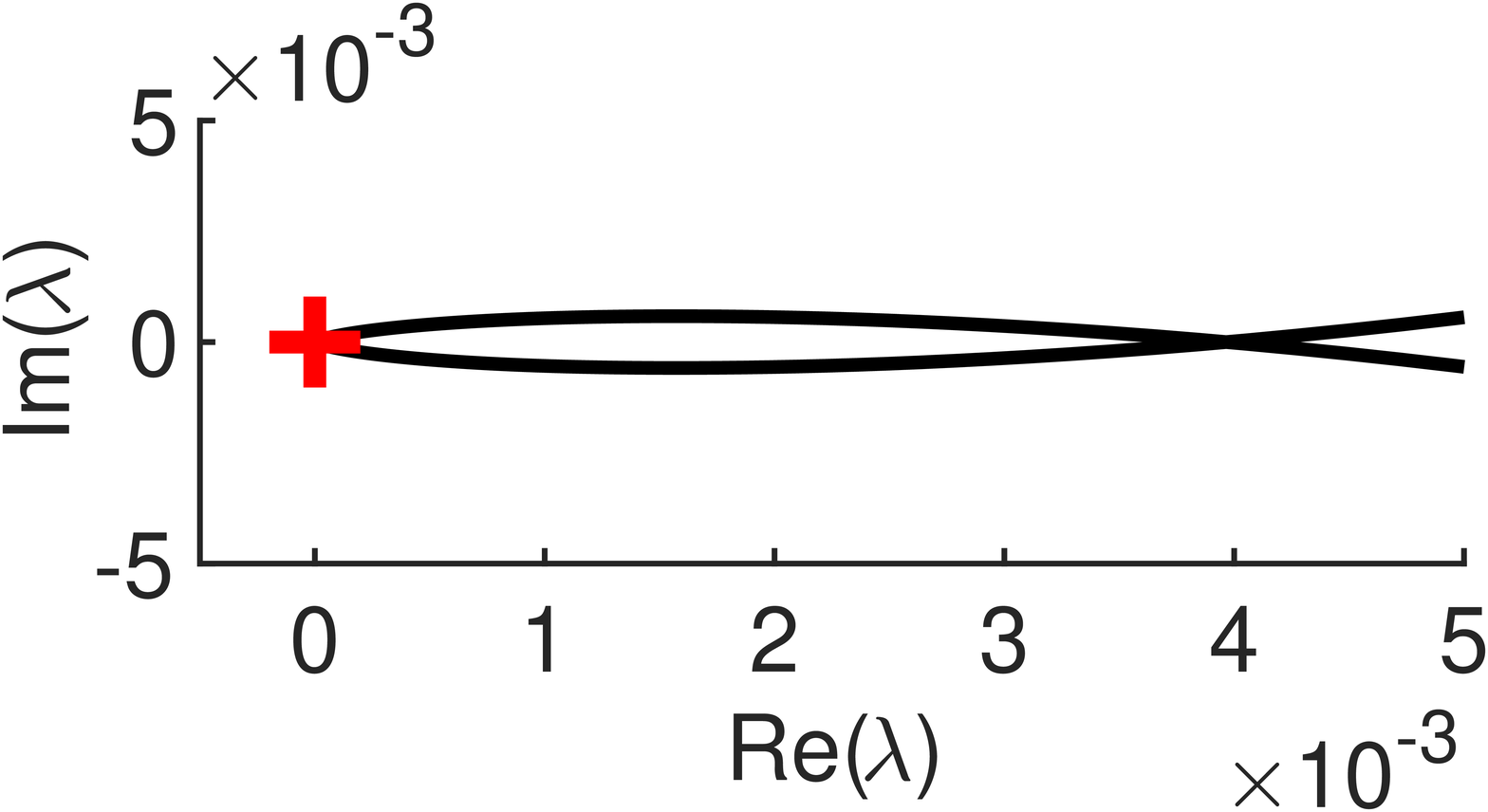}    
    \caption*{$h=8$}
    \end{subfigure}
    \vskip\baselineskip
    \begin{subfigure}[b]{\textwidth}
    \centering
    \includegraphics[trim=20 0 100 30, clip, width=0.3\textwidth]{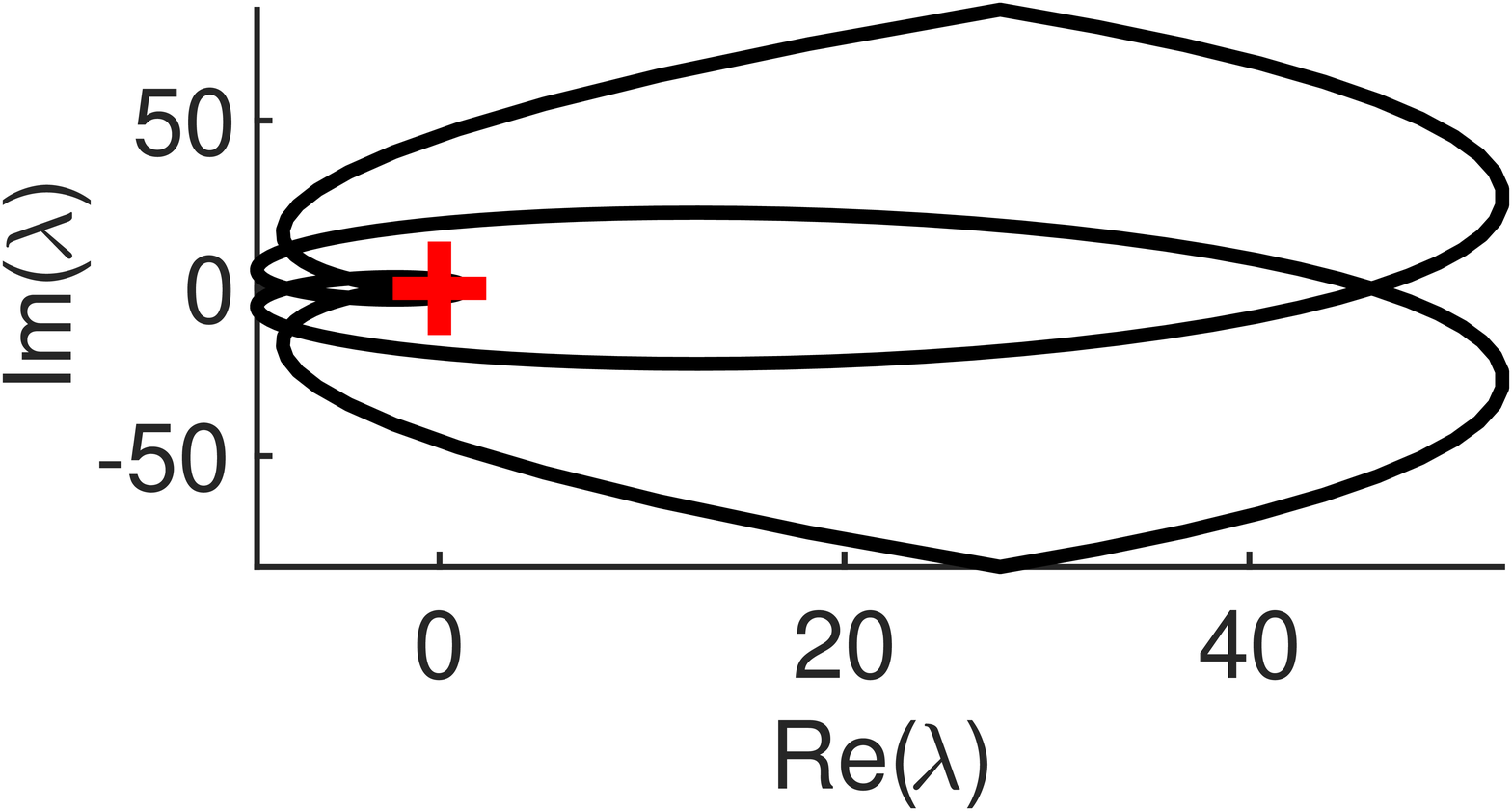}
    \hfill
    \includegraphics[trim=20 0 100 30, clip, width=0.3\textwidth]{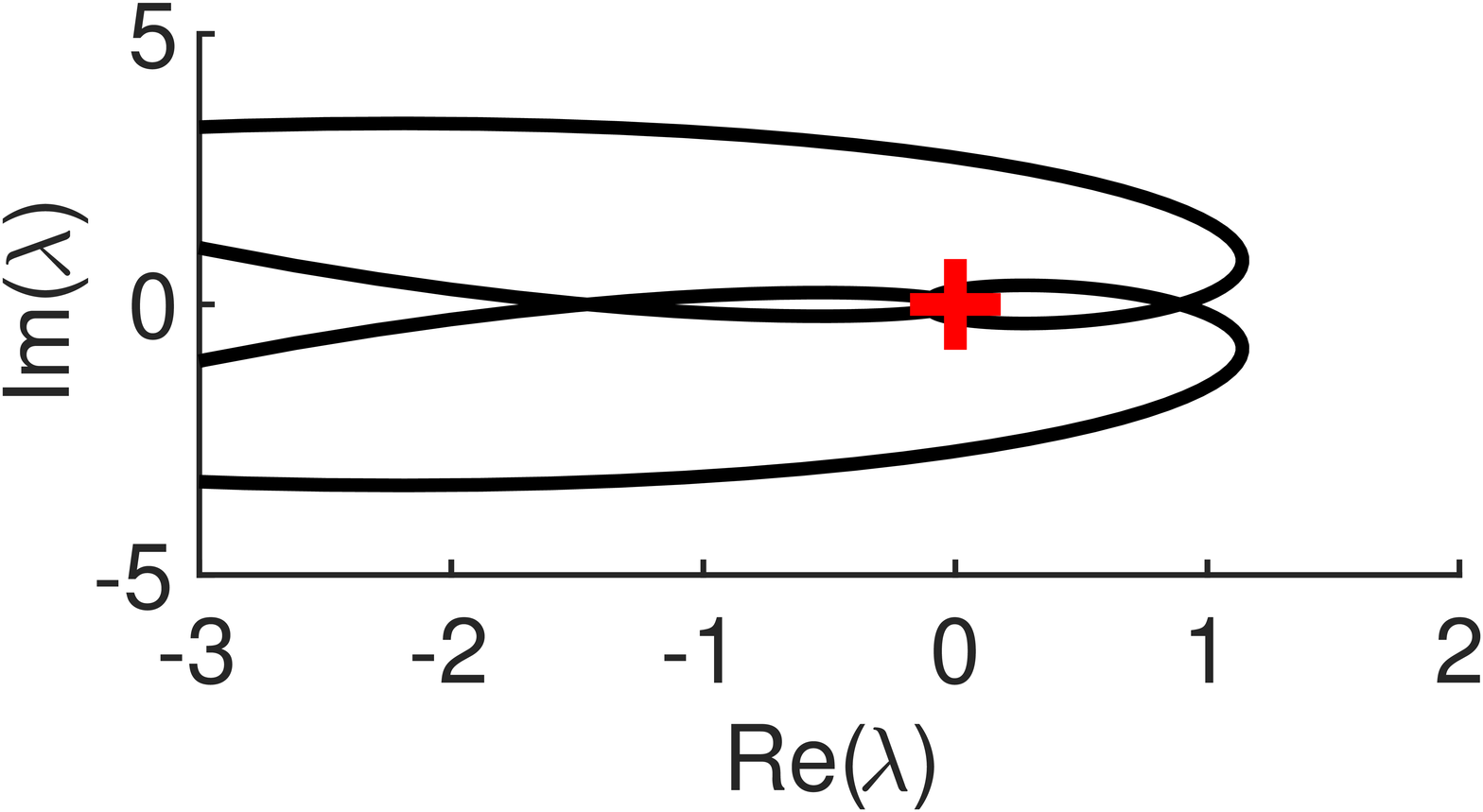}
	\hfill
    \includegraphics[trim=20 0 100 30, clip, width=0.3\textwidth]{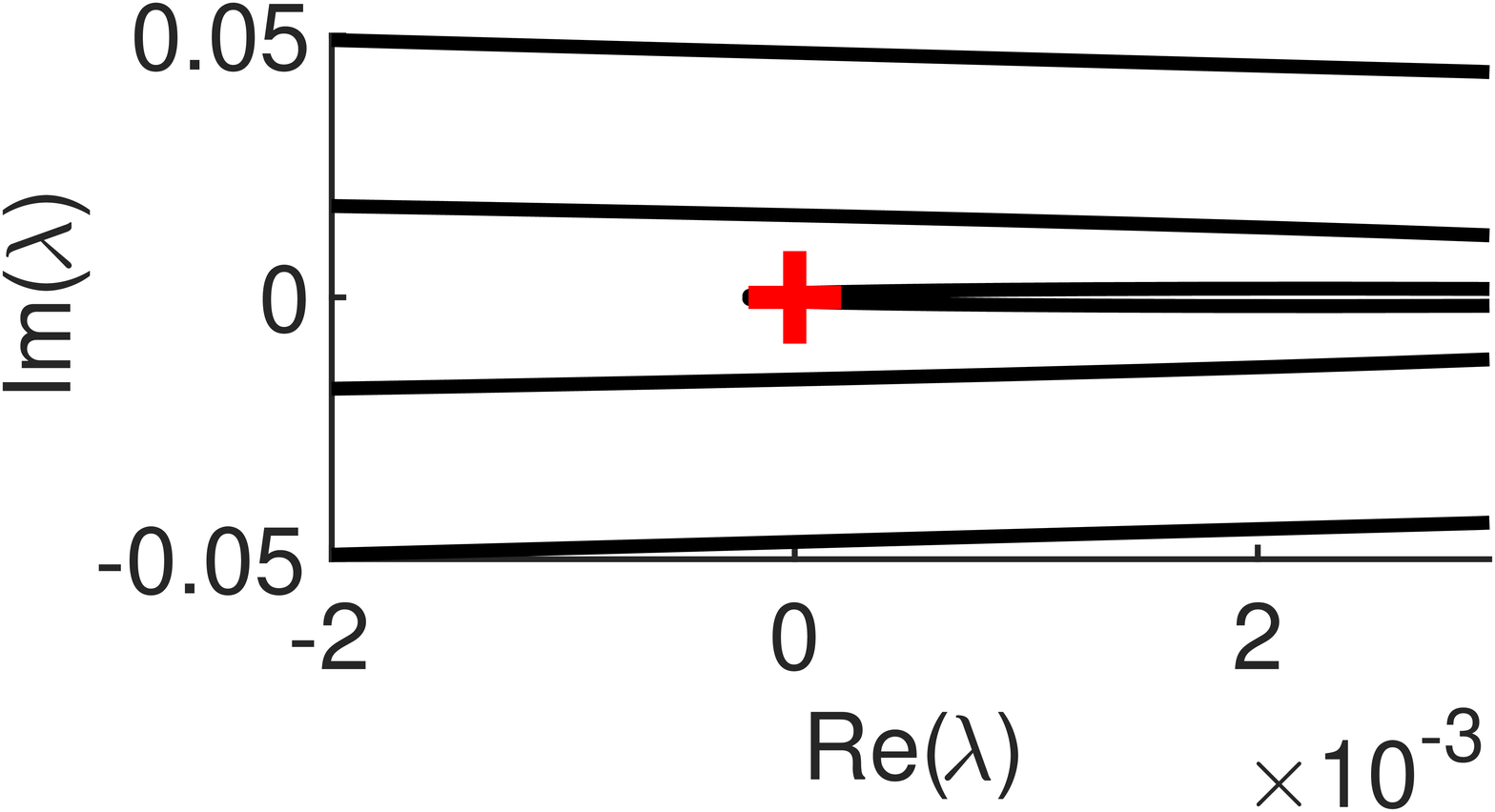}    
    \caption*{$h=20$}
    \end{subfigure}
    \caption{The left panels show the images of the contour Figure~\ref{fig:EvansFunctionOutput} (a) under the Evans function for $h=8$ (upper) and $h=20$ (lower). The middle and right panels show enlargements of the region around the origin, which is indicated by a red cross. Other parameters are as in Figure~\ref{fig:EvansFunctionOutput}. The winding number is six for $h=8$, and fourteen for $h=20$.}
    \label{fig:EvansFunctionOutputLargeH}
\end{figure}

\begin{Remark}
Since we are interested in the location of the point spectrum of solutions on the unit sphere, one could also start from linearizing~\eqref{eq:SphericalCoherentStructurePDE} around $\m_0$ in spherical coordinates which reads $(\theta_0,\phi_0)=(2\arctan(\exp(\smu\xi)),0)$. However, due to the singular behavior of the term $\sin^{-1}(\theta_0)$ for $|\xi|\to\infty$, convergence is not guaranteed and \STABLAB cannot be applied in a straightforward way. Therefore, we chose to study the location of the point spectrum in the original coordinates.
\end{Remark}

\subsection{Numerical simulation method and freezing}\label{sec:NumericalSimulationMethod}~

Before discussing our results from numerical simulations, we comment on the methods used. The results of numerical time integration of~\eqref{eq:LLGS} rely on a semi-implicit method with a finite element scheme in space from the \MATLAB package \textsc{Pde2Path}~\cite{dohnal2014pde2path}. In order to capture the spreading speeds and frequencies throughout the simulations, we extend the implementation of the `freezing' method from~\cite{BeynThuemmler2004} as discussed in~\cite{rademacher2017symmetries}, such that both symmetries are tracked. Briefly, here we first project onto the two generators and secondly solve the two component equation system for the current translation as well as current rotation, both within each time step. In order to ensure $\m \in \SS^2$, in each time step we scale the resulting predictor pointwise onto the unit sphere. In this way, we perform simulations with a typical domain size $[-50,50]$ with grid spacing $\Delta_x=\num{e-2}$, homogeneous Neumann boundary conditions. The final time within the bistable parameter regime is $t=50$, and $t=100$ within the monostable parameter regime, both with step size $\Delta_t=\num{e-4}$. The final time is increased for $h$ close to $h^s_+$. The initial conditions within the bistable parameter regime are of step function-type, connecting the state $+\3$ on the left and $-\3$ on the right, whereas the initial conditions in the monostable parameter regime are perturbations locally around the unstable state $-\3$. Parameters are again $\alpha=1, \beta=0.75$, and $\mu=-1$. For an applied field less than $h^\Omega=2.75$ and close to $h^s_+=1.92$, we observe modulated fronts which are characterized by oscillating frozen speed and frequency (cf.\ Figure~\ref{fig:ModulatedFront}). Although these oscillations mostly decay for a final time extended to $t=500$, the decay significantly slowed down for an applied field quite close to $h^s_+$. Therefore, we do not plot a few points in the figures for which there seems to be no convergence at $t=500$.

Our results from these numerical simulations confirm the analytical predictions and also the persistence for $|\cc|\ll 1$, even for $\cc$ far away from zero. The results for $\cc=0$ are presented in Figure~\ref{fig:SpeedAndFrequencyComparison} and for $\cc=\pm 0.5$ in Figure~\ref{fig:SelectedSpeedAndFrequency}.

The frozen speeds and frequencies in case $\cc=0$, but also for $\cc\neq 0$, approach the linear predictions~\eqref{eq:SelectedSpeedMonostable} and~\eqref{eq:SelectedFrequencyMonostable} from below for $h^\Omega< h$, where no modulated fronts are observed. We suspect that this is a consequence of the logarithmic correction, which is known to be associated with the spreading speed selection in scalar reaction-diffusion equations (see for example~\cite{ebert2000front, hamel2013}).
In our numerical results also the spreading frequency appears to be approached monotonically from below, which we suspect is due to the functional relation of spreading speed and frequency. For $h^s_+<h<h^\Omega$, this behavior does not persists in a strict sense, due to the aforementioned oscillations in $s$ and $\Omega$ (see Figure~\ref{fig:ModulatedFront} for an example).

\begin{figure}
\centering
\begin{tabular}{cc}
  \begin{minipage}[c]{.33\textwidth}
  \centering
  \includegraphics[trim=0 0 0 150, clip, width=\textwidth]{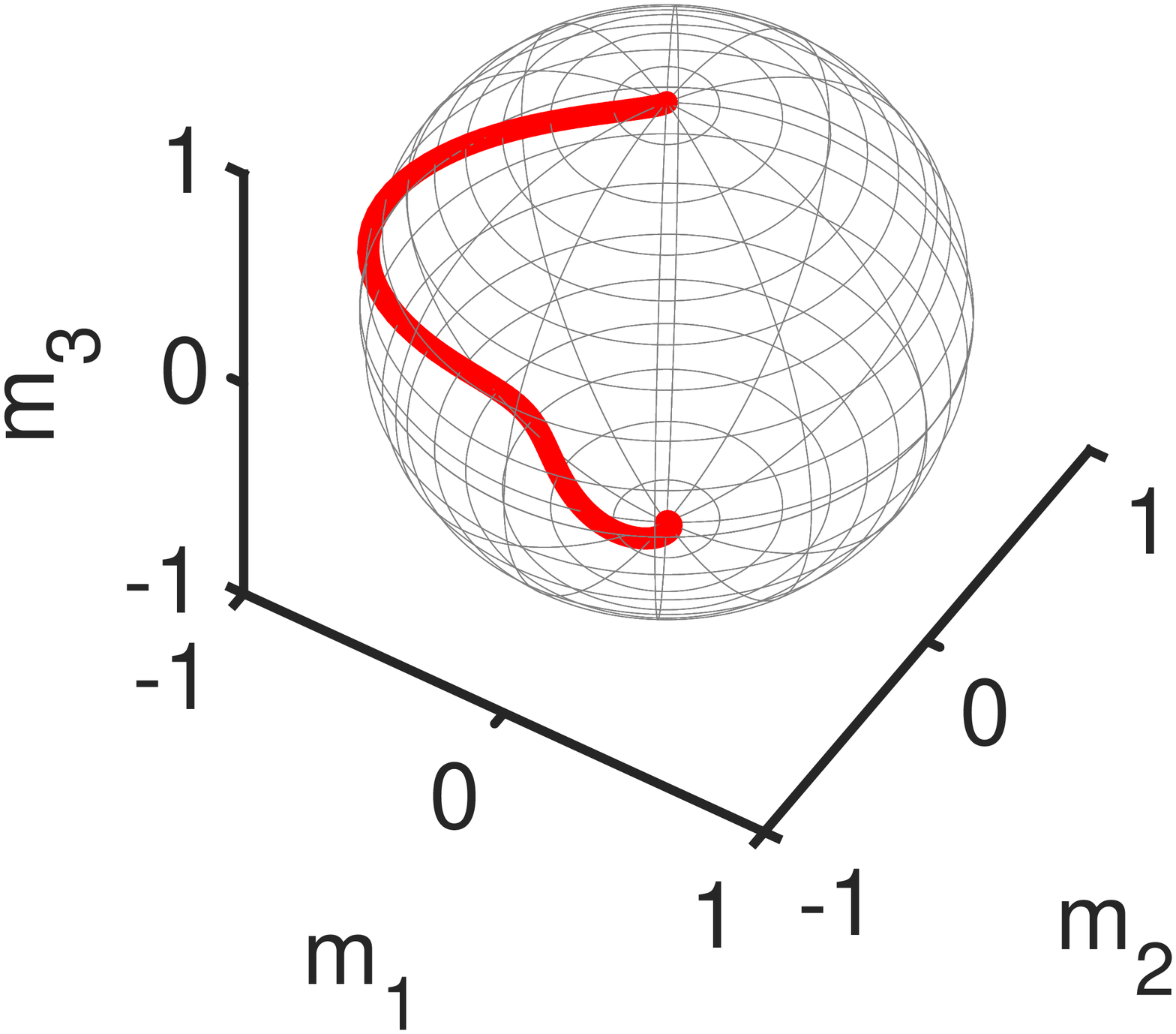}
\end{minipage}
\begin{minipage}[c]{.52\textwidth}
  \begin{subfigure}{\textwidth}
  \centering
  \includegraphics[trim=0 0 10 0, clip,width=0.6\textwidth]{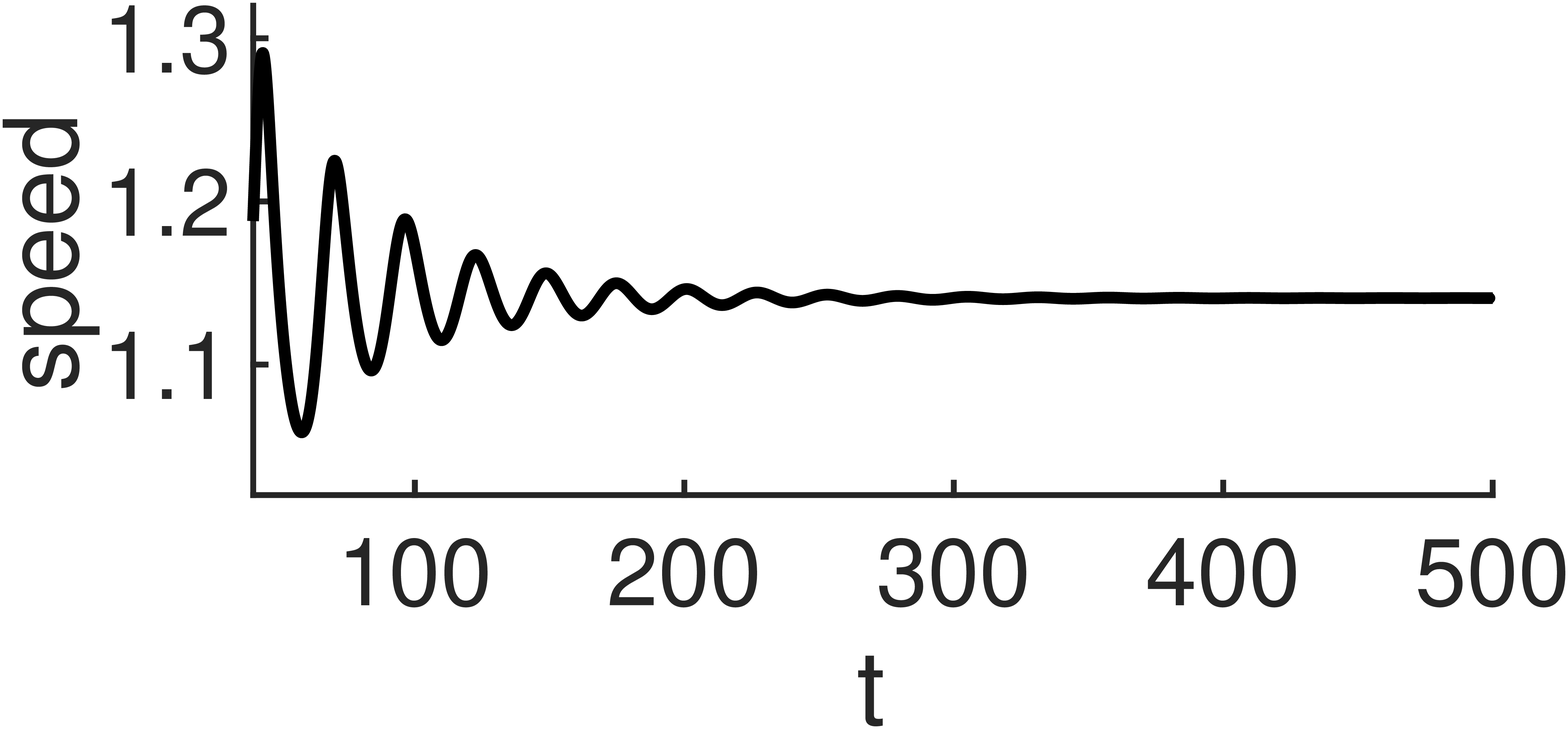}
  \end{subfigure}
  
  \vspace*{0mm}
  \begin{subfigure}{\textwidth}
  \centering
  \includegraphics[width=0.6\textwidth]{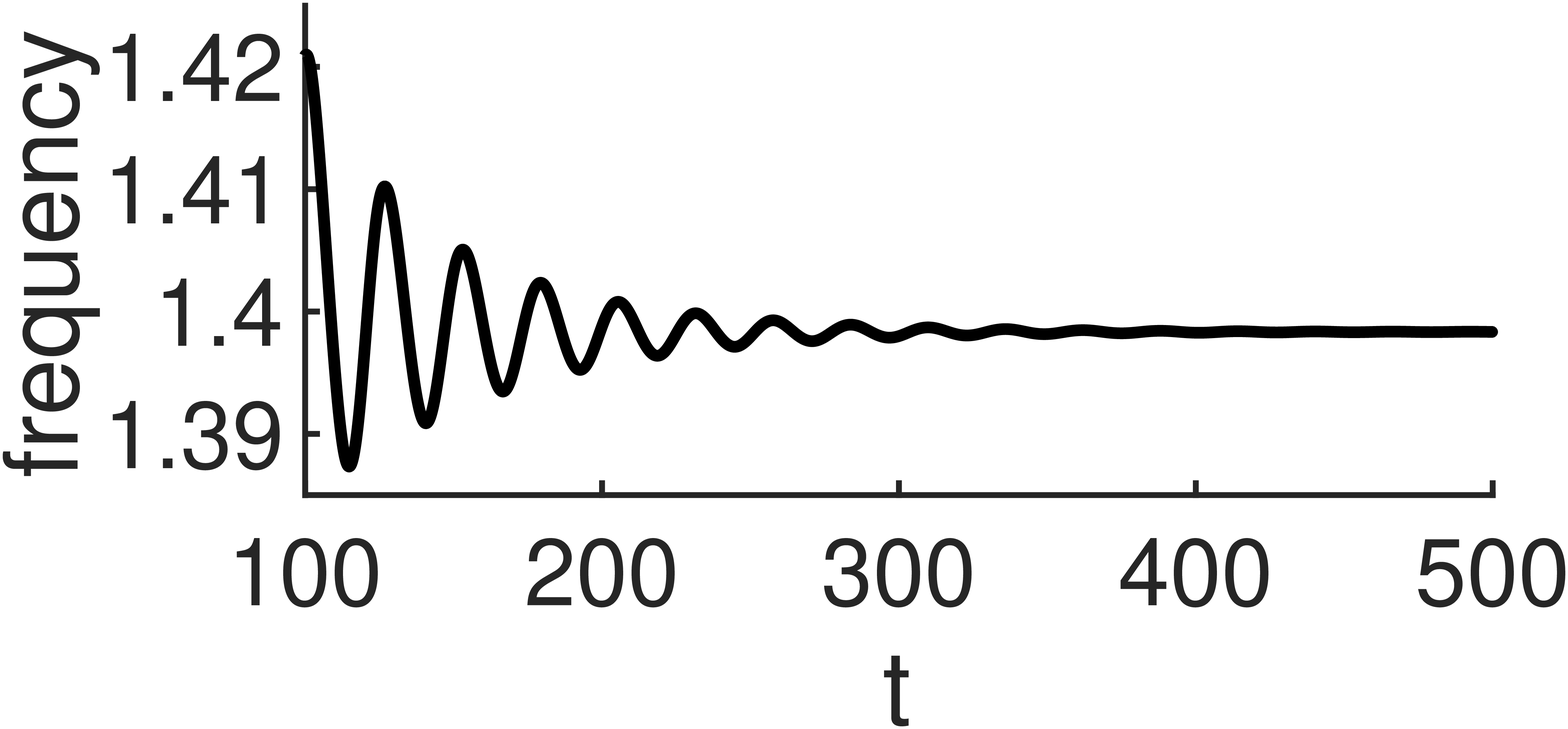}
  \vspace*{-5mm}
  \end{subfigure}
\end{minipage}

\end{tabular}
\vspace*{5mm}
\caption{Left panel: DW at final time $t=500$ projected onto $\SS^2$. Right panels: temporal evolution of the speed (upper) and frequency (lower) for $t\in[100,500]$. Parameters are $h=2.4$ and $\cc=0$, other parameters are as in Figure~\ref{fig:StabilityRegionsUp-AndDown-Magnetization}. Linear predictions are $s^\mathrm{lin}=1.14$ as well as $\Omega^\mathrm{lin}=1.4$.}
\label{fig:ModulatedFront}
\end{figure}

For $h<\beta/\alpha-\mu$, we find that $\m_0$ is selected, which is expected due to stable spectrum (see \S~\ref{sec:StabilitySteadyStates} and \S~\ref{sec:PointSpectrumNumerics}). For $\beta/\alpha-\mu<h<h^s_+$ in the monostable regime, a nonlinear selection mechanism appears to dominate the linear spreading speed, which is plausible due to stable point spectrum numerically observed in \S~\ref{sec:PointSpectrumNumerics}. This is a well-known phenomenon in scalar reaction diffusion equations, cf.\ e.g.~\cite{bramson1983convergence, ebert2000front, garnier2012inside, kirchgassner1992nonlinear, van2003front}, and so these DWs appear to be pushed fronts, as mentioned before. For $h^s_+<h$, fronts with speed and frequency given by the linear predictions, i.e., $s^\mathrm{lin}$ and $\Omega^\mathrm{lin}$, are numerically selected over time. These dynamically selected DWs must differ from $\m_0$ due to Corollary~\ref{cor:DifferenceToHomogeneousDW}. This behavior persists also for $\cc$ far away from zero (cf.\ Figure~\ref{fig:SelectedSpeedAndFrequency}).

\begin{figure}[h]
    \centering
    \includegraphics[trim=50 0 0 0, clip, width=0.8\textwidth]{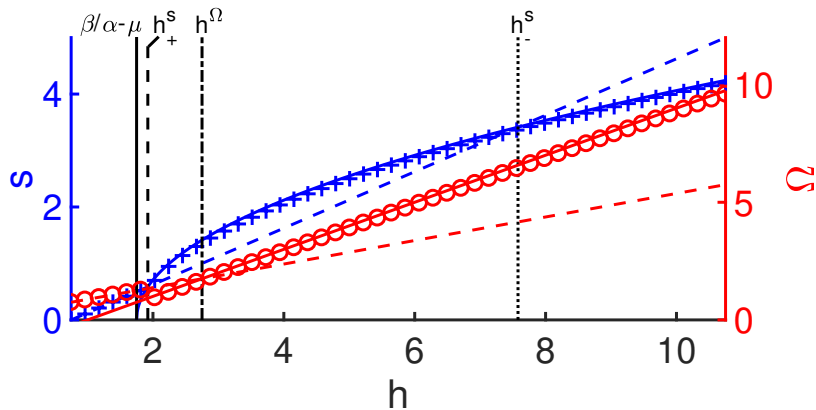}
    \caption{Speed and frequency selection for $h\in(\beta/\alpha, \beta/\alpha-10\mu)$ in case $\cc=0$. Other parameters are as in Figure~\ref{fig:StabilityRegionsUp-AndDown-Magnetization}. Speed and frequency of $\m_0$, cf.~\eqref{eq:HomogeneousSpeed} (dashed blue) as well as~\eqref{eq:HomogeneousFrequency} (dashed red), and linear spreading speed~\eqref{eq:SelectedSpeedMonostable} (solid blue) and spreading frequency~\eqref{eq:SelectedFrequencyMonostable} (solid red). Numerically observed (asymptotic) speeds and frequencies are indicated by blue crosses and red circles, respectively. Vertical lines mark $h=\beta/\alpha-\mu=1.75$ (solid), $h^s_+=1.92$ (dashed), $h^\Omega=2.75$ (dotted dashed), and $h^s_-=7.58$ (dotted).}
    \label{fig:SpeedAndFrequencyComparison}
\end{figure}

\begin{figure}[h]
\centering
   	\begin{subfigure}[b]{0.49\textwidth}
		\includegraphics[width=\textwidth]{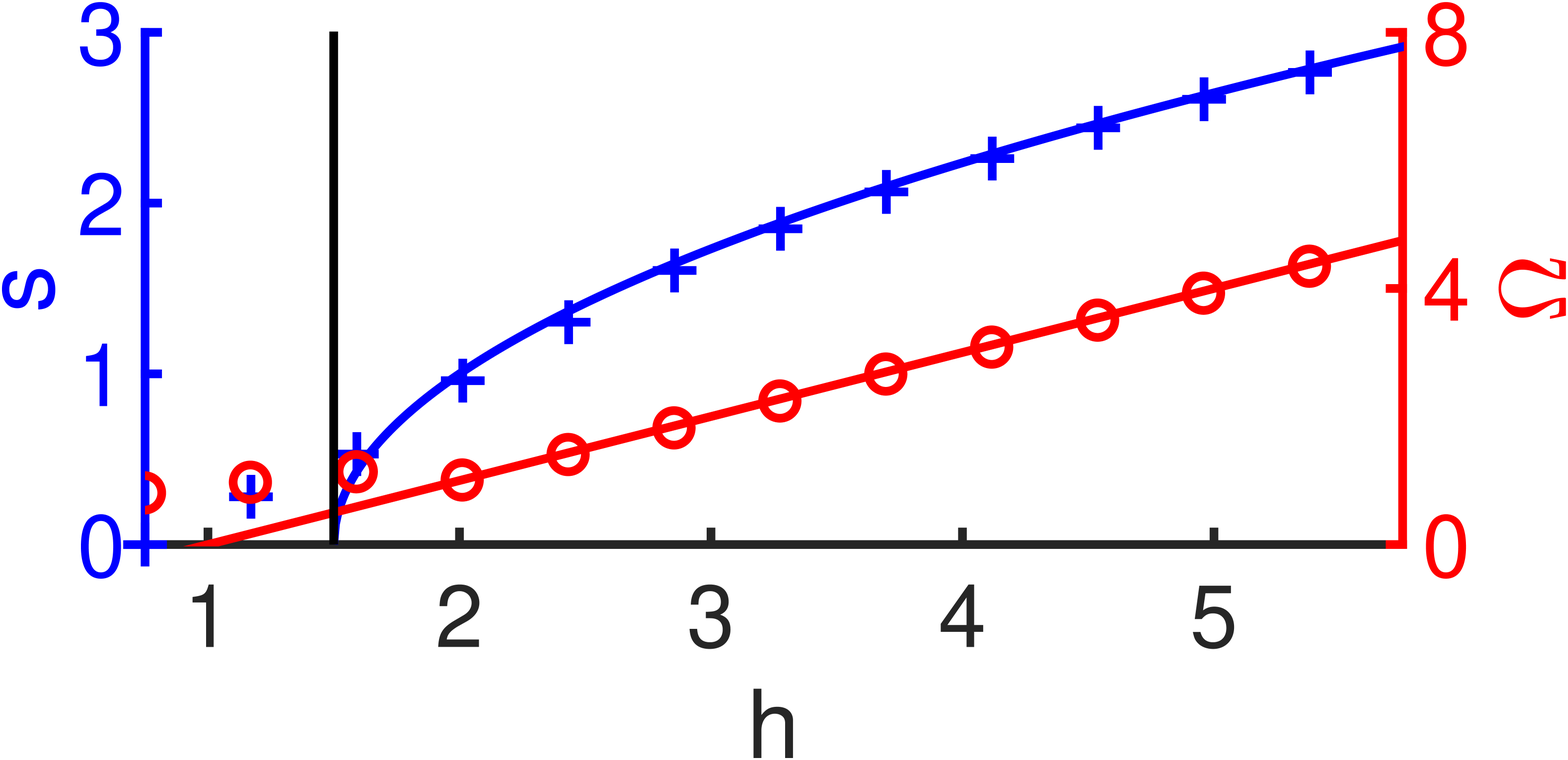}
	\caption*{\small $\cc=-0.5$}
	\end{subfigure}
	\hfill
	\begin{subfigure}[b]{0.49\textwidth}
		\includegraphics[width=\textwidth]{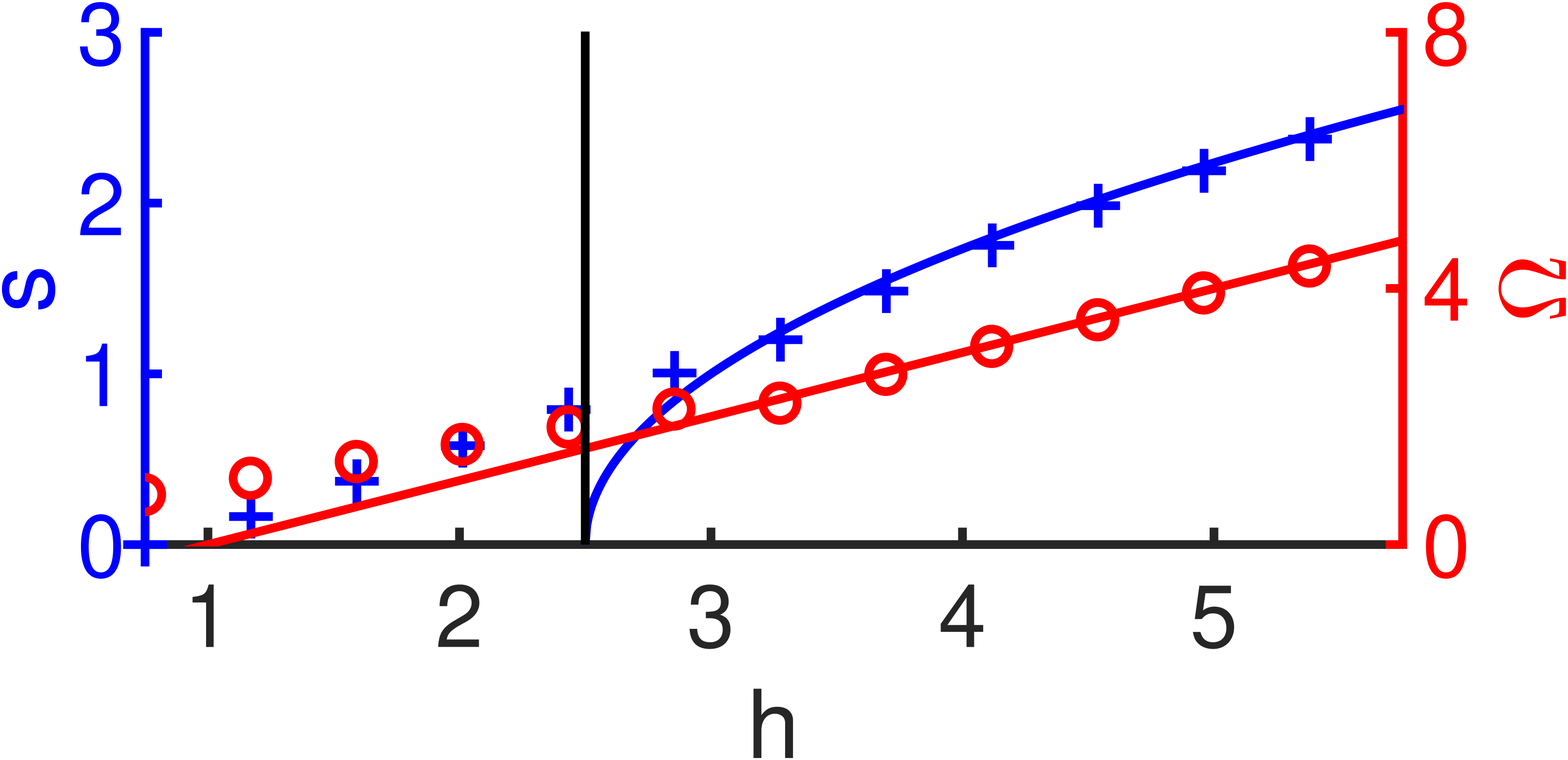}
	\caption*{\small $\cc=0.5$}
	\end{subfigure}
	\caption{Linear spreading speed~\eqref{eq:SelectedSpeedMonostable} (solid blue) and spreading frequency~\eqref{eq:SelectedFrequencyMonostable} (solid red) for $h\in(\beta/\alpha, \beta/\alpha-5\mu)$ in case $\cc=-0.5$ in (a) and $\cc=0.5$ in (b). Other parameters are as in Figure~\ref{fig:StabilityRegionsUp-AndDown-Magnetization}. Numerically observed (asymptotic) speed and frequency is indicated by blue crosses and red circles, respectively. Black vertical line separates the bistable region to the left from the monostable region to the right, cf.\ Figure~\ref{fig:FreezingZeroCCPLight}.}
	\label{fig:SelectedSpeedAndFrequency}
\end{figure}

\section{Discussion and Outlook}\label{sec:DiscussionAndOutlook}
In this paper, we have investigated the dynamics of DWs that connect the up- and down-magnetization states in axially symmetric models of spintronic nanowires. Taking the applied field as the main parameter, we have initially separated the parameter range into the bistable and the monostable region.

For the bistable regime, and without spin-torque effect, we have analytically extended the threshold for stability of an explicitly known family of homogeneous DWs and also into the region where spin-torque effects are present. We have also verified the stability of this family beyond this threshold numerically up to (and beyond) the transition into the monostable regime. Moreover, we have presented a heuristic formula which relates the applied field and polarization ratio that allows to identify parameters for (nearly) standing DWs. We have numerically corroborated that this prediction approximates the direction of propagation well for a broad range of $\cc\neq 0$. We believe that this can be a basis to study weak interactions of multiple stacked DWs in a nanowire that individually would be stationary, but that interact through their tails and thus are expected to form metastable states.

Concerning the monostable regime, we have first presented an explicit formula for the absolute spectrum for a class of complex matrix operators for two-component systems. From this, we were able to relate the most unstable points of the absolute spectrum to the location of the essential singularities of the associated pointwise Green's functions. This allowed us to explicitly determine the linear spreading speed and the associated frequency. We have numerically verified that, for increasing $h$, the explicit DWs are merely convectively unstable until the point of absolute instability of an asymptotic state. Specifically, we find that until this point the explicit family possesses stable point spectrum and also stable essential spectrum in a weighted $L^2$ space. At the same time, this provides numerical evidence that the explicit family yields pushed fronts in this parameter regime, whereas for larger values of $h$ pulled fronts of unknown analytic form are selected. The latter ones are certainly of inhomogeneous DW type and we expect these to be also non-flat, which is the case for zero spin-torque. Moreover, while the explicit family of DWs propagates faster than the linear prediction for applied fields beyond a larger threshold, we find that in this regime the explicit family possesses unstable point spectrum and thus cannot be dynamically selected.

In order to illustrate and corroborate these results, we have presented numerical long-time simulations for applied fields in both regimes; bistable and monostable, using the method of freezing that filters motion in the continuous symmetry groups by selecting the appropriate speed and rotation frequency. We have found that these principles of selecting DWs persist in case of spin-torque effects, i.e., $0<|\cc|<1$.

A natural step towards a deeper understanding of domain wall motion in nanowires is on the one hand to study the selection mechanism of fronts that propagate with the linear predictions at least for the case $\cc=0$ via the energy functional. On the other hand, an Evans function approach could be utilized to analytically verify the stability of the explicit family of homogeneous DWs. A further interesting problem to study, for which one could start from the case $\cc=0$, are the properties and selection principles of modulated fronts which numerically seem to appear close to $h^\Omega$.

Since the main results are also valid for $\cc=0$, they transfer over to the well-known~\eqref{eq:LLG} equation.

\section*{Acknowledgments}
The authors acknowledge support by the Deutsche \linebreak Forschungsgemeinschaft (German Research Foundation), Projektnummer \linebreak 281474342/GRK2224/1. The authors gratefully thank Matt Holzer for stimulating discussions related to \S~\ref{sec:Pointwise growth, absolute spectra, and double roots}, as well as Blake Barker for support on \STABLAB .

\appendix
\section*{Appendix}\label{sec:Appendix}
The matrix $A_\eta(\xi;\lambda)$ in~\eqref{eq:WeightedFirstOrderSystem} reads
\[A_\eta(\xi;\lambda)\coloneqq \begin{pmatrix}
0&1&0&0&0&0\\
-u_{11}&-v_{11}&-u_{12}&-v_{12}&-u_{13}&-v_{13}\\
0&0&0&1&0&0\\
-u_{21}&-v_{21}&-u_{22}&-v_{22}&-u_{23}&-v_{23}\\
0&0&0&0&0&1\\
-u_{31}&-v_{31}&-u_{32}&-v_{32}&-u_{33}&-v_{33}
\end{pmatrix},\]
where
\begin{align*}
u_{11}&=\frac{\alpha^2+\om_1^2}{1+\alpha^2}|\overline{\m}'|^2+\frac{\beta(\alpha^2+\om_1^2)}{\alpha(1+\alpha^2)}\om_3-\frac{(\alpha^2+\om_1^2)(h-\mu\om_3)}{1+\alpha^2}\om_3 - \frac{\alpha^2+\om_1^2}{\alpha}\lambda\\
&\quad-\frac{\om_3\om''_3}{1+\alpha^2}+\Omega\om_3-\frac{\alpha\beta+h-\mu\om_3}{1+\alpha^2}\om_3+s\eta\frac{\alpha^2+\om_1^2}{\alpha}+\eta^2
\end{align*}
\begin{align*}
u_{12}&=-\frac{\alpha^2+\om_1^2}{\alpha(1+\alpha^2)}\om''_3+\Omega\frac{\alpha^2+\om_1^2}{\alpha}-\frac{(\alpha^2+\om_1^2)(h-\mu\om_3+\alpha\beta)}{\alpha(1+\alpha^2)}\\
&\quad -\frac{\alpha}{1+\alpha^2}|\overline{\m}'|^2\om_3+\frac{\alpha(h-\mu\om_3)-\beta}{1+\alpha^2}\om_3-\lambda\om_3+\frac{\om_1\om_3\om_1''}{\alpha(1+\alpha^2)}-s\eta\om_3
\end{align*}
\begin{align*}
u_{13}&=\frac{\beta(\alpha^2+\om_1^2)}{\alpha(1+\alpha^2)}\om_1+\frac{\mu(\alpha^2+\om_1^2)}{1+\alpha^2}\om_1\om_3-\frac{(\alpha^2+\om_1^2)(h-\mu\om_3)}{1+\alpha^2}\om_1\\
&\quad +\frac{\om''_1\om_3}{1+\alpha^2}+\frac{\mu\om_1\om_3}{1+\alpha^2}+\frac{\om_1\om_3}{1+\alpha^2)}|\overline{\m}'|^2-\frac{2\alpha(h-\mu\om_3)-2\beta}{\alpha(1+\alpha^2)}\om_1\om_3\\&\quad+\frac{\mu(\om_3^2-1)}{1+\alpha^2}\om_1\om_3-\frac{\om_1\om_3}{\alpha}\lambda+\frac{s\eta}{\alpha}
\om_1\om_3
\end{align*}
\begin{align*}
u_{21}&=\frac{\alpha}{1+\alpha^2}|\overline{\m}|^2\om_3+\frac{\beta}{1+\alpha^2}\om_3^2-\frac{\alpha(h-\mu\om_3)}{1+\alpha^2}\om_3^2-\lambda\om_3+\frac{\alpha}{1+\alpha^2}\om''_3\\
&\quad-\alpha\Omega+\frac{\alpha^2\beta+\alpha(h-\mu\om_3)}{1+\alpha^2}+s\eta\om_3
\end{align*}
\begin{align*}
u_{22}&=-\frac{\om_3\om''_3}{1+\alpha^2}+\Omega\om_3-\frac{\alpha\beta}{1+\alpha^2}\om_3-\frac{h-\mu\om_3}{1+\alpha^2}\om_3+\frac{\alpha^2}{1+\alpha^2}|\overline{\m}'|^2\\
&\quad+\frac{\alpha\beta}{1+\alpha^2}\om_3-\frac{\alpha^2(h-\mu\om_3)}{\om_3}-\alpha\lambda-\frac{\om_1\om''_1}{1+\alpha^2}+\alpha s\eta+\eta^2
\end{align*}
\begin{align*}
u_{23}&=-\frac{\alpha(h-\mu\om_3)-\beta}{1+\alpha^2}\om_1\om_3+\frac{\alpha\mu}{1+\alpha^2}\om_1\om_3^2-\frac{\alpha}{1+\alpha^2}\om''_1-\frac{\alpha\mu}{1+\alpha^2}\om_1\\
&\quad-\frac{\alpha}{1+\alpha^2}|\overline{\m}'|^2\om_1-\frac{2\beta}{1+\alpha^2}\om_1\om_3+\frac{2\alpha(h-\mu\om_3)}{1+\alpha^2}\om_1\om_3\\
&\quad-\frac{\alpha\mu(\om_3^2-1)}{1+\alpha^2}\om_1+\lambda\om_1-s\eta\om_1
\end{align*}
\begin{align*}
u_{31}&=\frac{\om_1\om_3}{1+\alpha^2}|\overline{\m}'|^2+\frac{\beta-\alpha(h-\mu\om_3)}{1+\alpha^2}\om_1\om_3^2-\frac{\om_1\om_3}{1+\alpha^2}\lambda+\frac{\om_1\om''_3}{1+\alpha^2}\\
&\quad-\Omega\om_1+\frac{h-\mu\om_3+\alpha\beta}{1+\alpha^2}\om_1+\frac{s\eta}{\alpha}\om_1\om_3
\end{align*}
\begin{align*}
u_{32}&=-\frac{\om_1\om_3\om''_3}{\alpha(1+\alpha^2)}+\frac{\om_1\om_3}{\alpha}\Omega-\frac{h-\mu\om_3+\alpha\beta}{\alpha(1+\alpha^2)}\om_1\om_3^2\\
&\quad-\frac{\om_1\om_3}{\alpha}\lambda+\frac{\alpha}{1+\alpha^2}|\overline{\m}'|^2\om_1-\frac{\alpha(h-\mu\om_3)-\beta}{1+\alpha^2}\om_1\om_3\\
&\quad-\lambda\om_1+\frac{\alpha^2+\om_3^2}{\alpha(1+\alpha^2)}\om''_1+s\eta\om_1
\end{align*}
\begin{align*}
u_{33}&=\frac{\beta}{\alpha(1+\alpha^2)}\om_1^2\om_3-\frac{\alpha(h-\mu\om_3)}{\alpha(1+\alpha^2)}\om_1^2\om_3+\frac{\mu}{1+\alpha^2}\om_1^2\om_3^2\\
&\quad -\frac{\om_1\om''_1}{1+\alpha^2}-\frac{\mu}{1+\alpha^2}\om_1^2+\frac{\alpha^2+\om_3^2}{1+\alpha^2}|\overline{\m}'|^2\\
&\quad-\frac{(\alpha^2+\om_3^2)(2\alpha(h-\mu\om_3)-2\beta)}{\alpha(1+\alpha^2)}\om_3\\
&\quad+\frac{\mu(\alpha^2+\om_3^2)(\om_3^2-1)}{1+\alpha^2}-\frac{\alpha^2+\om_3^2}{\alpha}\lambda+\frac{\alpha^2+\om_3^2}{\alpha}s\eta+\eta^2
\end{align*}

\[v_{11}=\frac{\alpha^2+\om_1^2}{\alpha}s+2\eta,\quad v_{12}=-\om_3 s,\quad v_{13}=\frac{\om_1 \om_3}{\alpha}s,\]
\[v_{21}=\om_3 s,\quad v_{22}=\alpha s+2\eta,\quad v_{23}=-\om_1 s\]
\[v_{31}=\frac{\om_1 \om_3}{\alpha}s,\quad v_{32}=\om_1 s,\quad v_{33}=\frac{\alpha^2+\om_3^2}{\alpha}s+2\eta,\]
and where
\[\om_1=\sech(\smu \xi), \quad \om_3=-\tanh(\smu \xi), |\overline{\m}'|^2=-\mu\sech^2(\smu \xi),\]
\[\om_1''=\mu\sech(\smu\xi)(\sech^2(\smu x)-\tanh^2(\smu \xi)),\]
\[\om''_3=-2\mu\tanh(\smu \xi)\sech^2(\smu \xi).\]


\end{document}